\documentclass[11pt,a4paper]{scrartcl}
\usepackage{scrextend}
\usepackage{complexity} 
\usepackage{multirow}
\usepackage{nicefrac}
\usepackage{comment}
\usepackage{amsfonts,amsthm,amsmath}

\usepackage{vmargin}
\setmarginsrb{1in}{1in}{1in}{1in}{0pt}{0pt}{0pt}{6mm}

\usepackage{xcolor}
\usepackage{mdframed}

\newtheorem{theorem}{Theorem}[section]
\newtheorem{lemma}{Lemma}[section]
\newtheorem{claim}{Claim}[section]
\newtheorem{corollary}{Corollary}[section]
\newtheorem{definition}{Definition}[section]
\newtheorem{observation}{Observation}[section]
\newtheorem{proposition}{Proposition}[section]

\usepackage{amssymb,amsmath,mathrsfs}
\usepackage{mathtools}

\usepackage{listings}
\usepackage{mathrsfs}
\usepackage{amsfonts,amsmath}
\usepackage{graphicx,tikz}
\RequirePackage{fancyhdr}
\usepackage{boxedminipage}
\usepackage{todonotes}

\usepackage{graphics}
\usepackage{framed}
\usepackage{enumerate}
\usepackage{float}
\usepackage{xspace}

\newcommand{\splitcontract}{{\textsc{Split Contraction}}\xspace}
\newcommand{\fpt}{{\textsf{FPT}}\xspace}
\newcommand{\denseclique}{\textsc{Densest-$k$-Subgraph}\xspace}
\newcommand{\densecliqueperfect}{\textsc{Densest-$k$-Subgraph with Perfect Completeness}\xspace}
\newcommand{\colordenseclique}{{\textsc{Multicolored Densest-$k$-Subgraph}}\xspace}
\newcommand{\colordensecliqueperfect}{{\textsc{Multicolored Densest-$k$-Subgraph with Perfect Completeness}}\xspace}
\newcommand{\edgegadget}{{\mathtt{ES}}\xspace}
\newcommand{\spec}{{\em special vertices}\xspace}
\newcommand{\specialvert}{{\mathtt{SV}}\xspace}

\newcommand{\tp}{{\em cap vertex}\xspace}

\newcommand{\tps}{{cap vertex}\xspace}
\newcommand{\specs}{{\em special vertex}\xspace}
\newcommand{\specss}{{special vertex}\xspace}
\newcommand{\w}{{W}}

\newcommand{\extra}{\texttt{extra edges}}

\newcommand{\SpC}{\textsc{SpC}}

\newcommand{\chc}{\textsc{Chordal Contraction}}

\newcommand{\calA}{\mathcal{A}}

\newcommand{\ClC}{\textsc{ClC}}
\newcommand{\calF}{\mathcal{F}}

\newcommand{\calO}{\ensuremath{{\mathcal O}}}
\newcommand{\OO}{\mathcal{O}}

\newcommand{\OPT}{\textsc{OPT}}

\newcommand{\calS}{\mathcal{S}}

\newcommand{\calW}{\mathcal{W}}

\newcommand{\WH}{\textsf{W[1]-Hard}}
\newcommand{\WTH}{\textsf{W[2]-Hard}}

\newcommand{\CONPpoly}{\textsf{coNP/poly}}

\newcommand{\yes}{{yes}}
\newcommand{\no}{{no}}

\newtheorem{reduction rule}{Reduction Rule}[section]
\newtheorem{marking-scheme}{Marking Scheme}[section]
\newtheorem{remark1}{Remark}[section]
\newtheorem{condition1}{Condition}[section]

\newcounter{condition}[section]
\newenvironment{condition}[1][]{\refstepcounter{condition}\par\medskip \noindent \textit{Condition~\thecondition}:~[#1]\rmfamily}{}

\makeatletter
\newtheorem*{rep@theorem}{\rep@title}
\newcommand{\newreptheorem}[2]{%
\newenvironment{rep#1}[1]{%
 \def\rep@title{#2 \ref{##1}}%
 \begin{rep@theorem}}%
 {\end{rep@theorem}}}
\makeatother

\newreptheorem{theorem}{Theorem}

\usepackage{booktabs}
\usepackage[pdftex, plainpages = false, pdfpagelabels, 
                 bookmarks=false,
                 bookmarksopen = true,
                 bookmarksnumbered = true,
                 breaklinks = true,
                 linktocpage,
                 pagebackref,
                 colorlinks = true,  
                 linkcolor = blue,
                 urlcolor  = blue,
                 citecolor = red,
                 anchorcolor = green,
                 hyperindex = true,
                 hyperfigures
                 ]{hyperref}

\newcommand{\defparproblem}[4]{
\begin{tcolorbox}[colback=green!5!white,colframe=green!75!black]
  \begin{tabular*}{\textwidth}{@{\extracolsep{\fill}}lr} #1  & {\bf{Parameter:}} #3 \\ \end{tabular*}
  {\bf{Input:}} #2  \\
  {\bf{Question:}} #4
\end{tcolorbox}
}

\usepackage{wrapfig}
\usepackage{tcolorbox}

\usepackage{amsthm}
\usepackage{amsmath}
\usepackage{amsfonts}
\usepackage{thmtools,thm-restate}
\usepackage{comment}
\makeatletter
\DeclareOldFontCommand{\rm}{\normalfont\rmfamily}{\mathrm}
\DeclareOldFontCommand{\sf}{\normalfont\sffamily}{\mathsf}
\DeclareOldFontCommand{\tt}{\normalfont\ttfamily}{\mathtt}
\DeclareOldFontCommand{\bf}{\normalfont\bfseries}{\mathbf}
\DeclareOldFontCommand{\it}{\normalfont\itshape}{\mathit}
\DeclareOldFontCommand{\sl}{\normalfont\slshape}{\@nomath\sl}
\DeclareOldFontCommand{\sc}{\normalfont\scshape}{\@nomath\sc}
\makeatother

\usepackage{enumitem,kantlipsum}
\usepackage{lmodern}
\usepackage{flushend}
\usepackage[T1]{fontenc}
\usepackage[utf8]{inputenc}

\date{}
\begin{document}
\thispagestyle{empty}
\title{On the Parameterized Approximability of Contraction to Classes of Chordal Graphs}

\author{Spoorthy Gunda\thanks{Simon Fraser University, Burnaby, Canada. \texttt{sgunda@sfu.ca}} 
\and 
Pallavi Jain\thanks{Indian Institute of Technology Jodhpur, Jodhpur, India. \texttt{pallavi@iitj.ac.in}}
\and 
Daniel Lokshtanov\thanks{University of California, Santa Barbara, USA. \texttt{daniello@ucsb.edu}}
\and
Saket Saurabh\thanks{The Institute of Mathematical Sciences, HBNI, Chennai, India, and University of Bergen, Norway. \texttt{saket@imsc.res.in}}
\and
Prafullkumar Tale\thanks{Max Planck Institute for Informatics, Saarland Informatics Campus, Saarbr\"ucken, Germany. \texttt{prafullkumar.tale@mpi-inf.mpg.de}}
}

\maketitle

\let\thefootnote\relax\footnotetext{
\begin{minipage}{0.7\textwidth}
\textbf{Funding:} \emph{Saket Saurabh:} 
This project has received funding from the European Research Council (ERC) under the European Union's Horizon $2020$ research and innovation programme (grant agreement No $819416$), and Swarnajayanti Fellowship (No DST/SJF/MSA01/2017-18).
\end{minipage}
\begin{minipage}{0.25\textwidth}
  \centering
    \includegraphics[scale=0.75]{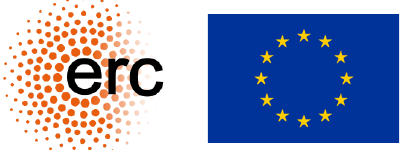}
\end{minipage}
\begin{minipage}{\textwidth}
\noindent \emph{Prafullkumar Tale:} This research is a part of a project that has received funding from the European Research Council (ERC) under the European Union's Horizon $2020$ research and innovation programme under grant agreement SYSTEMATICGRAPH (No. $725978$). Most part of this project was completed when the author was a Senior Research Fellow at The Institute of Mathematical Sciences, HBNI, Chennai, India.
\end{minipage}
}

\begin{abstract}
A graph operation that {\em contracts edges} is one of the fundamental operations in the theory of graph
minors. Parameterized Complexity of editing to a family of graphs by contracting $k$ edges has recently gained substantial scientific attention, and several new results have been obtained.  
Some important families of graphs, namely the subfamilies of
chordal graphs, in the context of edge contractions, have proven to be significantly difficult 
than one might expect. In this paper, we study the  \textsc{$\cal F$-Contraction} problem, where $\cal F$ is a subfamily of chordal graphs, in the realm of parameterized approximation. Formally, given a graph $G$ and an integer $k$, \textsc{ $\cal F$-Contraction} asks whether there exists $X \subseteq E(G)$ such that $G/X \in \cal F$ and $|X| \leq k$. Here, $G/X$ is the graph obtained from $G$ by contracting edges in $X$. We obtain the following results for the \textsc{ $\cal F$-Contraction} problem. 
\begin{itemize}
\item \textsc{Clique Contraction} is known to be \FPT. However, unless $\NP \subseteq \CONPpoly$, it does not admit a polynomial kernel. We show that it admits a polynomial-size approximate kernelization scheme (\textsf{PSAKS}). That is, it admits a $(1 + \epsilon)$-approximate kernel with $\calO(k^{f(\epsilon) })$ vertices for every $\epsilon >0$.

\item \textsc{Split Contraction} is known to be \WH. We deconstruct this intractability result in two ways. Firstly, we give a $(2+\epsilon)$-approximate polynomial kernel for \textsc{Split Contraction} (which also implies a factor $(2+\epsilon)$-\FPT-approximation algorithm for \textsc{ Split Contraction}). Furthermore, we show that, assuming \textsf{ Gap-ETH}, there is no $\left(\frac{5}{4}-\delta \right)$-\FPT-approximation algorithm for \textsc{Split Contraction}. Here, $\epsilon, \delta>0$ are fixed constants. 

\item \textsc{Chordal Contraction} is known to be \WTH. We complement this result by observing that the existing \textsf{W[2]-hardness} reduction can be adapted to show that, assuming \FPT $\neq$ \textsf{W[1]}, there is no $F(k)$-\FPT-approximation algorithm for \textsc{Chordal Contraction}. Here, $F(k)$ is an arbitrary function depending on $k$ alone.
\end{itemize}
We say that an algorithm is an $h(k)$-\FPT-approximation algorithm for the 
\textsc{$\cal F$-Contraction} problem, if it runs in \FPT\ time, and on any input $(G, k)$ such that there exists $X \subseteq E(G)$ satisfying $G/X \in \cal F$ and $|X| \leq k$, it outputs an edge set $Y$ of size at most $h(k) \cdot k$ for which $G/Y$ is in $\cal F$. We find it extremely interesting that three closely related problems have different behavior with respect to \FPT-approximation.

\end{abstract}

\pagestyle{plain}
\setcounter{page}{1}
\section{Introduction}
\label{sec:intro}
Graph modification problems have been extensively studied since the inception of Parameterized Complexity in the early `90s. The input of a typical graph modification problem consists of a graph $G$ and a positive integer $k$, and the objective is to edit $k$ vertices (or edges) so that the resulting graph belongs to some particular family, $\cal F$, of graphs. These problems are not only mathematically and structurally challenging, but have also led to the discovery of several important techniques in the field of Parameterized Complexity. It would be completely appropriate to say that solutions to these problems played a central role in the growth of the field. In fact, just in the last few years, parameterized algorithms have been developed for several graph editing problems~\cite{caoM14,cao15ICALP,cao2015,cao16SODA,ProperInterval14,IntervalSODA16,fomin2013,
fomin2014,thresholdESA,drangeSTACS14,drangeESA15,splitGraphs}.  
The focus of all of these papers and  the vast majority of papers on parameterized graph editing problems 
has so far been limited to edit operations that delete vertices, delete edges or add edges.

In recent years, a different edit operation has begun to attract significant scientific attention. This operation, which is arguably the most natural edit operation apart from deletions/insertions of vertices/edges, is the one that contracts an edge. Here, given an edge $uv$  
that exists in the input graph, we remove the edge from the graph and merge its two endpoints. Edge contraction is a fundamental operation in the theory of graph minors. For some particular family of graphs, $\cal F$, we say that a graph $G$ belongs to ${\cal F}+kv$, ${\cal F}+ke$ or ${\cal F}-ke$ if some graph in ${\cal F}$ can be obtained by deleting at most $k$ vertices from $G$, deleting at most $k$ edges from $G$ or adding at most $k$ edges to $G$, respectively. Using this terminology, we say that a graph $G$ belongs to ${\cal F}|ke$ if some graph in ${\cal F}$ can be obtained by contracting at most $k$ edges in $G$.  In this paper, we study the following problem.

\smallskip

\defparproblem{\textsc{$\mathcal{F}$-Contraction}}{A graph $G$ and an integer $k$}{$k$}
{Does $G$ belong to $\mathcal{F}|ke$?}
\smallskip

\noindent 
For several families of graphs $\cal F$, early papers by Watanabe et al.~\cite{contractEarly2,contractEarly3}, and Asano and Hirata~\cite{contractEarly1} showed that \textsc{ ${\cal F}$-Edge Contraction} is \NP-complete.

In the framework of Parameterized Complexity, these problems exhibit properties that are quite different from those problems where we only delete or add vertices and edges. Indeed, a well-known result by Cai \cite{cai1996} states that in case $\cal F$ is a hereditary family of graphs with a finite set of forbidden induced subgraphs, then the graph modification problems, ${\cal F}+kv$, ${\cal F}+ke$ or ${\cal F}-ke$, defined by $\cal F$ admits a simple \FPT\ algorithm (an algorithm with running time $f(k)n^{\OO(1)}$). However, for {\sc ${\cal F}$-Contraction}, the result by Cai \cite{cai1996} does not hold. In particular, Lokshtanov et al.~\cite{elimNew} and Cai and Guo~\cite{CliqueContractUpp} independently showed that if $\cal F$ is either the family of $P_\ell$-free graphs for some $\ell\geq 5$ or the family of $C_\ell$-free graphs for some $\ell\geq 4$, then {\sc ${\cal F}$-Contraction} is \WTH\ ({\sf W[i]}-{\sf hardness}, for $i\geq 1$, is an analogue to \NP-{\sf hardness} in Parameterized Complexity, and is used to rule out \FPT-algorithm for the problem) when parameterized by $k$ (the number of edges to be contracted). These results immediately imply that {\sc Chordal Contraction} is \WTH\ when parameterized by $k$.  
The parameterized hardness result for {\sc Chordal Contraction} led to finding subfamilies of chordal graphs, where the problem could be shown to be \FPT. Two subfamilies that have been considered in the literature are families of 
{\em split graphs} and {\em cliques}. Cai and Guo~\cite{CliqueContractUpp} showed that {\sc Clique Contraction} is 
\FPT, however, it does not admit a polynomial kernel. Later, Cai and Guo \cite{SplitContractErr} also claimed to design an algorithm that solves {\sc Split Contraction} in time $2^{\OO(k^2)}\cdot n^{\OO(1)}$, which proves that the problem is \FPT. However, Agrawal et al.~\cite{DBLP:conf/stacs/AgrawalLSZ17} found an error with the proof and showed that {\sc Split Contraction}  is  \WH.

%
  
  \begin{tcolorbox}[colback=red!5!white,colframe=red!75!black]
 Inspired by the intractable results that {\sc Chordal Contraction}, {\sc Split Contraction} and {\sc Clique Contraction} are \WTH, \WH, and does not admit polynomial kernel, respectively, we study them from the viewpoint of parameterized approximation. 
\end{tcolorbox}

\vspace{-0.35cm}
\paragraph{Our Results and Methods.} We start by defining a few basic definitions in parameterized approximation. To formally define these, we need a notion of parameterized optimization problems. 
We defer formal definitions to Section~\ref{sec:prelim} and give intuitive definitions here. 
We say that an algorithm is an $h(k)$-\FPT-approximation algorithm for the {\sc $\cal F$-Contraction} problem, if it runs in \FPT\ time, and on any input $(G, k)$ if there exists $X \subseteq E(G)$ such that $G/X \in \cal F$ and $|X| \leq k$, it outputs an edge set $Y$ of size at most $h(k)\cdot k$ and $G/Y \in \cal F$. Let $\alpha \geq 1$ be a real number. We now give an informal definition of $\alpha$-approximate kernels. The kernelization algorithm takes an instance $I$ with parameter $k$, runs in polynomial time, and produces a new instance $I'$ with parameter $k'$. Both $k'$ and the size of $I'$ should be bounded in terms of just the parameter $k$. That is, there exists a function $g(k)$ such that $|I'| \leq g(k)$ and $k' \leq g(k)$. This function $g(k)$ is called the {\em size} of the kernel. For minimization problems, we also require the following from $\alpha$-approximate kernels: For every $c \geq 1$, a $c$-approximate solution $S'$ to $I'$ can be transformed in polynomial time into a $(c \cdot \alpha)$-approximate solution $S$ to $I$. However, if the quality of $S'$ is ``worse than'' $k'$, or $(c \cdot \alpha) \cdot OPT(I) > k$, the algorithm that transforms $S'$ into $S$ is allowed to fail. Here, $OPT(I)$ is the value of the optimum solution of the instance $I$.

Our first result is about {\sc Clique Contraction}. It is known to be \FPT. However, unless $\NP \subseteq \CONPpoly$, it does not admit a polynomial kernel~\cite{CliqueContractUpp}. We show that it admits a {\sf PSAKS}. That is, it admits a $(1 + \epsilon)$-approximate polynomial kernel with $\calO(k^{f(\epsilon) })$ vertices for every $\epsilon >0$. In particular, we obtain the following result.

 \begin{theorem}\label{thm:clique-lossy} For any $\epsilon > 0$, \textsc{Clique Contraction} parameterized by the size of solution $k$, admits a time efficient $(1 + \epsilon)$-approximate polynomial kernel with $\calO(k^{d+1})$ vertices,  where $d = \lceil \frac{1 }{\epsilon} \rceil$.
\end{theorem}

\vspace{-0.65cm}
\paragraph{Overview of the proof of Theorem~\ref{thm:clique-lossy}.} 
Let us fix an input $(G,k)$ and a constant $\epsilon > 0$. 
Given a graph $G$, contracting edges of $G$ to get into a graph class $\cal F$ 
 is same as partitioning the vertex set $V(G)$ into connected sets, $W_1, W_2, \ldots ,W_\ell$, and then contracting each connected set to a vertex. These connected sets are called {\em witness sets}. A witness set $W_i$ is called {\em non-trivial}, if $|W_i|\geq 2$, and {\em trivial} otherwise.  

Observe that if a graph $G$ can be transformed into a clique by contracting edges in $F$, then $G$ can also be converted into a clique by deleting all the endpoints of edges in $F$.   
This observation implies that if $G$ is \emph{$k$-contractible} to a clique, then there exists an induced clique of size at least $|V(G)| - 2k$. Let $I$ be a set of vertices in $G$, which induces this large clique and let $C = V(G) \setminus I$. Observe that $C$ forms a vertex cover in the graph $\overline{G}$ (graph with vertex set $V(G)$ and those edges that are not present in $E(G)$). Using a factor $2$-approximation algorithm, we find a vertex cover $X$ of $\overline{G}$. Let $Y=V(\overline{G})-X$ be an independent set in $\overline{G}$. 
If $|X|>4k$, we immediately say No. Now, suppose that we have some solution and let 
 $W_1, W_2, \ldots ,W_{\ell}$ be those witness sets that are either non-trivial or contained in $X$. Now, let us say that a set $W_i$ is {\em nice} if it has at least one vertex outside $X$, and {\em small} if it contains less than 
 $\OO(1/\epsilon)$ vertices. A set that is not {\em small} is {\em large}. {Observe that there exists a 
 $(1+\epsilon)$-approximate solution where {\em the only sets} that are not nice are small.} Also, observe that 
 all nice sets are adjacent. 
 Now, we classify all subsets of $X$ of size at most $\OO(1/\epsilon)$ as {\em possible} and {\em impossible} small witness sets. Notice that if a set $A \subseteq X$ has more than $2k$ non-neighbors, then it can not possibly be a witness set, as one of these non-neighbors will be a trivial witness set. Now for every set, $A \subseteq X$ of size at most $\OO(1/\epsilon)$ mark all of its non-neighbors, but if there are more than $2k$, then mark $2k+1$ of them. Now, look at an unmarked vertex in $Y$, the only reason it could still be relevant if it is part of some $W_i$. So its job is $(a)$ connecting the vertices in $W_i$, or $(b)$ potentially being the vertex in $Y$ that is 
 making some $W_i$ nice, or $(c)$ it is a neighbor to {\em all} the small (not nice) subsets of $X$ in the solution.  
Now notice that any vertex in $Y$ that is unmarked does jobs $(b)$ and $(c)$ equally well. So we only need to care about connectivity. Look at some nice and small set $W_i$; we only need to preserve the neighborhoods of the vertices of 
$Y$ into $W_i$.  
For every subset of size $\OO(1/\epsilon)$, we keep one vertex in $Y$ that has that set in its neighborhood. Notice that we do not care that different $W_i$'s use different marked vertices for connectivity because merging two $W_i$'s is more profitable for us. Finally, we delete all unmarked vertices and obtain an $(1+\epsilon)$-approximate kernel of size roughly $k^{\OO(1/\epsilon)}$.
We argue that this kernelization algorithm is time efficient i.e. the running time is polynomial in the size of an input and the constant in the exponent is independent of $\epsilon$.
This completes the overview of the proof for Theorem~\ref{thm:clique-lossy}.
Next, we move to {\sc Split Contraction}.

{\sc Split Contraction} is known to be \WH~\cite{DBLP:conf/stacs/AgrawalLSZ17}.  We ask ourselves whether {\sc Split Contraction} is {\em completely \FPT-inapproximable} or admits an $\alpha$-\FPT-approximation algorithm, for some fixed constant $\alpha>0$.  
We obtain two results towards our goal.

\begin{theorem}\label{thm:split-lossy} 
For every $\epsilon > 0$, \textsc{Split Contraction} admits a factor $(2+\epsilon)$-\FPT-approximation algorithm. In fact, for any $\epsilon > 0$, \textsc{Split Contraction} admits a $(2 + \epsilon)$-approximate kernel with $\calO(k^{f(\epsilon)})$ vertices.
\end{theorem}

Given, Theorem~\ref{thm:split-lossy}, it is natural to ask whether \textsc{Split Contraction} admits a factor $(1+\epsilon)$-\FPT-approximation algorithm, for every $\epsilon > 0$. We show that this is not true and obtain the following hardness result. 

\begin{theorem} \label{thm:split-no-lossy} 
 Assuming {\sf Gap-ETH}, no \FPT\ time algorithm can approximate  {\sc Split Contraction} within a 
 factor of $\left(\frac{5}{4}-\delta \right)$, for any fixed constant $\delta>0$. 
\end{theorem}
 
 \vspace{-0.65cm}
 \paragraph{Overview of the proofs of Theorems~\ref{thm:split-lossy} and~\ref{thm:split-no-lossy}.}  
Our proof for Theorem~\ref{thm:split-lossy} uses ideas for $(1+\epsilon)$-approximate kernel for {\sc Clique Contraction} (Theorem~\ref{thm:clique-lossy}) and thus we omit its overview. 
Towards the proof of Theorem~\ref{thm:split-no-lossy}, we give a gap preserving reduction from a variant of the {\sc Densest-$k$-Subgraph} problem (given a graph $G$ and an integer $k$, find a subset $S \subseteq V(G)$ of $k$ vertices that induces maximum number of edges).
Chalermsook et al.~\cite{DBLP:conf/focs/ChalermsookCKLM17} showed that, assuming {\sf Gap-ETH}\footnote{We refer the readers to~\cite{DBLP:conf/focs/ChalermsookCKLM17} for the definition of \textsf{Gap-ETH} and related terms.}, for any $g=o(1)$, there is no \FPT-time algorithm that, given an integer $k$ and any graph $G$ on $n$ vertices that contains at least one $k$-clique, 
 always output $S\subseteq V(G)$, of size $k$, such that ${\sf Den}(S)\geq k^{-g(k)}$. Here, ${\sf Den}(S)=|E(G[S])|/{|S|\choose 2}$. We need a strengthening of this result that says that assuming {\sf Gap-ETH}, for any $g=o(1)$ and for any constant $\alpha >1$, there is no \FPT-time algorithm that, given an integer $k$ and any graph $G$ on $n$ vertices that contains at least one $k$-clique, always outputs $S\subseteq V(G)$, of size $\alpha k$, such that ${\sf Den}(S)\geq k^{-g(k)}$. Starting from this result, we give a gap-preserving reduction to {\sc Split Contraction} that takes \FPT\ time and obtain Theorem~\ref{thm:split-no-lossy}. Given an instance $(G,k)$ of {\sc Densest-$k$-Subgraph}, we first use color coding to partition the edges into $t = \binom{k}{2}$ color classes such that every color class contains exactly one edge of a ``densest subgraph'' (or a clique). For each color class we make one {\em edge selection gadget}. Each edge selection gadget corresponding to the color class $j$ consists of an independent set $\texttt{ES}_j$ that contains a vertex corresponding to each edge in the color class$j,$ and a \emph{cap vertex} $g_j$ that is adjacent to every vertex in $\texttt{ES}_j$. 
 Next, we add a sufficiently large clique $Z$ of size $\rho \cdot |V(G)|$, where for every vertex $v\in V(G)$, we have $\rho$ vertices. Every vertex in an edge selection gadget is adjacent to every vertex of $Z$, except those corresponding to the endpoints of the edge the vertex represents. Finally, we add a clique $\texttt{SV}$ of size $t$ that has one vertex $s_j$ for each edge selection gadget. Make the vertex $s_j$ adjacent to every vertex in $\texttt{ES}_j$. We also add sufficient guards on vertices everywhere, so that ``unwanted'' contractions do not happen. The idea of the reduction is to contract edges in a way that the vertices in 
$\texttt{SV}$, $Z$, and $g_j$, $j\in \{1,\ldots,t\}$,  become a giant clique and other vertices become part of an independent set, resulting in a split graph. Towards this we first use $2t$ contractions so that $g_j$, $s_j$, and a vertex $a_j \in \texttt{ES}_j$ are contracted into one.
One way to ensure that they form a clique along with $Z$ is to contract each of them to a vertex in $Z$.
However, this will again require $t$ edge contractions. 
We set our budget in a way that this is not possible. Thus, what we need is to destroy the non-neighbors of $a_j$.  One way to do this again will be to match the vertices obtained after the first round of $2t$ contractions in a way that there are no non-adjacencies left. However, this will also cost $t/2$, and our budget does not allow this. The other option (which we take) is to take the union of all non-neighbors of $a_j$, say $N$, and contract each of them to one of the vertex in $Z\setminus N$. Observe that to minimize the contractions to get rid of non-neighbors of $a_j$, we would like to minimize $|N|$. This will happen when $N$ spans a large number of edges. Thus, it precisely captures the {\sc Densest-$k$-Subgraph} problem. The budget is chosen in a way that we get the desired gap-preserving reduction, which enables us to prove Theorem~\ref{thm:split-no-lossy}.

Our final result concerns {\sc Chordal Contraction}. Lokshtanov et al.~\cite{elimNew} showed that {\sc Chordal Contraction} is \WTH. We observe that the existing {\sf W[2]-hardness} reduction can be adapted to show the following theorem. 
 
 \begin{theorem} \label{thm:chordal-no-lossy} 
 Assuming \FPT $\neq$ {\sf W[1]}, no \FPT\ time algorithm can approximate  {\sc Chordal Contraction} within a factor of $F(k)$. Here, $F(k)$ is a function depending on $k$ alone. 
\end{theorem}
 
 \vspace{-0.65cm}
\paragraph{Overview of the proof of Theorem~\ref{thm:chordal-no-lossy}.} 
Towards proving Theorem~\ref{thm:chordal-no-lossy}, we give a $1$-approximate polynomial parameter transformation (1-appt)
from {\sc Set Cover} (given a universe $U$, a family of subsets $\cal S$, and an integer $k$, 
we shall decide the existence of a subfamily of size $k$ that contains all the elements of $U$) to {\sc Chordal Contraction}. That is, given any solution of size at most 
$\ell$ for {\sc Chordal Contraction}, we can transform this into a solution for {\sc Set Cover} of size at most $\ell$. Karthik et al.~\cite{DBLP:conf/stoc/SLM18} showed that assuming \FPT $\neq$ {\sf W[1]}, no \FPT\ time algorithm can approximate  {\sc Set Cover} within a factor of $F(k)$. Pipelining this result with our reduction we get Theorem~\ref{thm:chordal-no-lossy}. 

\vspace{-0.35cm}
\paragraph{Related Work.} 
To the best of our knowledge, Heggernes et al.~\cite{treePathContract} was the first to explicitly study {\sc ${\cal F}$-Contraction} from the viewpoint of Parameterized Complexity. They showed that in  case $\cal F$ is the family of trees, {\sc ${\cal F}$-Contraction} is \FPT\ but does not admit a polynomial kernel, while in case $\cal F$ is the family of paths, the corresponding problem admits a faster algorithm and an $\OO(k)$-vertex kernel. Golovach et al.~\cite{planarContract} proved that if $\cal F$ is the family of planar graphs, then {\sc ${\cal F}$-Contraction} is again \FPT. Moreover, Cai and Guo~\cite{CliqueContractUpp} showed that in case $\cal F$ is the family of cliques, {\sc ${\cal F}$-Contraction} is solvable in time $2^{\OO(k\log k)}\cdot n^{\OO(1)}$, while in case $\cal F$ is the family of chordal graphs, the problem is \WTH. 
Heggernes et al.~\cite{bipartiteContract} developed an \FPT\ algorithm for the case where $\cal F$ is the family of bipartite graphs. Later, a faster algorithm was proposed by Guillemot and Marx~\cite{bipartiteContract2}.

Pioneering work of Lokshtanov et al.~\cite{lossy} on the approximate kernel is being followed by a series of papers generalizing/improving results mentioned in this work and establishing lossy kernels for various other problems.
Lossy kernels for some variations of \textsc{Connected Vertex Cover}~\cite{eiben2017lossy,krithika2018revisiting}, {\sc Connected Feedback Vertex Set}~\cite{DBLP:conf/esa/Ramanujan19}, \textsc{Steiner Tree}~\cite{dvorak2018parameterized} and \textsc{Dominating Set}~\cite{eiben2018lossy, siebertz2017lossy} have been established (also see~\cite{manurangsi2018note, van20181+}).    
Krithika et al.~\cite{lossy-fst} were first to study graph contraction problems from the lenses of lossy kernelization. They proved that for any $\alpha > 1$, \textsc{Tree Contraction} admits an $\alpha$-lossy kernel with $\mathcal{O}(k^d)$ vertices, where $d = \lceil \alpha/(\alpha - 1) \rceil$. Agarwal et al.~\cite{agarwal2017parameterized} proved similar result for \textsc{$\mathcal{F}$-Contraction} problems where graph class $\mathcal{F}$ is defined in parametric way from set of trees. Eiben et al.~\cite{eiben2017lossy} obtained similar result for {\sc Connected $\cal H$-Hitting Set} problem. 

\vspace{-0.35cm}
\paragraph{Guide to the paper.} We start by giving the notations and preliminaries that we use throughout the paper in Section~\ref{sec:prelim}. This section is best used as a reference, rather than being read linearly. 
In Section~\ref{sec:lossy_clique} we give the $(1+\epsilon)$--approximate polynomial kernel for \textsc{Clique Contraction}. Section~\ref{sec:lossy_split} gives the $(2+\epsilon)$--approximate polynomial kernel for \textsc{Split Contraction}. The ideas here are  similar to those used in Section~\ref{sec:lossy_clique}, and thus an  eager reader  could skip further. 
In Section~\ref{hardness_splt},  we show that assuming {\sf Gap-ETH}, no \FPT\ time algorithm can approximate  {\sc Split Contraction} within a 
 factor of $\left(\frac{5}{4}-\delta \right)$, for any fixed constant $\delta>0$. Section~\ref{hardness_chordal} shows that,  
  assuming \FPT $\neq$ {\sf W[1]}, no \FPT\ time algorithm can approximate  {\sc Chordal Contraction} within a factor of $F(k)$. This is an adaptation of the existing {\sf W[2]-hardness} reduction  and may be skipped. Thus, our main technical results appear in Sections~\ref{sec:lossy_clique} and \ref{hardness_splt}. We conclude the paper with some interesting open problems in Section~\ref{section_conclusion}.

\section{Preliminaries}
\label{sec:prelim}

In this section, we give notations and definitions that we use throughout the paper. Unless specified, we will be using all general graph terminologies from the book of Diestel~\cite{diestel-book}. 

\subsection{Graph  Theoretic Definitions and Notations}
For an undirected graph $G$, sets $V(G)$ and $E(G)$ denote the set of vertices and edges, respectively. Two vertices $u, v$ in $V(G)$ are said to be \emph{adjacent} if there is an edge $uv$ in $E(G)$. The neighborhood of a vertex $v$, denoted by $N_G(v)$, is the set of vertices adjacent to $v$ in $G$. For subset $S$ of vertices, we define $N(S) =\bigcup_{v \in S} N(v)) \setminus S$. 
The subscript in the notation for the neighborhood is omitted if the graph under consideration is clear. For a set of edges $F$, set $V(F)$ denotes the endpoints of edges in $F$. For a subset $S$ of $V(G)$, we denote the graph obtained by deleting $S$ from $G$ by $G - S$ and the subgraph of $G$ induced on set $S$ by $G[S]$. For two subsets $S_1, S_2$ of $V(G)$, we say $S_1, S_2$ {\em are adjacent if there exists an edge} with one endpoint in $S_1$ and other in $S_2$. 

An edge $e$ in $G$ is a \emph{chord} of a cycle $C$ (resp. path $P$) if (i) both the endpoints of $e$ are in $C$ (resp. in $P$), and (ii) edge $e$ is not in $C$ (resp. not in $P$). An \emph{induced cycle} (resp. \emph{path}) is a cycle (resp. path) which has no chord. We denote induced cycle and path on $\ell$ vertices by $C_{\ell}$ and $P_{\ell}$, respectively.  
A \emph{complete graph} $G$ is an undirected graph in which for every pair of vertices $u,v\in V(G)$, there is an edge $uv$ in $E(G)$. As an immediate consequence of definition we get the following. 

\begin{lemma}
\label{lem:complete-characterization}
A connected graph $G$ is  complete if and only if $G$ does not contain an induced $P_3$.
\end{lemma}

A \emph{clique} is a subset of vertices in the graph that induces a complete graph. A set $I \subseteq V(G)$ of pairwise non-adjacent vertices is called an {\em independent set}.
A graph $G$ is a \emph{split graph} if $V(G)$ can be partitioned into a clique and an independent set. For split graph $G$, partition $(X, Y)$ is \emph{split partition} if $X$ is a clique and $Y$ is an independent set. In this article, whenever we mention a split partition, we first mention the clique followed by the independent set. We will also use the following well-known characterization of split graphs. Let, $2K_2$ be a graph induced on four vertices, which contains exactly two edges and no isolated vertices. 

\begin{lemma}[\cite{golumbic2004algorithmic}]
\label{lem:split-characterization}
A graph $G$ is a split graph if and only if it does not contain $C_4, C_5$ or $2K_2$ as an induced subgraph.  
\end{lemma}

A graph \(G\) is \emph{chordal} if every induced cycle in \(G\) is a triangle; equivalently, if every cycle of length at least four has a chord.  
A vertex subset $S \subseteq V(G)$ is said to \emph{cover} an edge $uv \in E(G)$ if $S \cap \{u,v\} \neq \emptyset$. A vertex subset $S \subseteq V(G)$ is called a \emph{vertex cover} in $G$ if it covers all the edges in $G$.

We start with the following observation, which is useful to find a large induced clique in the input graph.
The \emph{complement} of $G$, denoted by $\bar{G}$, is a graph whose vertex set is $V(G)$ and edge set is precisely those edges which are not present in $E(G)$.
Note that given a graph $G$, if $S$ is a set of vertices such that $G-S$ is a clique, then $S$ is a vertex cover in the complement graphs of $G$, denoted by $\bar{G}$, as $\bar{G}-S$ is edgeless. Using the well-known factor $2$-approximation algorithm for {\sc Vertex Cover}~\cite{bar1981linear}, we have following.

\begin{observation}[\cite{bar1981linear}]
 \label{obs:factor-2-approx}
 There is a factor $2$-approximation algorithm to compute a set of vertices whose deletion results in a complete graph.
\end{observation}

Using, Lemma~\ref{lem:split-characterization} one can  obtain a simple factor $5$-approximation algorithm for deleting vertices to get a split graph. 

 \begin{observation}\label{obsn:approx-svd-solution} There is a factor $5$-approximation algorithm to compute a set of vertices whose deletion results in a split graph.
 \end{observation}
Recently, for every $\epsilon >0$, a factor $(2+\epsilon)$-approximation algorithm for deleting vertices to get a split graph has been obtained~\cite{LokshtanovMPPS20}. However, for our purposes Observation~\ref{obsn:approx-svd-solution} will suffice. 

\subsection{Graph Contraction}
The {\em contraction} of edge $e = uv$ in $G$ deletes vertices $u$ and $v$ from $G$, and adds a new vertex, which is made adjacent to vertices that were adjacent to either $u$ or $v$. Any parallel edges added in the process are deleted so that the graph remains simple. The resulting graph is denoted by $G/e$. Formally, for a given graph $G$ and edge $e = uv$, we define $G/e$ in the following way: $V(G/e) = (V(G) \cup \{w\}) \backslash \{u, v\}$ and $E(G/e) = \{xy \mid x,y \in V(G) \setminus \{u, v\}, xy \in E(G)\} \cup \{wx |\ x \in N_G(u) \cup N_G(v)\}$. 
For a subset of edges $F$ in $G$, graph $G/ F$ denotes the graph obtained from $G$ by repeatedly contracting edges in $F$ until no such edge remains.
We say that a graph $G$ is \emph{contractible} to a graph $H$ if there exists an onto function $\psi: V(G) \rightarrow V(H)$ such that the following properties hold.
\begin{itemize}
\setlength{\itemsep}{-2pt}
\item For any vertex $h$ in $V(H)$, graph $G[W(h)]$ is connected, where set $W(h) := \{v \in V(G) \mid \psi(v)= h\}$.
\item For any two vertices $h, h'$ in $V(H)$, edge $hh'$ is present in $H$ if and only if there exists an edge in $G$ with one endpoint in $W(h)$ and another in $W(h')$.
\end{itemize}
For a vertex $h$ in $H$, set $W(h)$ is called a \emph{witness set} associated with $h$. We define $H$-\emph{witness structure} of $G$, denoted by $\mathcal{W}$, as collection of all witness sets. Formally, $\mathcal{W}=\{W(h) \mid h \in V(H)\}$. Witness structure $\mathcal{W}$ is a partition of vertices in $G$, where each witness forms a connected set in $G$. Recall that if a \emph{witness set} contains more than one vertex, then we call it \emph{non-trivial} witness set, otherwise  a \emph{trivial} witness set. 

If graph $G$ has a $H$-witness structure, then graph $H$ can be obtained from $G$ by a series of edge contractions. For a fixed $H$-witness structure, let $F$ be the union of spanning trees of all witness sets. By convention, the spanning tree of a singleton set is an empty set. Thus, to obtain  $H$ from $G$, it is {\em sufficient} to contract edges in $F$. If such witness structure exists, then we say that graph $G$ is contractible to $H$. We say that graph $G$ is \emph{$k$-contractible} to $H$ if cardinality of $F$ is at most $k$. In other words, $H$ can be obtained from $G$ by at most $k$ edge contractions. Following observation is an immediate consequence of definitions.
\begin{observation}
 \label{obs:witness-structure-property} If graph $G$ is $k$-contractible to graph $H$, then the following statements are true.  
\begin{itemize}
\setlength{\itemsep}{-2pt}
\item For any witness set $W$ in a $H$-witness structure of $G$, the cardinality of $W$ is at most $k+1$. 
\item For a fixed $H$-witness structure, the number of vertices in $G$, which are contained in non-trivial witness sets is at most $2k$.
\end{itemize}
\end{observation}

In the following two observations, we state that if a graph can be transformed into a clique or a split graph by contracting few edges, then it can also be converted into a clique or split graph by deleting few vertices.

\begin{observation}
 \label{obs:vertex-solution}
 If a graph $G$ is $k$-contractible to a clique, then $G$ can be converted into a clique by deleting at most $2k$ vertices.
\end{observation}
\begin{proof} Let $F$ be a set of edges of size at most $k$ such that $G/F$ is a clique. Let $\calW$ be a $G/F$-witness structure of $G$. Let $X$ be a set of all vertices which are contained in the non-trivial witness sets in $\calW$. 
 By Observation~\ref{obs:witness-structure-property}, size of $X$ is at most $2k$.
 Any two vertices in $V(G) \setminus X$ are adjacent to each other as these vertices form singleton sets, which are adjacent in $G/F$. Hence, $G$ can be converted into a clique by deleting vertices in $X$.
\end{proof}

\begin{observation}\label{obsn-splitvd-solution} If a graph $G$ is $k$-contractible to a split graph then $G$ can be converted into a split graph by deleting at most $2k$ vertices. 
\end{observation}
\begin{proof}
 For  graph $G$, let $F$ be the set of edges such that $G/F$ is a split graph and $|F| \leq k$. Let $V(F)$ be the collection of all endpoints of edges in $F$. Since cardinality of $F$ is at most $k$, $|V(F)|$ is at most $2k$. We argue that $G - V(F)$ is a split graph. For the sake of contradiction, assume that $G - V(F)$ is not a split graph. We know that a graph is split if and only if it does not contain induced $C_4, C_5$ or $2K_2$. This implies that there exists a set of vertices $V'$ in $V(G) \setminus V(F)$ such that $G[V']$ is either $C_4, C_5$ or $2K_2$. Since no edge in $F$ is incident on any vertices in $V'$, $G/F[V']$ is isomorphic to $G[V']$. Hence, there exists a $C_4, C_5$ or $2K_2$ in $G/F$ contradicting the fact that $G/F$ is a split graph. Hence, our assumption is wrong and $G - V(F)$ is a split graph. \end{proof}

Consider a connected graph $G$ which is $k$-contractible to the clique $K_{\ell}$. Let $\calW$ be a $K_{\ell}$-witness structure of $G$. The following observation gives a sufficient condition for obtaining a witness structure of 
an induced subgraph of $G$ from  $\calW$. 
%

\begin{observation} \label{obs:merging-witness-sets} Let $\calW$ be a clique witness structure of  $G$. If there exists two different witness sets $W(t_1), W(t_2)$ in $\calW$ and a vertex $v$ in $W(t_1)$ such that the set $W(t) = (W(t_1) \cup W(t_2)) \setminus \{v\}$ is a connected set in $G - \{v\}$, then $\calW'$ is a clique witness structure of $G - \{v\}$, where $\calW'$ is obtained from $\calW$ by removing $W(t_1), W(t_2)$ and adding $W(t)$.
\end{observation}

\begin{proof} Let $G' = G - \{v\}$. Note that $\calW'$ is a partition of vertices in $G'$. Any set in $\calW' \setminus \{W(t)\}$ is a witness set in $\calW$ and does not contain $v$. Hence, these sets are connected in $G'$. Since $G'[W(t)]$ is also connected, all the witness sets in $\calW'$ are connected in $G'$.

 Consider any two witness sets $W(t'), W(t'')$ in $\calW'$. If none of these two is equal to $W(t)$ then both of these sets are present in $\calW$. Since none of these witness sets contains vertex $v$, they are adjacent to each other in $G'$. Now, consider a case when one of them, say $W(t'')$, is equal to $W(t)$. As witness sets $W(t')$ and $W(t_2)$ are present in $\calW$, there exists an edge with one endpoint in $W(t')$ and another in $W(t_2)$. The same edge is present in graph $G'$ as it is not incident on $v$. Since $W(t_2)$ is subset of $W(t)$, sets $W(t')$ and $W(t)$ are adjacent in $G'$. Hence any two witness sets in $\calW'$ are adjacent to each other. This implies that $\calW'$ is a clique witness structure of graph $G - \{v\}$. 
\end{proof}

In the case of \textsc{Split Contraction}, the following observation guarantees the existence of witness structure with a particular property.

\begin{observation}\label{obs:bags_in_clique} For a connected graph $G$, let $F$ be a set of edges such that $G/F$ is a split graph. Then, there exists a set of edges $F'$ which satisfy the following properties: $(i)$ $G/F'$ is a split graph. $(ii)$ The number of edges in $F'$ is at most $|F|$. $(iii)$ There exists a split partition of $G/F'$ such that all vertices in $G/F'$ which correspond to a non-trivial witness set in $G/F'$-witness structure of $G$ are in clique side. 
\end{observation}
\begin{proof}
 Let $(C,\;I)$ be a split partition of vertices of $G/F$ such that $C$ is a clique and $I$ is an independent set. If all the vertices corresponding to non-trivial witness sets are in $C$, then the observation is true. Consider a vertex $a$ in $I$ which corresponds to a non-trivial witness set $W_a$. Since $G$ is connected, $G/F$ is a connected split graph. This implies that there exists a vertex, say $b$, in $C$ which is adjacent to $a$ in $G/F$. We denote witness set corresponding to $b$ by $W_b$. We construct a new witness structure by shifting all but one vertices in $W_a$ to $W_b$. Since $ab$ is an edge in $G/F$, there exists an edge in $G$ with one endpoint in $W_a$ and another in $W_b$. Let that edge be $u_au_b$ with vertices $u_a$ and $u_b$ contained in sets $W_a$ and $W_b$, respectively. Consider a spanning tree $T$ of graph $G[W_a]$ which is rooted at $u_a$. We can replace edges in $F$ whose both endpoints are in $V(W_a)$ with $E(T)$ to obtain another set of edges $F^*$ such that $G/F^*$ is a split graph. Formally, $F^* = (F \cup E(T)) \setminus (E(G[W_a]) \cap F)$. Note that the number of edges in $F^*$ and $F$ are same. Let $v_1$ be a leaf vertex in $T$ and $v_2$ be its unique neighbor. Consider $F' = (F^* \cup \{u_au_b\}) \setminus \{v_1v_2\}$. Since edge $v_1v_2$ is in $F^*$ and $u_au_b$ is not in $F^*$, $|F'| = |F^*|$. We now argue that $G/F'$ is also a split graph.
 Let $\calW^{\prime}$ be the $G/F'$-witness structure of $G$. Note that $\calW^{\prime}$ can be obtained from $G/F^*$-witness structure $\calW^*$ of $G$ by replacing $W_a$ by $\{v_1\}$ and $W_b$ by $W_b \cup (W_a \setminus \{v_1\})$. Since all other witness set remains unchanged any witness set which was adjacent to $W_b$ is also adjacent to $W_b \cup (W_a \setminus \{v_1\})$. Similarly, any witness set which was not adjacent to  $W_a$ is not adjacent to $\{v_2\}$. In other words, this operation of shifting edges did not remove any vertex from the neighborhood of $b$ (which is in $C)$ nor it added any vertex in the neighborhood of $a$ (which is in $I$). Hence, $G/F'$ is also a split graph with $(C, I)$ as one of its split partition. Note that there exists a split partition of $G/F'$ such that the number of vertices in the independent side corresponding to non-trivial witness set is one less than the number of vertices in $I$ which corresponds to non-trivial witness sets. Hence, by repeating this process at most $|V(G)|$ times, we get a set of edges that satisfy three properties mentioned in the observation.   
\end{proof}

\subsection{Parameterized Complexity and Lossy Kernelization}
An important notion in parameterized complexity is \emph{kernelization}, which captures the efficiency of data reduction techniques. A parameterized problem $\rm \Pi$ admits a kernel of size $g(k)$ (or $g(k)$-kernel) if there is a polynomial time algorithm (called {\em kernelization algorithm}) which takes as  input $(I,k)$, and returns an instance $(I',k')$ of $\Pi$ such that: $(i)$ $(I, k)$ is a \yes-instance if and only if $(I', k')$ is a \yes-instance; and $(ii)$ $|I'| + k' \leq g(k)$, where $g(\cdot)$ is a computable function whose value depends only on $k$. Depending on whether the function $g(\cdot)$ is \emph{linear, polynomial} or \emph{exponential}, the problem is said to admit a \emph{linear, polynomial} or \emph{exponential kernel}, respectively. We refer to the corresponding chapters in the books~\cite{fomin2019kernelization,saurabh-book,DF-new,flumgrohe,niedermeier2006} for a detailed introduction to the field of kernelization.

In lossy kernelization, we work with the optimization analog of parameterized problem. Along with an instance and a parameter, an optimization analog of the problem also has a string called \emph{solution}. We start with the definition of a {\em parameterized optimization problem}. It is the parameterized analog of an optimization 
problem used in the theory of approximation algorithms.

\begin{definition}[Parameterized Optimization Problem] A parameterized optimization problem is a computable function $\Pi : \Sigma^* \times \mathbb{N} \times \Sigma^* \mapsto \mathbb{R} \cup \{\pm \infty\}$. The instances of $\Pi$ are pairs $(I,k) \in \Sigma^* \times \mathbb{N}$ and a solution to $(I,k)$ which is simply a string $S \in \Sigma^*$ such that $|S| \leq |I|+k$. 
\end{definition}

The {\em value} of a solution $S$ is $\Pi(I,k,S)$. In this paper, all optimization problems are minimization problems. Therefore, we present the rest of the section only with respect to parameterized \emph{minimization} problem. 
The {\em optimum value} of $(I,k)$ is defined as: 
$$\textsc{OPT}_{\Pi}(I, k)= \min_{S \in \Sigma^*,\, |S| \leq |I|+k} \Pi(I,k,S),$$ 
and an {\em optimum solution} for $(I,k)$ is a solution $S$ such that $\Pi(I,k,S)=\textsc{OPT}_{\Pi}(I, k)$. For a constant $c > 1$, $S$ is \emph{$c$-factor approximate} solution for $(I,k)$ if $\frac{\Pi(I, k, S)}{OPT_{\Pi}(I, k)} \le c$. 
We omit the subscript $\Pi$ in the notation for optimum value if the problem under consideration is clear from the context.

For some parameterized optimization problems we are unable to obtain \FPT\ algorithms, and we are also unable to find satisfactory polynomial time approximation algorithms. In this case one might aim for \FPT\-approximation algorithms, algorithms that run in time $f(k)n^c$ and provide good approximate solutions to the instance. 

\begin{definition}\label{def:fptAppx}
Let $\alpha \geq 1$ be constant. A fixed parameter tractable $\alpha$-approximation algorithm for a parameterized optimization problem $\Pi$ is an algorithm that takes as input an instance $(I,k)$, runs in time $f(k)|I|^{\OO(1)}$, and outputs a solution $S$ such that $\Pi(I,k,S) \leq \alpha \cdot \mbox{{\sc OPT}}(I,k)$ if $\Pi$ is a minimization problem, and $\alpha \cdot \Pi(I,k,S) \geq  \mbox{{\sc OPT}}(I,k)$ if $\Pi$ is a maximization problem.
\end{definition}
Note that Definition~\ref{def:fptAppx} only defines constant factor FPT-approximation algorithms. The definition can in a natural way be extended to approximation algorithms whose approximation ratio depends on the parameter $k$, on the instance $I$, or on both. Next, we define an {\em $\alpha$-approximate polynomial-time preprocessing algorithm} for a parameterized minimization problem $\Pi$ as follows.

\begin{definition}[$\alpha$-Approximate Polynomial-time Preprocessing Algorithm] \label{def:lossy-kernel} 

Let $\alpha \ge 1$ be a real number and $\Pi$ be a parameterized minimization problem. An $\alpha$-approximate polynomial-time preprocessing algorithm is defined as a pair of polynomial-time algorithms, called the {\em reduction algorithm} and the {\em solution lifting algorithm}, that satisfy the following properties.
\begin{itemize}
\setlength{\itemsep}{-2pt}
\item Given an instance $(I,k)$ of $\Pi$, the reduction algorithm computes an instance $(I',k')$ of $\Pi$. 
\item Given instances $(I,k)$ and $(I',k')$ of $\Pi$, and a solution $S'$ to $(I',k')$, the solution lifting algorithm computes a solution $S$ to $(I,k)$ such that $\frac{\Pi(I,k,S)}{\textsc{OPT}(I,k)} \leq \alpha \cdot \frac{\Pi(I',k',S')}{\textsc{OPT}(I',k')}$.
\end{itemize}
\end{definition}

We sometimes refer $\alpha$-approximate polynomial-time preprocessing algorithm kernel as \emph{$\alpha$-lossy rule} or \emph{$\alpha$-reduction rule}. 

\section{Lossy Kernel for \textsc{Clique Contraction}}
\label{sec:lossy_clique}

In this section, we present a lossy kernel for \textsc{Clique Contraction}. We first define a natural optimization version of the problem.
$$\textsc{ClC}(G, k, F) = \left\{ \begin{array}{rl} \min\{|F|, k + 1\} &\mbox{\text{ if } $G/F$ \text{ is a clique }} \\ \infty &\mbox{ otherwise} \end{array} \right.$$

If the number of vertices in the input graph is at most $k+3$, then we can return the same instance as a kernel for the given problem. Further, we assume that the input graph is connected; otherwise, it can not be edited into a clique by edge contraction only.  
Thus, we only consider {\em connected graphs with at least $k + 3$ vertices}. By the definition of optimization problem, for any set of edges $F$, if $G/F$ is a clique, then the maximum value of $\textsc{ClC}(G, k, F)$ is $k + 1$. Hence, any spanning tree of $G$ is a solution of cost $k + 1$. We call it a \emph{trivial solution} for the given instance. Consider an instance $(P_4, 1)$, where $P_4$ is a path on four vertices. One needs to contract at least two edges to convert $P_4$ into a clique. We call $(P_{4}, 1)$  a \emph{trivial \textsc{No}-instance} for this problem. Finally, we assume that we are given an $\epsilon>0$.

We start with a reduction rule, which says that if the minimum number of vertices that need to be deleted from an input graph to obtain a clique is large, then we can return a trivial instance as a lossy kernel.

\begin{reduction rule}
 \label{rr:large-cvd}
 For a given instance $(G, k)$, apply the algorithm mentioned in Observation~\ref{obs:factor-2-approx} to find a set $X$ such that $G - X$ is a clique. If the size of $X$ is greater than $4k$, then return $(P_4, 1)$.  
\end{reduction rule}

\begin{lemma} \label{lemma:rr-large-cvd-safe}
 Reduction Rule~\ref{rr:large-cvd} is a $1$-reduction rule.
\end{lemma}

\begin{proof} Let $(G, k)$ be an instance of \textsc{Clique Contraction} such that the Reduction Rule~\ref{rr:large-cvd} returns $(P_4, 1)$ when applied on it. The solution lifting algorithm returns a spanning tree $F$ of $G$.  Note that for a set of edges $F'$, if $P_4/F'$ is a clique then $F'$ contains at least two edges. This implies $\ClC(P_4, 1, F') = 2$ and $\OPT(P_4, 1) = 2$. 

 Since a factor $2$-approximation algorithm returned a set of size strictly more than $4k$, for any set $X'$ of size at most $2k$, $G - X'$ is not a clique. But by Observation~\ref{obs:vertex-solution}, if $G$ is $k$-contractible to a clique then $G$ can be edited into a clique by deleting at most $2k$ vertices. Hence, for any set of edges $F^*$ if $G/F^*$ is a clique, then the size of $F^*$ is at least $k + 1$. This implies that $\OPT(G, k) = k + 1$, and for a spanning tree $F$ of $G$, $\ClC(G, k, F) = k + 1$.
  
Combining these values, we get $\frac{\ClC(G, k, F)}{\OPT(G, k)} = \frac{k + 1}{k + 1} = \frac{2}{2} = \frac{\ClC(P_4, 1, F')}{\OPT(P_4, 1)}$. This implies that if $F'$ is factor $c$-approximate solution for $(P_4, 1)$, then $F$ is factor $c$-approximate solution for $(G,k)$. This concludes the proof.
\end{proof}

We now consider an instance $(G, k)$ for which Reduction Rule~\ref{rr:large-cvd} does not return a trivial instance. This implies that for a given graph $G$, in polynomial time, one can find a partition $(X, Y)$ of $V(G)$ such that $G - X = G[Y]$ is a clique and $|X|$ is at most $4k$. For $\epsilon > 0$, find a smallest integer $d$, such that $\frac{d + 1}{d} \le 1 + \epsilon$. In other words, fix $d = \lceil \frac{1}{\epsilon} \rceil$.
We note that if the number of vertices in the graph is at most $\calO(k^{d + 1})$, then the algorithm returns this graph as a lossy kernel of the desired size. Hence, without loss of generality, we assume that the number of vertices in the graph is larger than $\calO(k^{d + 1})$.

Given an instance $(G, k)$, a partition $(X, Y)$ of $V(G)$ with $G[Y]$ being a clique, and an integer $d$, consider the following two marking schemes.

\begin{marking-scheme}
 \label{marking:nbrs}
 For a subset $A$ of $X$, let $M_1(A)$ be the set of vertices in $Y$ whose neighborhood contains $A$. For every subset $A$ of $X$ which is of size at most $d$, mark a vertex in $M_1(A)$.
\end{marking-scheme}
\noindent Formally, $M_1(A) = \{y \in Y | A \subseteq N(y)\}$. If $M_1(A)$ is an empty set, then the marking scheme does not mark any vertex. If it is non-empty, then the marking scheme arbitrarily chooses a vertex and marks it.
\begin{marking-scheme}
 \label{marking:non-nbrs}
 For a subset $A$ of $X$, let $M_2(A)$ be the set of vertices in $Y$ whose neighborhood does not intersect $A$. For every subset $A$ of $X$ which is of size at most $d$, mark $2k + 1$ vertices in $M_2(A)$.
\end{marking-scheme}
\noindent Formally, $M_2(A) = \{ y \in Y | N(y) \cap A = \emptyset \}$. If the number of vertices in $M_2(A)$ is at most $2k + 1$, then the marking scheme marks all vertices in $M_2(A)$. If it is larger than $2k + 1$, then it arbitrarily chooses $2k + 1$ vertices and marks them.

\begin{reduction rule}
 \label{rr:delete-marked} For a given instance $(G, k)$, partition $(X, Y)$ of $V(G)$ with $G[Y]$ being a clique, and an integer $d$, apply the Marking Schemes~\ref{marking:nbrs} and \ref{marking:non-nbrs}. Let $G'$ be the graph obtained from $G$ by deleting all the unmarked vertices in $Y$. Return the instance $(G', k)$.
\end{reduction rule}

Above reduction rule can be applied in time $|X|^{d} \cdot |V(G)|^{\calO(1)} = \OO(k^{\OO(d)}|V(G)|^{\calO(1)})$ as $|X|$ is at most $4k$. 
Note that $G'$ is an induced subgraph of $G$. We first show that since $G$ is a connected graph, $G'$ is also connected. 
In the following lemma, we prove a stronger statement. 

\begin{lemma} \label{lemma:connected-clique} 
Consider instance $(G, k)$ of \textsc{Clique Contraction}. 
Let $Y'$ be the set of vertices marked by Marking Scheme~\ref{marking:nbrs} or \ref{marking:non-nbrs} for some positive integer $d$. 
For any subset $Y''$ of $Y \setminus Y'$, let $G''$ be the graph obtained from $G$ by deleting $Y''$.
Then, $G''$ is connected.
\end{lemma}
\begin{proof}
Recall that, by our assumption, $G$ is connected  and $Y$ is a clique in $G$.
Hence, for every vertex in $X$, there exists a path from it to some vertex in $Y$.
By the construction of $G''$, $(X, Y\setminus Y'')$ forms a partition of $V(G'')$ and $Y \setminus Y''$ is a clique in $G''$.
To prove that $G''$ is connected, it is sufficient to prove that for every vertex in $X$, there exists a path from it to a vertex in $Y \setminus Y''$ in $G$.

Fix an arbitrary vertex, say $x$, in $X$.
Consider a path $P$ from $x$ to $y$ in $G$, where $y$ is some vertex in $Y$. 
Without loss of generality, we can assume that $y$ is the only vertex in $V(P) \cap Y$.
We argue that there exists another path, say $P_1$, from $x$ to a vertex in $Y \setminus Y''$.
If $y$ is in $Y \setminus Y''$ then $P_1 = P$ is a desired path. 
Consider the case when $y$ is in $Y''$.
Let $x_0$ be the vertex in $V(P)$ which is adjacent with $y$. 
Note that $x_0$ may be same as $x$.
As Marking Scheme~\ref{marking:nbrs} considers all subsets of size at most $d$, it considered singleton set $\{x_0\}$.
As $x_0$ is adjacent with $y$, we have $\{x_0\} \subseteq N(y)$.
Since $y$ is in $Y''$, and hence unmarked, there exists a vertex, say $y_1$, in $Y$ which has been marked by Marking Scheme~\ref{marking:nbrs}.
Consider a path $P_1$ obtained from $P$ by deleting vertex $y$ (and hence edge $x_0y$) and adding vertex $y_1$ with edge $x_0y_1$. 
This is a desired path from $x$ to a vertex in $Y \setminus Y''$.
As $x$ is an arbitrary vertex in $X$, this statement is true for any vertex in $X$ and
hence $G''$ is connected.
\end{proof}

\begin{figure}[t]
 \centering
 \includegraphics[scale=0.60]{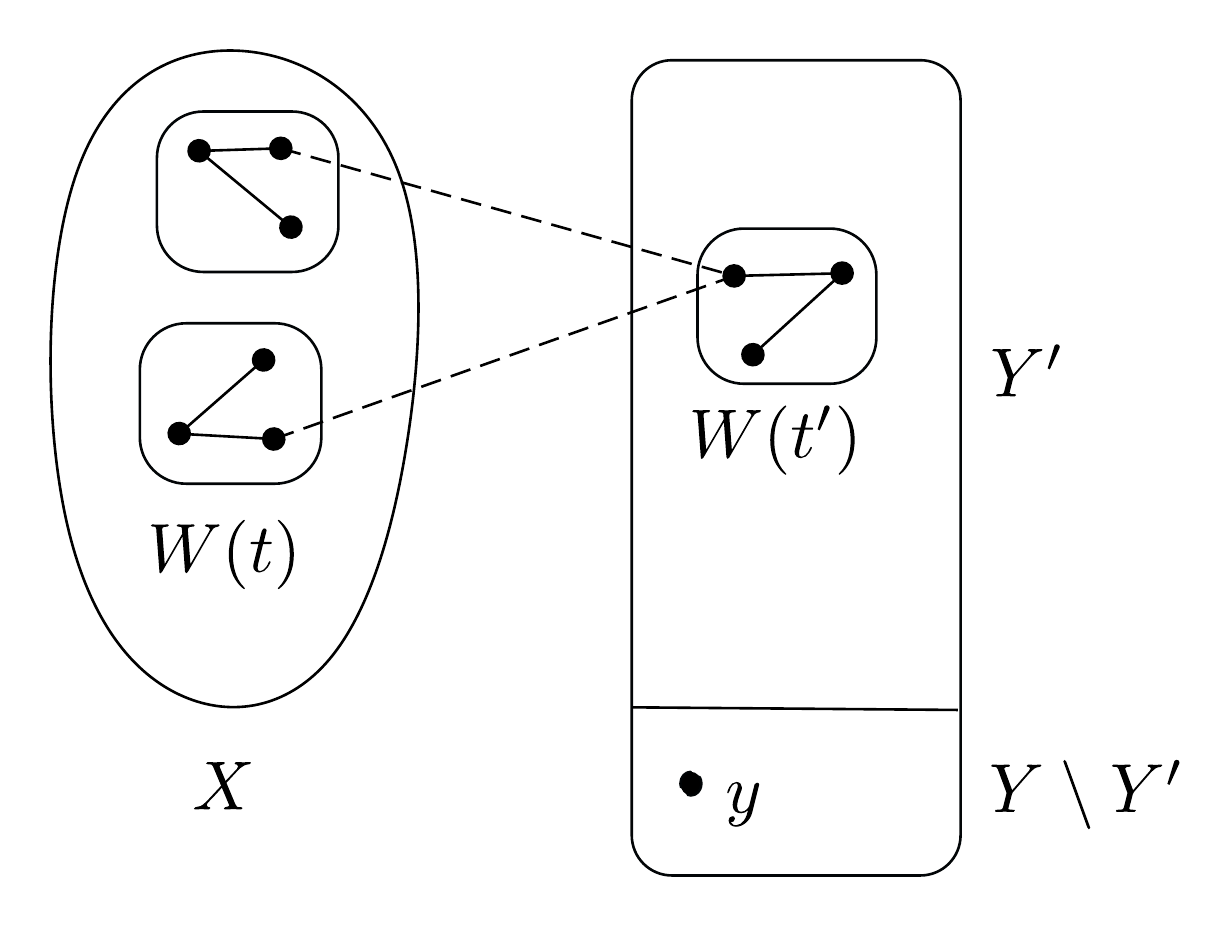}
 \caption{Straight lines (e.g within $W(t)$) represent edges in original solution $F$. Dashed lines (e.g. across $W(t)$ and $W(t')$) represents extra edges added to solution $F$. Please refer to the proof of Lemma~\ref{lemma:solution-lifting}. \label{fig:solution-lifting}}
\end{figure}

Thus, because of Lemma~\ref{lemma:connected-clique}, from now onwards, we assume that $G'$ is connected. 
In fact, in our one of the proof, we will iteratively remove vertices from $Y\setminus Y'$, and Lemma~\ref{lemma:connected-clique} ensures that the graph at each step remains connected.  In the following lemma, we argue that given a solution for $(G', k)$,  we can construct a solution of almost the same size for $(G, k)$. 

\begin{lemma} \label{lemma:solution-lifting} Let $(G', k)$ be the instance returned by Reduction Rule~\ref{rr:delete-marked} when applied on an instance $(G, k)$. If there exists a set of edges of size at most $k$, say $F'$, such that $G'/F'$ is a clique, then there exists a set of edges $F$ such that $G/F$ is a clique and cardinality of $F$ is at most $(1 + \epsilon) \cdot |F'|$.
\end{lemma}

\begin{proof} If no vertex in $Y$ is deleted, then $G'$ and $G$ are identical graphs, and the statement is true. We assume that at least one vertex in $Y$ is deleted. Let $Y'$ be the set of vertices in $Y$, which are marked. Note that the sets $X, Y'$ forms a partition of $V(G')$ such that $Y'$ is a clique and a proper subset of $Y$. Let $\calW'$ be a $G'/F'$-witness structure of $G'$. We construct a clique witness structure $\calW$ of $G$ from $\calW'$ by adding singleton witness sets $\{y\}$ for every vertex $y$ in $Y \setminus Y'$. Since $G[Y \setminus Y']$ is a clique in $G$, any two newly added witness sets are adjacent to each other. Moreover, any witness set in $\calW'$, which intersects $Y'$ is also adjacent to the newly added witness sets. We now consider witness sets in $\calW'$, which do not intersect $Y'$.

Let $\calW^\star$ be a collection of witness sets $W(t)$ in $\calW'$ such that $W(t)$ is contained in $X$ and there exists a vertex $y$ in $Y \setminus Y'$ whose neighborhood does not intersect with $W(t)$. See Figure~\ref{fig:solution-lifting}. We argue that every witness set in $\calW^\star$ has at least $d + 1$ vertices. For the sake of contradiction, assume that there exists a witness set $W(t)$ in $\calW^\star$ which contains at most $d$ vertices. Since Marking Scheme~\ref{marking:non-nbrs} iterated over all the subsets of $X$ of size at most $d$, it also considered $W(t)$ while marking. Note that the vertex $y$ belongs to the set $M_2(W(t))$. Since $y$ is unmarked, there are $2k + 1$ vertices in $M_2(W(t))$ which have been marked. All these marked vertices are in $G'$. Since the cardinality of $F'$ is at most $k$, the number of vertices in $V(F')$ is at most $2k$. Hence, at least one marked vertex in $M_2(W(t))$ is a singleton witness set in $\calW'$. However, there is no edge between this singleton witness set and $W(t)$. This non-existence of an edge contradicts the fact that any two witness sets in $\calW'$ are adjacent to each other in $G'$. Hence, our assumption is wrong, and $W(t)$ has at least $d + 1$  vertices.

Next, we show that there exists a witness set in $\calW'$ that intersects $Y'$. This is ensured by the fact that $G'$ is connected, and we are in the case where at least one vertex in $Y$ is deleted. The last assertion implies that  $Y'$ is non-empty, and hence there must be a witness set in $\calW'$ that intersects $Y'$. 
Let $W(t')$ be a witness set in  $\calW'$ that intersects  $Y'$. 
Note that $W(t')$ is adjacent to every vertex in $Y\setminus Y'$. Let $W(t)$ be a witness set in $\calW^\star$. Since $W(t')$ and $W(t)$ are two witness sets in the $G'/F'$-witness structure, there exists an edge with one endpoint in $W(t')$ and another in $W(t)$. Therefore, the set $W(t') \cup W(t)$ is adjacent to every other witness set in $\calW$. 

We now describe how to obtain $F$ from $F'$. We initialize $F = F'$. For every witness set $W(t)$ in $\calW^\star$ add an edge between $W(t)$ and $W(t')$ to the set $F'$. Equivalently, we construct a new witness set by taking the union of $W(t')$ and all witness sets $W(t)$ in $\calW^\star$. This witness set is adjacent to every vertex in $Y \setminus Y'$, and hence $G/F$ is a clique. We now argue the size bound on $F$. Note that we have added one extra edge for every witness set $W(t)$ in $\calW^\star$. We also know that every such witness set has at least $d + 1$ vertices. Hence, we have added one extra edge for at least $d$ edges in the solution $F'$. Moreover, since witness sets in $\calW^\star$ are vertex disjoint, no edge in $F$ can be part of two witness sets. This implies that the number of edges in $F'$ is at most $(\nicefrac{d + 1}{d})|F| \le (1 + \epsilon) \cdot |F|$.
\end{proof}

In the following lemma, we argue that the value of the optimum solution for the reduced instance can be upper bounded by the value of an optimum solution for the original instance. 

\begin{lemma}\label{lemma:reduction-rule} Let $(G', k)$ be the instance returned by Reduction Rule~\ref{rr:delete-marked} when applied on an instance $(G, k)$. If $\OPT(G, k) \le k$, then $\OPT(G', k) \le \OPT(G, k)$.
\end{lemma}
\begin{proof} Let $F$ be a set of at most $k$ edges in $G$ such that $\OPT(G, k) = \ClC(G, k, F)$ and $\calW$ be a $G/F$-witness structure of $G$. Since we are working with a minimization problem, to prove this lemma it is sufficient to find a solution for $G'$ which is of size $|F|$.
 Recall that $(X, Y)$ is a partition of $V(G)$ such that $G - X = G[Y]$ is a clique. 
 Let $Y'$ be the set of vertices that were marked by either of the marking schemes. In other words, $(X, Y')$ is a partition of $G'$ such that $G' - X = G'[Y]$ is a clique. 
 We proceed as follows. At each step, we construct graph $G^\star$ from $G$ by deleting one or more vertices of $Y \setminus Y'$. Simultaneously, we also construct a set of edges $F^\star$ from $F$ by either replacing the existing edges by new ones or by simply adding extra edges to $F$. At any intermediate state, we ensure that $G^\star/F^\star$ is a clique, and the number of edges in $F^\star$ is at most $|F|$. Let $F^{\circ} = F$ be an optimum solution for the input instance $(G, k)$. {\em For notational convenience, we rename $G^\star$ to $G$ and $F^\star$ to $F$ at regular intervals but do not change $F^{\circ}$. } 

To obtain $G^\star$ and $F^\star$, we delete witness sets which are subsets of $Y \setminus Y'$ (Condition~~\ref{cond:witness-subset-YY}) and modify the ones which intersect with $Y \setminus Y'$. Every witness set of latter type intersects with $Y'$ or $X$ or both. We partition these non-trivial witness sets in $\calW$ into two groups depending on whether the intersection with $X$ is empty (Condition~\ref{cond:witness-does-not-intersectsX}) or not (Condition~\ref{cond:witness-intersectsX}). We first modify the witness sets that satisfy the least indexed condition. If there does not exist a witness set which satisfies either of these three conditions, then $Y \setminus Y'$ is an empty set, and the lemma is vacuously true.

\begin{condition1}\label{cond:witness-subset-YY}
 There exists a witness set $W(t)$ in $\calW$ which is a subset of $Y \setminus Y'$.
\end{condition1}

Construct $G^\star$ from $G$ by deleting the witness sets $W(t)$ in $\calW$. Let $F^\star$ be obtained from $F$ by deleting those edges whose both the endpoints are in $W(t)$. Since the class of cliques is closed under vertex deletion, $G^\star/F^\star$ is a clique, and as we only deleted edges from $F$, we have $|F^*|\leq |F|$. We repeat this process until there exists a witness set that satisfies Condition~\ref{cond:witness-subset-YY}.  

%
%
%

\smallskip

\begin{tcolorbox}[colback=red!5!white,colframe=red!75!black]
\centerline {At this stage we rename $G^\star$ to $G$ and $F^\star$ to $F$.}
\end{tcolorbox}

\begin{condition1}\label{cond:witness-does-not-intersectsX}
There exists a witness set $W(t)$ in $\calW$ which contains vertices from $Y \setminus Y'$ but does not intersect $X$.
\end{condition1}

Since $W(t)$ is not contained in $Y\setminus Y'$ and $W(t)\cap X$ is empty it must intersect with $Y'$. See Figure~\ref{fig:reduction-rule}. Let $y_4$ and $y_5$ be vertices in $W(t) \cap Y'$ and $W(t) \cap (Y \setminus Y')$, respectively. Let $W(t_1)$, different from $W(t)$, be a witness set which intersects $Y'$. Since $Y'$ is large and non-empty, such a witness set exists.
Let $y_6$ be a vertex in the set $W(t_1) \cap Y'$. Consider the witness sets $W(t), W(t_1)$ and vertex $y_5$ in $W(t)$ in graph $G$. Lemma~\ref{lemma:connected-clique} implies that these witness sets  satisfy the premise of Observation~\ref{obs:merging-witness-sets}. This implies $\calW^\star$ is a clique witness structure of $G - \{y_5\}$, where $\calW^\star$ is obtained from $\calW$ by removing $W(t), W(t_1)$ and adding $(W(t) \cup W(t_1)) \setminus \{y_5\}$. This corresponds to replacing an edge in $F$ which was incident to $y_5$ with the one across $W(t)$ and $W(t_1)$. For example, in Figure~\ref{fig:reduction-rule}, we replace edge $y_4y_5$ in the set $F$ with an edge $y_4y_6$ to obtain a solution for $G - \{y_5\}$. An edge in $F$ has been replaced with another edge and one vertex in $Y \setminus Y'$ is deleted. The size of $F^\star$ is same as that of $F$ and $G^\star/F^\star$ is a clique. We repeat this process until there exist a witness set which satisfies Condition~\ref{cond:witness-does-not-intersectsX}.

\begin{tcolorbox}[colback=red!5!white,colframe=red!75!black]
\centerline {At this stage we rename $G^\star$ to $G$ and $F^\star$ to $F$.}
\end{tcolorbox}

\begin{figure}[t]
 \centering
 \includegraphics[scale=0.60]{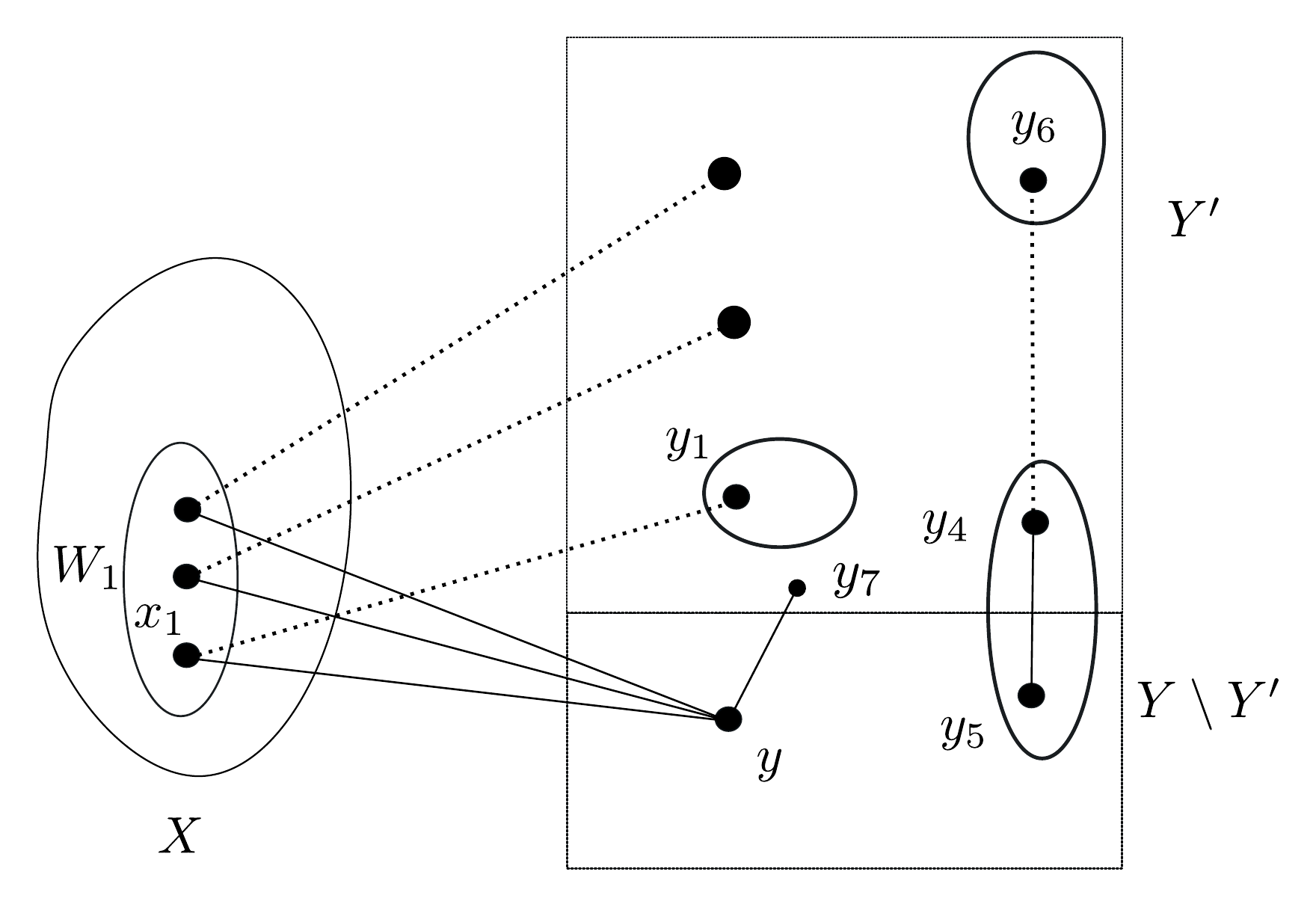}
 \caption{Straight lines (e.g. $y_4y_5$) represent edges in original solution $F$. Dotted lines (e.g. $y_4y_6)$ represents edges which are replaced for some edges in $F$. 
  Please refer to the proof of Lemma~\ref{lemma:reduction-rule}. \label{fig:reduction-rule}}
\end{figure}
\begin{condition1}\label{cond:witness-intersectsX} There exists a witness set $W(t)$ in $\calW$ which contains vertices from $Y \setminus Y'$ and intersects $X$.
\end{condition1}

Let $y$ be a vertex in $W(t) \cap (Y \setminus Y')$, $X_t$ be the set of vertices in $W(t) \cap X$ which are adjacent to $y$ via edges in $F$, and $Q_t$ be the set of vertices in $W(t)\cap Y$ which are adjacent to $y$ via edges in $F$. 
We find a \emph{substitute} for $y$ in $Y'$.
If the set $X_t$ is empty then the vertex $y$ is adjacent only with the vertices of $Y$, in this case the edges incident to $y$ can be replaced as mentioned in the Condition~\ref{cond:witness-does-not-intersectsX}.
Assume that $X_t$ is non-empty.
For every vertex $x$ in $X_t$ the set $\{x\}$ is considered by Marking Scheme ~\ref{marking:nbrs}. 
Since $y$ is adjacent to every vertex $x$ in $X_t$, the set $M_1(\{x\})$ is non-empty.
As $y$ is in $Y\setminus Y'$, and hence unmarked, for every $x$ in $X_t$, there is a vertex in $M_1(\{x\})$, say $y_x$, different from $y$ which has been marked.
We construct $F^\star$ from $F$ by the following operation: For every vertex $x$ in $X_t$, replace the edge $xy$ in $F$ by $xy_x$.
Fix a vertex $x_o$ in $X_t$, and for every vertex $u$ in $Q_t$, replace the edge $uy$ in $F$ with $uy_{x_o}$.
Since we are replacing a set of edges in $F$ with another set of edges of same size we have $|F^\star|\leq |F|$.
(For example, in Figure~\ref{fig:reduction-rule}, $X_t = W_1$ and $Q_t = \{y_7\}$.
Edges $yx_1, yy_7$ are replaced by $x_1y_1, y_1y_7$ resp.)
We argue that if $G^\star$ is obtained from $G$ by removing $y$, then $G^\star/F^\star$ is a clique.

We argue that contracting edges in $F^{\star}$ partitions $W(t)$ into $|X_t| + |Q_t|$ many parts and merges each part with some witness set in $\calW \setminus \{W(t)\}$.
Recall that $F$ contains a spanning tree of graph $G[W(t)]$.
Let $T$ be a spanning tree of $G[W(t)]$ such that $E(T) \subseteq F$ and $T$ contains all edges in $F$ that are incident on $y$.
It is easy to see that such a spanning tree exists.
Let $y$ be the root of tree $T$.
For every $z$ in $X_t \cup Q_t$, let $W'(z)$ be the set of vertices in the subtree of $T$ rooted at $z$.
As $V(T) = W(t)$, set $\{W'(z) |\ z \in X_t \cup Q_t \}$ is a partition of $W(t) \setminus \{y\}$.
For every $x$ in $X_t$, let $W(y_x)$ be the witness set in $\calW$ containing the vertex $y_x$.
For every $x$ in $X_t \setminus \{x_o\}$, let $W^\star(y_x)$ be the set $ W(y_x) \cup W'(x)$.
Let $W^\star(y_{x_o})$ be the set $W(y_{x_o}) \cup W'(x_o) \cup \bigcup_{y'} W'(y')$ for all $y'$ in $Q_t$.
We obtain $\calW^\star$ from $\calW$ by removing $W(t)$ and $W(y_x)$ for every $x$ in $X_t$, and adding the sets $W^\star(y_x)$ for every $x$ in $X_t$.
Since $W^\star(y_x)$ contains the set $W(y_x)$ which was adjacent to every witness set in $\calW$, $W^\star(y_x)$ will be adjacent with every witness set in $\calW^\star$. 
We repeat this process until there exists a witness set that satisfies this condition.

 Any vertex in $Y \setminus Y'$ must be a part of some witness set in $\calW$, and any witness set in $\calW$ satisfies at least one of the above conditions. If there are no witness sets that satisfy these conditions, then $Y \setminus Y'$ is empty. 
This implies $G^\star = G'$ and there exists a solution $F^\star$ of size at most $ |F^{\circ}|$. This concludes the proof of the lemma.
\end{proof}

We are now in a position to prove the following lemma.

\begin{lemma} \label{lemma:rr-delete-marked-safe}
Reduction Rule~\ref{rr:delete-marked}, along with a solution lifting algorithm, is an $(1 + \epsilon)$-reduction rule.
\end{lemma}
\begin{proof} Let $(G', k)$ be the instance returned by Reduction Rule~\ref{rr:delete-marked} when applied on an instance $(G, k)$. We present a solution lifting algorithm as follows. For a solution $F'$ for $(G, k)$ if $\textsc{ClC}(G', k, F') = k + 1$, then the solution lifting algorithm returns a spanning tree $F$ of $G$ (a trivial solution) as solution for $(G, k)$. In this case, $\ClC(G, k, F) = \ClC(G', k, F')$. If $\textsc{ClC}(G', k, F') \le k$, then size of $F'$ is at most $k$ and $G'/F'$ is a clique. Solution lifting algorithm uses Lemma~\ref{lemma:solution-lifting} to construct a solution $F$ for $(G, k)$ such that cardinality of $F$ is at most $(1 + \epsilon) \cdot |F'|$. In this case, $\ClC(G, k, F) \le (1 + \epsilon) \cdot \ClC(G', k, F')$. Hence, there exists a solution lifting algorithm which given a solution $F'$ for $(G', k')$ returns a solution $F$ for $(G, k)$ such that $\ClC(G, k, F) \le (1 + \epsilon) \cdot \ClC(G', k, F')$.

 If $\OPT(G, k) \le k$, then by Lemma~\ref{lemma:reduction-rule}, $\OPT(G', k) \le \OPT(G, k)$. If $\OPT(G, k) = k + 1$ then $\OPT(G', k) \le k + 1 = \OPT(G, k)$. Hence in either case, $\OPT(G', k) \le \OPT(G, k)$.

 Combining the two inequalities, we get $\frac{\ClC(G, k, F)}{ \OPT(G, k)} \le \frac{(1 + \epsilon) \cdot \ClC(G', k, F')}{\OPT(G', k)}$. This implies that if $F'$ is a factor $c$-approximate solution for $(G', k)$ then $F$ is a factor $(c \cdot (1 + \epsilon))$-approximate solution for $(G,k)$. This concludes the proof.
\end{proof}

We are now in a position to present the main result of this section.

\begin{reptheorem}{thm:clique-lossy}
For any $\epsilon > 0$, \textsc{Clique Contraction} parameterized by the size of solution $k$, admits a time efficient $(1 + \epsilon)$-approximate polynomial kernel with $\calO(k^{d + 1})$ vertices, where $d = \lceil\frac{1 }{\epsilon}\rceil$.
\end{reptheorem}
\begin{proof} For a given instance $(G, k)$ with $|V(G)|\geq k+3$, a kernelization algorithm applies the Reduction Rule~\ref{rr:large-cvd}. If it returns a trivial instance, then the statement is vacuously true. If it does not return a trivial instance, then the algorithm partitions $V(G)$ into two sets $(X, Y)$ such that $G - X = G[Y]$ is a clique and size of $X$ is at most $4k$. Then the algorithm applies the Reduction Rule~\ref{rr:delete-marked} on the instance $(G, k)$ with the partition $(X, Y)$ and the integer $d = \lceil \frac{1 }{\epsilon} \rceil$. The algorithm returns the reduced instance as $(1 + \epsilon)$-lossy kernel for $(G, k)$.

 The correctness of the algorithm follows from Lemma~\ref{lemma:rr-large-cvd-safe} and Lemma~\ref{lemma:rr-delete-marked-safe} combined with the fact that Reduction Rule~\ref{rr:delete-marked} is applied at most once. By Observation~\ref{obs:factor-2-approx}, Reduction Rule~\ref{rr:large-cvd} can be applied in polynomial time. The size of the instance returned by Reduction Rule~\ref{rr:delete-marked} is at most $\calO((4k)^d \cdot (2k + 1) + 4k) = \calO(k^{d+1})$. Reduction Rule~\ref{rr:delete-marked} can be applied in time $n^{\calO(1)}$ if the number of vertices in $(G, k)$ is more than $\calO(k^{d+1})$.  
\end{proof}


\section{Lossy Kernel for Split Contraction}
\label{sec:lossy_split}

In this section, we present a lossy kernel for \textsc{Split Contraction}. We start by defining a 
natural optimization version of the problem.

$$\textsc{SpC}(G, k, F) = \left\{ \begin{array}{rl}
                    \min\{|F|, k + 1\} & \mbox{\text{ if } $G/F$ \text{ is a split graph}} \\
\infty & \mbox{ otherwise.} \end{array} \right.$$

We assume that the input graph is {\em connected and justify this assumption at the end}. If the number of vertices in the input graph is at most $k+3$, then we return the same instance as a kernel for the given problem. Thus we only consider inputs that have at least $k + 3$ vertices. By the definition of optimization problem, for any set of edges $F$ if $G/F$ is a split graph then the maximum value of $\textsc{SpC}(G, k, F)$ is $k + 1$. Hence, any spanning tree of $G$ is a solution of cost $k + 1$. We call it a \emph{trivial solution} for the given instance. Consider an instance $(C_5, 1)$ where $C_5$ is a cycle on five vertices. One needs to contract at least two edges to convert $C_5$ into a split graph. We say $(C_{5}, 1)$ a \emph{trivial no instance} for \textsc{Split Contraction}.

 We start with a reduction rule, which says that if the minimum number of vertices that need to be deleted from an input graph to obtain a split graph is large, then we can return a trivial instance as a lossy kernel.
 
 \begin{reduction rule}
 \label{rr:large-splitvd}
 Given an instance $(G, k)$, apply the algorithm mentioned in Observation~\ref{obsn:approx-svd-solution} to find a set $S$ such that $G-S$ is a split graph. If $|S|>10k$ then return $(C_5,1)$. 
\end{reduction rule}

\begin{lemma} \label{lemma:rr-large-splitvd-safe}
 Reduction Rule~\ref{rr:large-splitvd} is a $1$-reduction rule.
\end{lemma}
\begin{proof} Let $(G,k)$ be an instance such that Reduction Rule \ref{rr:large-splitvd} returns $(C_5,1)$ when applied on it. Solution lifting algorithm returns a spanning tree $F$ of $G$.

For a set of edges $F'$, if $C_5/F'$ is a split graph then $F'$ contains at least two edges. This implies $\OPT(C_5,1)=2$.

Since a factor $5$-approximate algorithm returns a set of size strictly more than $10k$, for any $S'$ of size at most $2k$, $G-S'$ is not a split graph. But by Observation \ref{obsn-splitvd-solution} if $G$ is $k$-contractible to a split graph then $G$ can be converted into a split graph by deleting at most $2k$ vertices. Hence, for any set of edges $F^\star$, if $G/F^\star$ is a split graph, then the size of $F^\star$ is at least $k+1$. This implies that $\OPT(G,k)=k+1$.

 Combining these values, we get $\frac{\SpC(G, k, F)}{\OPT(G, k)} = \frac{k + 1}{k + 1} = \frac{2}{2} = \frac{\SpC(C_5, 1, F')}{\OPT(C_5, 1)}$. This implies that if $F'$ is a factor $c$-approximate solution for $(C_5, 1)$, then $F$ is a factor $c$-approximate solution for $(G,k)$. This concludes the proof. 
\end{proof}

We consider an instance $(G, k)$ for which Reduction Rule~\ref{rr:large-splitvd} does not return a trivial instance.
This implies that for a given graph $G$, in polynomial time, one can find a partition $(S, X, Y)$ of $V(G)$ such that $|S|$ is at most $10k$ and $G - S$ is a split graph with $(X, Y)$ as its split partition.
Recall that, our objective is to present a $(2 + \epsilon)$-approximate polynomial kernel for given $\epsilon > 0$. 
Fix $\alpha = (2 + \epsilon)/2$.
Find a smallest integer $d$ such that $\frac{d + 1}{d} \le \alpha$. In other words, fix $d = \lceil \frac{1 }{\alpha -1 } \rceil$.
If the number of vertices in the graph is at most $\calO(k^{d^2 + d + 1})$, then the algorithm returns the original  graph as a lossy kernel of the desired size. Hence, without loss of generality, we assume that the number of vertices is larger than $\calO(k^{d^2 + d + 1})$.

Given an instance $(G,k)$, a partition $(S, X, Y)$ of $V(G)$, and an integer $d$, we marks some vertices in $Y$ using the following two marking schemes. 
\begin{marking-scheme}\label{mrk-nbrs-in-I} For a subset $A$ of $S$, let $M_{NY}(A)$ be the set of neighbors of $A$ in $Y$. For every subset $A$ of $S$ whose size is at most $d$, mark $k+2$ vertices in $M_{NY}(A).$ 
\end{marking-scheme}
Formally, $M_{NY}(A)=\{y\in Y| y \in N(A)\}$. If the number of vertices in $M_{NY}(A)$ is at most $k+2$, then the marking scheme marks all vertices in $M_{NY}$, else it arbitrarily chooses $k+2$ vertices and marks them.
\begin{marking-scheme}\label{mrk-cmn-nbrs-in-I} For a subset $A$ of $S$, let $M_{CY}(A)$ be the set of vertices in $Y$ whose neighborhood contains $A$. For every subset $A$ of $S$ whose size is at most $d$, mark a vertex in $M_{CY}(A)$. Marking scheme prefers a vertex with highest degree.
\end{marking-scheme}
Formally, $M_{CY}(A)=\{y \in Y | A \subseteq N(y)\}$. If $M_{CY}(A)$ is empty, then marking scheme does not mark any vertex, otherwise it picks a vertex with the highest degree.

Let $Y'$ be the set of vertices of $Y$ that have been marked by the Marking Schemes \ref{mrk-nbrs-in-I} or \ref{mrk-cmn-nbrs-in-I}. Using set $S$ and marked vertices in $Y$, Marking Schemes~\ref{mrk-cmn-nbr-in-C}~and~\ref{mrk-nonnbrs-in-C} marks some vertices in $X$. We remark that these two schemes are
similar to Marking Schemes~\ref{marking:nbrs}~and~\ref{marking:non-nbrs}
\begin{marking-scheme}\label{mrk-cmn-nbr-in-C} For a subset $A$ of $S\cup Y'$, let $M_{CX}(A)$ be the set of vertices in $X$ whose neighborhood contains $A$. For every subset $A$ of $S \cup Y'$ whose size is at most $d$, mark two vertices in $M_{CX}(A)$. 
\end{marking-scheme} 
Formally, $M_{CX}(A)=\{ x\in X |\ A\subseteq N(x) \}$. If $M_{CX}(A)$ is empty, then the marking scheme does not mark any vertex, and if it has only one vertex, then the marking scheme marks that vertex. If it has at least two vertices, then the marking scheme arbitrarily chooses two vertices and marks them.
\begin{marking-scheme}\label{mrk-nonnbrs-in-C} For a subset $A$ of $S \cup Y'$, let $M_{NX}(A)$ be the set of vertices in $X$ whose neighborhood does not intersect $A$. For every subset $A$ of $S\cup Y'$ whose size is at most $d$, mark $2k+2$ vertices in $M_{NX}(A)$.
\end{marking-scheme}
Formally, $M_{NX}(A)=\{x\in X |\ N(x)\cap A = \emptyset \}$. If the number of vertices in $M_{NX}(A)$ is at most $2k+2$, then the marking scheme marks all vertices in $M_{NX}(A)$. If the number is greater than $2k+2$, then the marking scheme arbitrarily chooses $2k+2$ vertices and marks them. 

\begin{reduction rule} \label{rr:splitcontraction-reduction} For a given instance $(G, k)$, a partition $(S, X, Y)$ of $V(G)$, and an integer $d$, 
 apply Marking Schemes~\ref{mrk-nbrs-in-I}~to~\ref{mrk-nonnbrs-in-C}. Let $G'$ be the graph obtained from $G$ by deleting unmarked vertices of $X$ and $Y$. Return the instance $(G',k)$.
\end{reduction rule} 

The number of vertices in $Y$ marked by Marking Schemes~\ref{mrk-nbrs-in-I}~and~\ref{mrk-cmn-nbrs-in-I} is at most $\calO(|S|^{d + 1})$. This implies that the total number of vertices marked by these four marking schemes is at most $\calO(|S|^{d^2 + d + 1})$. Above reduction rule can be applied in time $|S|^{d^2 + d + 1} \cdot |V(G)|^{\calO(1)} = |V(G)|^{\calO(1)}$ as $|S|$ is at most $10k$ and number of vertices in $G$ is at least $\calO(k^{d^2 + d+ 1})$.
Note that $G'$ is an induced subgraph of $G$, and hence $G' - S$ is a split graph with $(X', Y')$ as its split partition.
We first prove the following lemma which is similar to Lemma~\ref{lemma:connected-clique}.  
\begin{lemma} \label{lemma:connected-split} 
Consider instance $(G, k)$ of \textsc{Split Contraction}. 
Let $Y'$ be the set of vertices marked by Marking Schemes~\ref{mrk-nbrs-in-I}~to~\ref{mrk-nonnbrs-in-C} for some positive integer $d$. 
For any subset $Z''$ of $(X \setminus X') \cup (Y \setminus Y')$, let $G''$ be the graph obtained from $G$ by deleting $Z''$.
Then, $G''$ is connected.
\end{lemma}
\begin{proof}
Recall that, by our assumption, $G$ is connected and $X$ is a clique in $G$.
Hence, for every vertex in $S \cup Y$, there exists a path from it to some vertex in $X$.
By the construction of $G''$, $(S, X \setminus Z'', Y \setminus Z'')$ forms a partition of $V(G'')$ and $X \setminus Z''$ is a clique in $G''$.
To prove that $G''$ is connected, it is sufficient to prove that for every vertex in $S \cup (Y \setminus Z'')$, there exists a path from it to a vertex in $X \setminus Z''$ in $G$.

We first prove that every vertex in $Y \setminus Z''$ has a path from it to a vertex in $X \setminus Z''$.
Fix an arbitrary vertex, say $y$, in $Y \setminus Z''$.
Consider a path $P$ from $y$ to $x$ in $G$, where $x$ is some vertex in $X$. 
Without loss of generality, we can assume that $x$ is the only vertex in $V(P) \cap X$.
We argue that there exists a path, say $P_1$, from $y$ to a vertex in $X \setminus Z''$.
If $x$ is in $X \setminus Z''$ then $P_1 = P$ is a desired path. 
Consider the case when $x$ is in $Z''$.
Let $w$ be the vertex in $V(P)$ which is adjacent with $x$.
Note that $w$ is either in $S$ or in $Y$. 
It may be the same as $y$.
As Marking Scheme~\ref{mrk-cmn-nbr-in-C} considered all subsets of size at most $d$ in $S \cup (Y \setminus Z'')$, it considered the singleton set $\{w\}$.
As $w$ is adjacent with $x$, we have $\{w\} \subseteq N(x)$.
Since $x$ is in $Z''$, and hence unmarked, there exists a vertex, say $x_1$, in $X$ which has been marked by Marking Scheme~\ref{mrk-cmn-nbr-in-C}.
Consider the path $P_1$ obtained from $P$ by deleting vertex $x$ (and hence edge $wx$) and adding vertex $x_1$ with edge $wx_1$. 
This is a desired path from $y$ to a vertex in $X \setminus Z''$.
As $y$ is an arbitrary vertex in $Y \setminus Z''$, this statement is true for any vertex in $Y \setminus Z''$ and
hence $G''$ is connected.

Now, it is sufficient to argue that for every vertex in $S$, there exists a path from it to a vertex in $(X \cup Y) \setminus Z''$. 
Fix an arbitrary vertex, say $s$, in $S$.
Consider a path $P$ from $s$ to $u$ in $G$, where $u$ is some vertex in $X \cup Y $. 
As $G$ is connected, such a path exists. 
Without loss of generality, we can assume that $u$ is the only vertex in $V(P) \cap (X \cup Y)$.
We argue that there exists a path, say $P_1$, from $s$ to a vertex in $(X \cup Y) \setminus Z''$.
If $u$ is in $(X \cup Y) \setminus Z''$ then $P_1 = P$ is a desired path. 
Consider the case when $u$ is in $Z''$.
Let $s_0$ be the vertex in $V(P)$ which is adjacent with $u$.
Note that $s_0$ is in $S$ and it may be the same as $s$.
As Marking Scheme~\ref{mrk-nbrs-in-I} and \ref{mrk-cmn-nbr-in-C} considered all subsets of size at most $d$ in $S$, it considered the singleton set $\{s_0\}$.
As $s_0$ is adjacent with $u$, we have $\{s_0\} \subseteq N(u)$.
Since $u$ is in $Z''$, and hence unmarked, there exists a vertex, say $u_1$, in $(X \cup Y) \setminus Z''$ which has been marked by Marking Scheme~\ref{mrk-nbrs-in-I} or \ref{mrk-cmn-nbr-in-C}.
Consider a path $P_1$ obtained from $P$ by deleting vertex $u$ (and hence edge $s_0u$) and adding the vertex $u_1$ with edge $s_0u_1$. 
This is a desired path from $s$ to a vertex in $(X \cup Y) \setminus Z''$.
As $s$ is an arbitrary vertex in $S$, this statement is true for any vertex in $S$.

Hence, there exists a path from every vertex in $S \cup (Y \setminus Z'')$ to a vertex in $X \setminus Z''$ in $G''$.
As $X \setminus Z''$ is a clique in $G''$, we can conclude that $G''$ is connected.
\end{proof}

As in case of \textsc{Clique Contraction}, we will iteratively remove vertices from $(X \setminus X') \cup (Y\setminus Y')$, and Lemma~\ref{lemma:connected-split} ensures that the graph at each step remains connected. 

To avoid corner cases, we need to ensure that whenever $X \setminus X'$ is non-empty, there is at least one witness set, which contains a vertex in $X'$. We ensure that by marking a few additional vertices in $X$.
\begin{remark1}\label{remark-extra-coloring} Mark any $2k + 2$ vertices in $X$.
\end{remark1}
Note that we can not infer anything about the adjacency of these vertices with vertices in $S \cup Y$. We use these vertices only to add certain edges, which are entirely contained in $X$.

In Lemma~\ref{lem-numerator-SC}, we argue that given a solution for $(G', k)$, we can construct a solution of almost the same size for $(G, k)$.

\begin{figure}[t]
 \centering \includegraphics[scale=0.60]{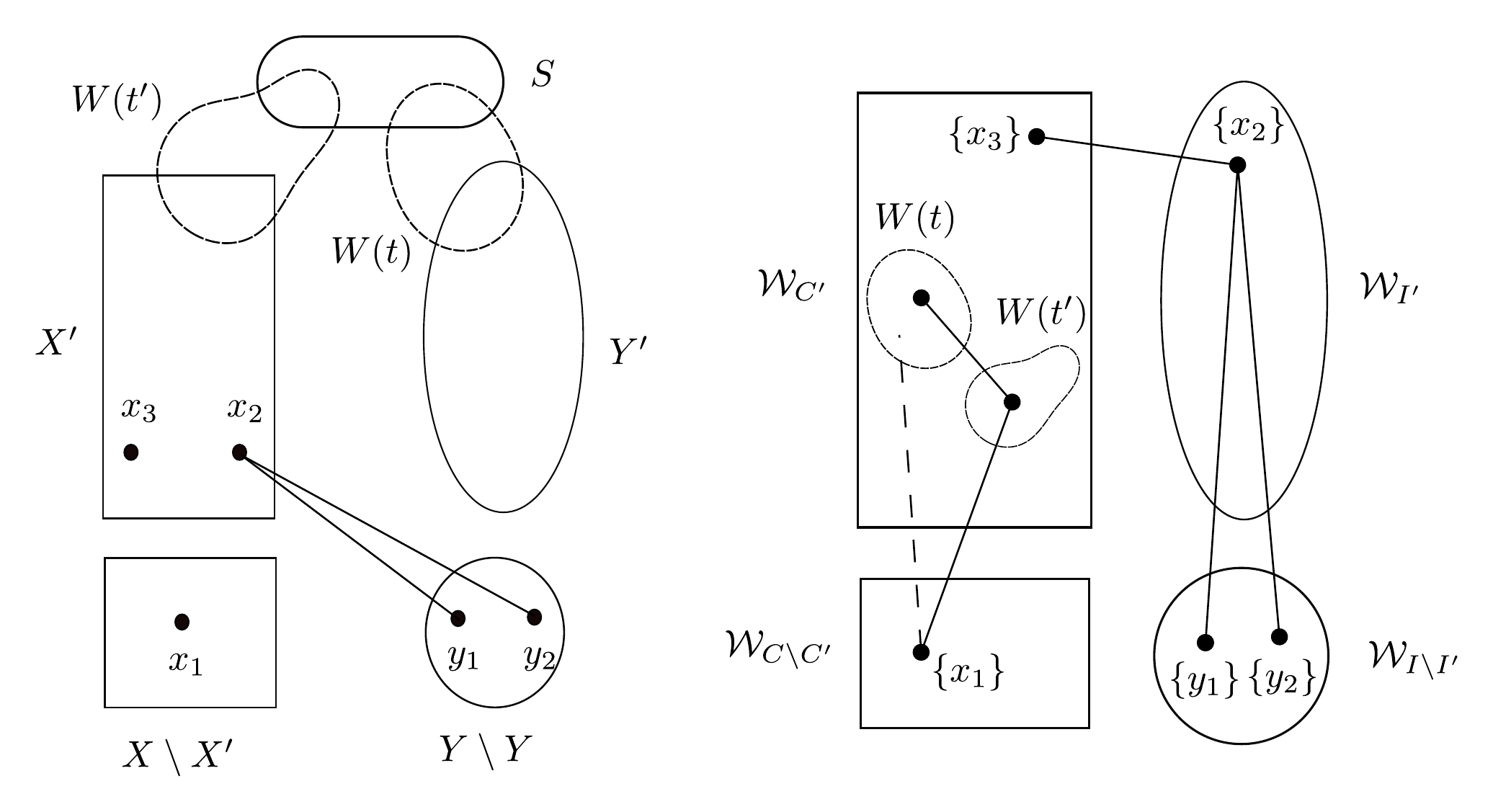}
 \caption{Left hand side figure shows partition $(S, X, Y)$ of $G$ where as right hand side figure shows witness sets in $G$ and $G'$. We need to modify witness sets in $\calW_{I}$ which are adjancent with each other (e.g. $\{y_1\}$ and $\{x_2\}$) and witness sets in $\calW_C$ which are not adjancent with each other (e.g. $W(t)$ and $\{x_1\}$). Please refer to Lemma~\ref{lem-numerator-SC}. \label{fig:solution-lifting-split}}
\end{figure}

\begin{lemma} \label{lem-numerator-SC} Let $(G',k)$ be the instance returned by Reduction Rule \ref{rr:splitcontraction-reduction} when applied on the instance $(G,k)$. If there exists a set of edges of size at most $k$, say $F'$, such that $G'/F'$ is a split graph then there exists a set of edges $F$ such that $G/F$ is a split graph and cardinality of $F$ is at most $\alpha \cdot |F'|+1$.
\end{lemma}

\begin{proof}
If no vertex in $X \cup Y$ has been deleted, then $G'$ and $G$ are identical graphs, and the statement is true.
We assume that at least one vertex from $X \cup Y$ has been deleted.
Recall that $X', Y'$ is the set of vertices in $X, Y$, respectively, that have been marked.
It is easy to see that $(S, X', Y')$ is a partition of $V(G')$ such that $G' - S$ is a split graph with $(X', Y')$ as one of its split partition.

Let $\calW'$ be a $G'/F'$-witness structure of $G'$.
By Lemma~\ref{lemma:connected-split}, $G'$ is connected and without loss of generality we assume that edges $F'$ satisfies the three properties mentioned in Observation~\ref{obs:bags_in_clique}.
Let $(C', I')$ be the split partition of $G'/F'$ mentioned in the third property in Observation~\ref{obs:bags_in_clique}.
 Let $\calW_{C'}$ (resp. $\calW_{I'}$) be the collection of witness sets in $\calW'$ which correspond to vertices in $C'$ (resp. $I'$). We intentionally name a subset of $\calW'$ with $\calW_{C'}$ (instead of $\calW'_{C'}$) as it simplifies notations in remaining proof. Note that any two witness sets in $\calW_{C'} $ are adjacent with each other in $G'$ and no two witness sets in $\calW_{I'}$ are adjacent with each other in $G'$.
 By Observation~\ref{obs:bags_in_clique}, all non-trivial witness sets in $\calW'$ are contained in $\calW_{C'}$. 
  
 We start constructing witness structure $\calW$ and set a of edges $F$ of $G$ from $\calW'$ and $F'$ as follows. For every vertex $u$ in $(X\setminus X') \cup (Y\setminus Y')$, add singleton witness sets $\{u\}$ to $\calW'$. Initialize $F$ to $F'$. Let $\calW_{C \setminus C'}$ be the set of newly added singleton witness sets that correspond to the vertices in $X \setminus X'$ and let $\calW_{I \setminus I'}$ be the set of newly added singleton witness sets that correspond to the vertices in $Y \setminus Y'$.

 In the remaining proof, we argue that we can carefully add some edges in $F'$ such that following two conditions are satisfied. $(i)$ Any two witness sets in $\calW_{C'} \cup \calW_{C \setminus C'}$ are adjacent with each other in $G$
 $(ii)$ No two witness sets in $\calW_{I'} \cup \calW_{I \setminus I'}$ are adjacent with each other in $G$.
 To ensure condition $(i)$, we might have to add one extra edge for $d$ edges present in $F'$. {\bf This addition of edges introduces the multiplicative factor of $\alpha$ in the upper bound for the size of $F$ in terms of $F'$. To ensure condition $(ii)$, we might have to contract an edge outside $F'$. This brings an additive factor of one. }

 Let $\calW^*$ be the collection of witness sets in $\calW_{C'}$ which violates Condition~$(i)$. In other words, $\calW^*$ is the collection of witness set $W(t)$ in $\calW_{C'}$ such that there exists a (singleton) witness set $\{x\}$ in $\calW_{C \setminus C'}$ which is not adjacent with $W(t)$. See Figure~\ref{fig:solution-lifting-split}. (For example, here witness set $\{x_1\}$ is not adjacent with $W(t)$.) We argue that every witness set in $\calW^*$ has at least $d + 1$ vertices. For the sake of contradiction, assume that there exists a witness set $W(t)$ in $\calW^*$ which contains at most $d$ vertices. Let $\{x_1\}$ be a singleton set in $\calW_{C \setminus C'}$ which is not adjacent with $W(t)$. Since $G[X]$ induces a clique in $G$, witness set $W(t)$ is contained in $S \cup Y'$.
Since Marking Scheme~\ref{mrk-nonnbrs-in-C} iterated over all sets of size at most $d$, it also considered $W(t)$ while marking. Note that $x_1$ is contained in $M_{NX}(W(t))$, a set of non-neighbors of $W(t)$ in $X$. As $x_1$ is unmarked, there are $2k + 2$ vertices in $M_{NX}(W(t))$ that have been marked. All these marked vertices are in $G'$. Since the cardinality of $F'$ is at most $k$, the number of vertices in $V(F')$ is at most $2k$. This implies that at least two marked vertices in $M_{NX}(W(t))$ remain as singleton witness set. Since these two witness sets are adjacent to each other, at least one of these sets in contained in $\calW_{C'}$. This contradicts the fact that any two witness sets in $\calW_{C'}$ are adjacent to each other. Hence our assumption is wrong and $W(t)$ has at least $d+1$ vertices. 

Fix a witness set, say $W(t')$, in $\calW_{C'}$ which intersects with $X'$. By Remark~\ref{remark-extra-coloring}, $X'$ is non-empty and such set exists. We note that $W(t')$ is adjacent with every (singleton) witness set in $\calW_{C \setminus C'}$. For every witness set $W(t)$ in $\calW^*$, we add an edge between $W(t)$ and $W(t')$ to the set $F'$. 
Equivalently, we construct a new witness set by taking the union of $W(t')$ and all the witness sets $W(t)$ in $\calW^*$.

We now argue the size bound of $F$. Note that we have added one extra edge for every witness set $W(t)$ in $\calW^*$. As every witness set of $\calW^*$ has at least $d+1$ vertices, we have added one extra edge for at least $d$ edges in the solution $F$. Moreover, since witness sets in $\calW^*$ are vertex disjoint, no edge in $F'$ can be part of two witness sets. This implies that the number of edges in $F$ is at most $\frac{d+1}{d}|F'| \leq \alpha \cdot |F'|$. 
 
We now consider Condition~$(ii)$. Let $\calW^*$ be the collection of witness sets in $\calW_{I'}$ which violates Condition~$(ii)$. In other words, $\calW^*$ is the collection of witness set $W(t)$ in $\calW_{I'}$ such that there exists a (singleton) witness set $\{y\}$ in $\calW_{I \setminus I'}$ which is adjacent with $W(t)$. 
Since $Y$ is an independent set in $G$, any witness set $W(t)$ in $\calW^*$ intersects with either $X'$ or $S$. We consider two cases depending on whether $W(t)$ intersects $X'$ or not. We argue that in the first case, we can add an extra edge to $F'$ and avoid all such cases, while the second case can not occur.

Consider a witness set $W(t)$ in $I'$ which intersects with $X'$. Let $a$ be the vertex in $I'$ and $W_a = W(t)$ be the witness set corresponding to it. By Remark~\ref{remark-extra-coloring}, there exists vertex $b$ in $G'/F'$ which is adjacent to $a$.
Let $W_b$ be the witness set in $\calW'$ corresponding to $b$. 
Remove $W_b, W_a$ from $\calW'$ and add $W_a \cup W_b$ as a witness set to $\calW'$ (or more specifically to $\calW_{C'}$). Since $X'$ is a clique, at most one vertex from $X'$ can be part of witness set in $\calW_{I'}$. Hence there is at most one such vertex in $I'$. Since edge across $W_a, W_b$ is not in $F'$, this operation adds one extra edge in $F'$. (For example, edge $x_2x_3$ in Figure~\ref{fig:solution-lifting-split}). Hence, in order to make sure that no witness set violates Condition $(i)$ and $(ii)$, we have added edges to $F'$ to obtain $F$ such that the size is at most $\alpha |F'| + 1$.

Now, consider the second case. 
Assume that there exists a witness $W(t)$ in $\calW^*$ which does not intersect with $X'$.
This implies that $W(t)$ is contained in $S$.
Let ${y_1}$ be a singleton witness set in $\calW_{I \setminus I'}$ which is adjacent with $W(t)$.
By Observation~\ref{obs:bags_in_clique}, we know that $W(t)$ is a singleton witness set and is contained in $S$. Hence set $A = W(t)$ has been considered by Marking Scheme~\ref{mrk-nbrs-in-I}.
Note that $y_1$ is contained in $M_{NY}(A)$, a set of neighbors of $A$ in $I$.
Since $y_1$ is unmarked, there are $k + 2$ vertices in $M_{NY}(A)$ that have been marked.
All these marked vertices are in $G'$. 
Since $Y'$ is an independent set in $G'$, at most $k$ vertices in $Y'$ can be incident on solution edges.
Only these vertices can be part of $\calW_{C'}$.
There can be at most one vertex which is a singleton witness set in $\calW_{C'}$.
Hence there exists at least one singleton witness set in $\calW_{I'}$ which is adjacent with $W(t)$.
This contradicts the fact that no two witness sets in $\calW_{I'}$ are adjacent to each other.
Hence our assumption is wrong, and no such witness structure exists in $\calW^*$.
This concludes the proof of the lemma. 
\end{proof}

In the following lemma, we argue that the value of the optimum solution for reduced instance can be upper bounded by the value of the optimum solution for the original instance. 

\begin{lemma} \label{lem-denominator-SC} Let $(G',k)$ be the instance returned by Reduction Rule \ref{rr:splitcontraction-reduction} when applied on an instance $(G,k)$. If $\OPT(G,k) \leq k$ then $\OPT(G',k) \leq 2\cdot \OPT(G,k)$.
\end{lemma}

\begin{proof}
 Let $F$ be a set of at most $k$ edges in $G$ such that $\OPT(G,k)=\SpC(G,k,F)$. Since we are working with a minimization problem, to prove the lemma, it is sufficient to find a solution for $G'$, which is of size at most $2 \cdot |F|$. Recall that $(S,X,Y)$ is a partition of $V(G)$ such that $G-S$ is a split graph with $(X, Y)$ as split partition where $X$ is a clique, and $Y$ is an independent set. 
 The set of vertices that have been marked in $Y$ is denoted by $Y'$, and the set of vertices that have been marked in $X$ is denoted by $X'$.
By our assumption, the input graph $G$ is connected. Without loss of generality, we can assume that $F$ satisfies three properties mentioned in Observation~\ref{obs:bags_in_clique}.
Let $(C, I)$ be the split partition of $G/F$ mentioned in the third property in Observation~\ref{obs:bags_in_clique}.
Let $\calW$ be a $G/F$-witness structure of $G$. Let $\calW_{C}$ (resp. $\calW_{I}$) be the collection of witness sets in $\calW$ which correspond to vertices in $C$ (resp. $I$).
Note that any two witness sets in $\calW_{C} $ are adjacent with each other in $G$, and no two witness sets in $\calW_{I}$ are adjacent with each other in $G$.
By Observation~\ref{obs:bags_in_clique}, all non-trivial witness sets in $\calW$ are contained in $\calW_{C}$. 
  
At each step, we construct the graph $G^*$ from $G$ by deleting one or more vertices in $(X\setminus X')\cup (Y\setminus Y')$.
By Lemma~\ref{lemma:connected-split}, $G^*$ is connected.
{\bf We also construct a set of edges $F^*$ from $F$ by replacing existing edges and/or adding extra edges to $F$.}
In terms of witness sets, we delete witness sets that are subsets of $(X\setminus X')\cup (Y\setminus Y')$ and modify the ones that intersect $(X\setminus X')\cup (Y\setminus Y')$.
After the modification, we represent the witness sets corresponding to vertices in the clique as $\calW^*_{C}$ and the independent set as $\calW^*_{I}$.
At any point, we ensure that any two witness sets in  
$\calW^*_{C}$ are adjacent to each other, and any two witness sets in $\calW^*_{I}$ are not adjacent to each other.
This implies that at any intermediate state, $G^*/F^*$ is a split graph.
We modify $F$ to obtain $F^*$ such that the number of edges in $F^*$ is at most $2 \cdot |F^{\circ}|$, where $F^{\circ}=F$ be an optimum solution for the original instance $(G, k)$.
For notational convenience, we rename $G^*$ to $G$ and $F^*$ to $F$ at regular intervals but do not change $F^{\circ}$. 

Since the class of split graphs is closed under vertex deletion, we can delete all witness sets, which are entirely contained in $(X\setminus X')\cup (Y\setminus Y')$. Suppose that there exists a witness set $W(t)$ in $\calW$ which is a subset of $(X\setminus X')\cup (Y\setminus Y')$, construct $G^*$ from $G$ by deleting witness set $W(t)$ in $\calW$. Delete the edges corresponding to spanning tree of $G[W(t)]$ from $F$ to obtain $F^*$.  
We repeat this process until there exists a witness set that satisfies this condition. After exhaustively applying this process, we have $|F^*|\leq |F|$.

\smallskip

\begin{tcolorbox}[colback=red!5!white,colframe=red!75!black]
\centerline {At this stage we rename $G^\star$ to $G$ and $F^\star$ to $F$.}
\end{tcolorbox}

Note that at this stage, there is no witness set in $\calW_{I}$ which contains vertex in $(X\setminus X')\cup (Y\setminus Y')$. Hence, we do not need to modify witness sets in $\calW_{I}$. In all the conditions mentioned below, the modification is done on non-trivial witness sets in $\calW_{C}$ only. These modifications do not affect the independent property of witness sets in $\calW_{I}$. So, to prove that the modified witness structure obtained corresponds to a split graph, it is enough to show that the witness sets in $\calW^*_C$ are connected, and any two of them are adjacent to each other.

We partition all non-trivial witness sets in $\calW_C$ with respect to their intersection with sets $S, X, Y$. For a non-trivial witness set $W(t)$, we denote its intersection with $S, X, Y$, respectively, using an ordered tuple $(i; j; k)$ where integers $i, j, k$ take $0$ or $1$. If $W(t)$ intersects with $Z$ in $\{S, X, Y\}$, then we assign the corresponding integer to $1$ and $0$ otherwise. Since $(S, X, Y)$ is a partition of $V(G)$, witness sets in $\calW_C$ can be partitioned into seven parts (excluding the trivial $(0; 0; 0)$ case). Note that since $Y$ is an independent set in $G$, no non-trivial witness set can contain only vertices in $Y$. This implies that there is no non-trivial witness set in partition of $\calW_C$ corresponding to $(0; 0; 1)$. Consider a witness set that is entirely contained in $S$. Since we do not delete any vertex in $S$ while creating $G^*$ from $G$, this witness set remains unchanged throughout the process. Hence we do not consider witness sets in partition corresponding to $(1; 0; 0)$.

This implies we only need to modify witness sets in five partitions of $\calW_C$. Each of these partitions can be further divided into subparts based on whether witness set intersects with $X \setminus X'$ or $Y \setminus Y'$ or both or none. We note that any witness set which does not intersect with $(X \setminus X') \cup (Y \setminus Y')$ is not affected; hence we only need to consider first three cases. 
We modify a witness set that satisfies the least indexed condition. If there does not exist any witness set which satisfies either of these conditions, then $(X\setminus X')\cup (Y\setminus Y')$ is an empty set, and the lemma is vacuously true. Hence in every case, we assume that witness set intersects $(X\setminus X')\cup (Y\setminus Y')$,  and it is not entirely contained inside it.

\begin{condition}[ Partition $(0; 1; 0)$] \label{cond:part_010}
 There exists a witness set, say $W(t)$, which intersects with $X$ but does not intersect with $S$ or $Y$.
\end{condition}
\vspace{0.1cm}

Since $W(t)$ is contained in set $X$, it intersects with $X \setminus X'$ but is not contained in it.
Hence both the sets $W(t) \cap X'$ and $W(t)\cap (X \setminus X')$ are non empty. Let $x_1$ be a vertex in $W(t) \cap (X\setminus X')$.  
Fix a witness set, say $W(t^\#)$, in $\calW_{C}$ which is different from $W(t)$ and intersects with $X'$. By Remark~\ref{remark-extra-coloring}, $X'$ is of size at least $2k + 2$ and hence such a set exists. Let $\calW_C^*$ be the witness set obtained from $\calW_C$ by removing $W(t), W(t^\#)$ and adding $(W(t) \setminus \{x_1\}) \cup W(t^\#)$. Note that in $\calW_C, W(t), W(t^\#)$ and $x_1$ satisfy the premise of Observation~\ref{obs:merging-witness-sets}. This implies that any two witness sets in $\calW^*_C$ are adjacent with each other.
Let $F^*$ be the set of edges obtained from $F$ by removing an edge incident on $x_1$ and adding an edge across $W(t)$ and $W(t^\#)$.
Hence, if $G^*$ is obtained from $G$ by deleting $x_1$ then $G^*/F^*$ is a split graph. Since we are deleting at least one edge from $F$ and adding only one edge, we have $|F^*| \leq |F|$.

\begin{figure}[t]
 \centering
 \includegraphics[scale=0.60]{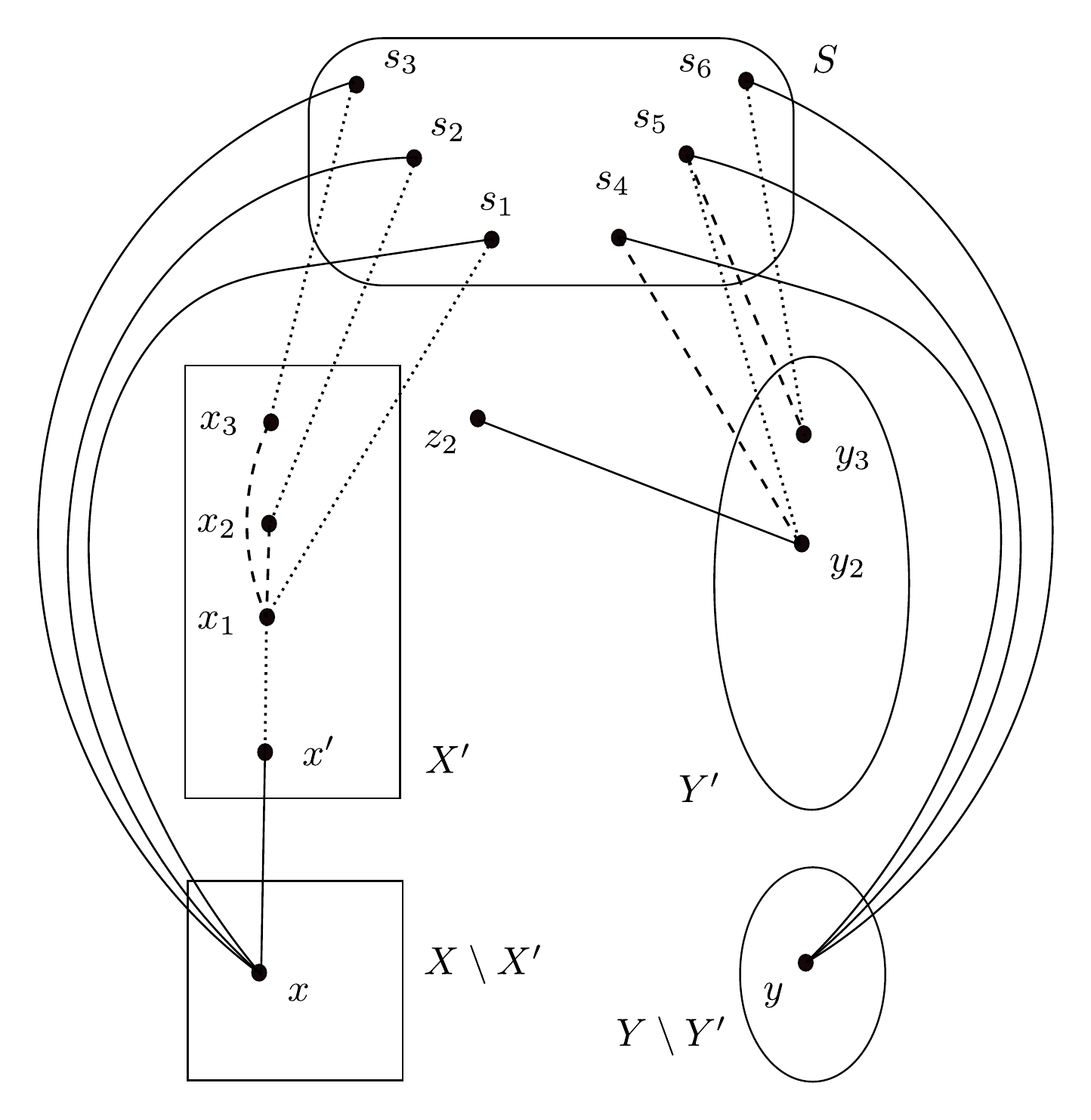}
 \caption{Straight lines (e.g. $xx'$) represent edges in original solution $F$. Dotted lines (e.g. $x_1x'$) represents an edges replaced in $F$. Dashed lines (e.g. $x_1x_2$) represents extra edges added to $F$. Please refer to Conditions~\ref{cond:part_110} and \ref{cond:part_101} in Lemma~\ref{lem-denominator-SC}. Vertex $z_2$ can be in $X$ or $S$. \label{fig:condition-2-3}}
\end{figure}

\smallskip

\begin{tcolorbox}[colback=red!5!white,colframe=red!75!black]
\centerline {At this stage we rename $G^\star$ to $G$ and $F^\star$ to $F$.}
\end{tcolorbox}

\begin{condition}[ Partition $(1; 1; 0)$] \label{cond:part_110} There exists a witness set, say $W(t)$, which intersects with $S$ and $X$ but does not intersect with $Y$.
\end{condition}
\vspace{0.1cm}
Since $W(t)$ does not intersect with $Y$, it intersects with $X \setminus X'$.
Let $x$ be a vertex in $W(t)\cap (X\setminus X')$; $S_t$ be the set of vertices in $W(t) \cap S$ which are adjacent with $x$ via edges in $F$ and $Q_t$ be the set of vertices in $W(t)\cap X$ which are adjacent with $x$ via edges of $F$.
We find a \emph{substitute} for $x$ in $X'$.
(This condition is same as that of Condition~\ref{cond:witness-intersectsX} in the proof of Lemma~\ref{lemma:reduction-rule}.)
Note that if $S_t$ is empty then $Q_t$ must be non-empty and we replace the edges incident on $x$ similar to that of in Condition \ref{cond:part_010}.
Every vertex in $S_t$ is considered by Marking Scheme~\ref{mrk-cmn-nbr-in-C}.
For every $s$ in $S_t$, vertex $x$ is contained in $M_{CX}(\{s\})$, the set of vertices in $X$ whose neighborhood contains $\{s\}$.
Since $x$ is in $X \setminus X'$, and hence unmarked, there are two vertices in $M_{CX}(\{s\})$, say $x^1_s, x^2_s$, different from $x$ which have been marked for each $s$ in $S_t$. 
Since $X$ is a clique in $G$, at most one of these two vertices are part of witness set contained in $\calW_{I}$. 
Without loss of generality, assume that for every vertex $s$ in $S_t$, vertex $x^1_s$ is contained in witness set contained in $\calW_C$.
We construct $F^*$ from $F$ by following operation: Arbitrarily fix a vertex $s_0$ in $S_t$.
Remove the edge $s_ox$ in $F$ and add the edge $s_ox^1_{s_o}$ and for every vertex $s$ in $S_t\setminus \{s_o\}$, remove the edge $sx$ in $F$ and add two edges $sx^1_s, x^1_sx^1_{s_o}$.
For every vertex $x'$ in $Q_t$ replace the edge $x'x$ by $x'x^1_{s_0}$.
(For example, in Figure~\ref{fig:condition-2-3}, let $W(t) = \{x, x', s_1, s_2, s_3\}$ and $x$ is in $X \setminus X'$.
This implies that $S_t = \{s_1, s_2, s_3\}$ and $Q_t = \{x'\}$.
Let $s_1$ be the vertex which is fixed arbitrarily.
In this process, edge $xs_3$ is deleted and two edges $x_3s_3, x_3x_1$ are added to $F$.)

We argue that contracting edges in $F^*$ merges $W(t)\setminus \{x\}$ into the witness set in $\calW^*_C$ that contains $x_s^1$.
Define $W := (W(t) \setminus \{x\}) \cup (\bigcup_{s \in S_t}W(x^1_s))$ where $W(x^1_s)$ is the witness set in $\calW$ containing the vertex $x^1_s$.
Let $\calW_C^*$ be the witness structure obtained from $\calW_C$ by removing every witness set which intersects $W$ and adding $W$.
Since $W$ contains $W(x^1_{s_0})$, it is adjacent with every other witness set in $\calW_C$ and hence in $\calW^*_C$.
This implies that any two witness sets in $\calW^*_C$ are adjacent to each other. 
Hence, if $G^*$ is obtained from $G$ by deleting $x$, then $G^*/F^*$ is a split graph.
Since we are adding at most two edges only for deleted edges in $F$, we have $|F^*| \leq 2 \cdot |F|$.

\smallskip

\begin{tcolorbox}[colback=red!5!white,colframe=red!75!black]
\centerline {At this stage we rename $G^\star$ to $G$ and $F^\star$ to $F$.}
\end{tcolorbox}

\begin{condition}[ Partition $(1; 0; 1)$] \label{cond:part_101} There exists a witness set, say $W(t)$, which intersects with $S$ and $Y$ but does not intersect with $X$.
\end{condition}
\vspace{0.1cm}

Since $W(t)$ does not intersect with $X$, it intersects with $Y \setminus Y'$. Let $y$ be a vertex in $W(t)\cap (Y\setminus Y')$; $S_t$ be the set of vertices in $W(t) \cap S$ which are adjacent with $y$ via edges in $F$. After removing all the edges incident on $y$ we add some edges to ensure connectivity of vertices in $S_t$ and some other to ensure adjacency among witness sets in $\calW_C$. Arbitrarily fix a vertex $s_0$ in $S_t$.
If $S_t \setminus \{s_0\}$ is an empty set then no edge needs to be included to ensure connectivity. For every vertex $s$ in $S_t \setminus \{s_0\}$, set $\{s_0, s\}$ is considered by Marking Scheme~\ref{mrk-cmn-nbrs-in-I}. (For any $\alpha > 1$, $d$ is at least two.)
For every $s$ in $S_t \setminus \{s_0\}$, vertex $y$ is contained in $M_{CY}(\{s_0, s\})$, the set of common neighbors of $\{s_0, s\}$ in $Y$. Since $y$ is in $Y \setminus Y'$, and hence unmarked, there is a vertex in $M_{CY}(\{s_0, s\})$ different from $y$ which has been marked for each set $\{s_0, s\}$. For every vertex $s$ in $S_t \setminus \{s_0\}$, let $y_s$ be the marked vertex which is adjacent with $s_0$ and $s$. We construct $F^*$ from $F$ by following operation: Remove the edge $s_oy$ and for every vertex $s$ in $S_t \setminus \{s_o\}$, remove the edge $sy$ in $F$ and add two edges $sy_s, s_oy_s$.
(For example, in Figure~\ref{fig:condition-2-3}, let $W(t) = \{y, s_4, s_5, s_6\}$ and $y$ is in $Y \setminus Y'$. This implies that $S_t = \{s_4, s_5, s_6\}$. Let $s_4$ be the vertex which is fixed arbitrarily. In this process, edges $s_4y, s_5y$ are deleted and two edge $s_5y_2, s_4y_2$ are added to $F$.) Now we include additional edges to ensure adjacency among witness sets in $\calW_C$. 
 Towards this we first prove that there always exists a witness set $W(t^\#)$ different from $W(t)$, in $\calW_{C}$ that is adjacent with either $W(t)\setminus \{y\}$ or with $y_s$ for some $s$ in $S_t\setminus \{s_o\}$. Suppose that there is no set $W(t^\#)$ that is adjacent with $W(t)\setminus \{y\}$ this implies that every witness set in $\calW$ is adjacent with only vertex $y$ in $W(t)$. As the size of $X'$ is at least $2k + 2$ and $y$ is adjacent with every witness set in $\calW \setminus \{W(t)\}$ and every vertex in $S_t$, $y$ is adjacent with at least $k + 2$ vertices.
For a vertex $s'$ in $S_t\setminus \{s_o\}$, let $y_{s'}$ be the vertex marked by Marking Scheme~\ref{mrk-cmn-nbrs-in-I} while considering set $\{s_o,s'\}$. Note that Marking Scheme~\ref{mrk-cmn-nbrs-in-I} preferred $y_{s'}$ over $y$. This implies that $y_{s'}$ is adjacent with at least $k + 2$ vertices. Hence, $y_{s'}$ has at least one neighbor outside $W(t)$. Note that $y_{s'}$ is not in the clique because $W(t)\setminus \{y\}$ did not have a neighbor in $\calW_{C}$, this implies that the neighbor of $y_{s'}$ is in $\calW_{C}$. Since there always exists another witness set $W(t^\#)$ that is adjacent with $(W(t)\setminus \{y\}) \cup \{y_{s} |\ \text{for every}\; s\; \text{in}\; S_t\setminus \{s_o\}\}$, we add the edge across $W(t), W(t^\#)$ in $F$.

Let $W$ be the superset of $W(t)\setminus \{y\}$ which contains witness sets of all the newly added vertices to $W(t)$. Formally, $W=(W(t)\setminus \{y\})\cup(\bigcup_{s\in S_t\setminus \{s_o\}}W(y_s))$  where $W(y_s)$ is the witness set containing $y_s$. Note that every vertex in $S_t$ is connected with each other (as we have added edges $y_ss, y_ss_0$ for every $s$ in $S_t \setminus \{s_o\}$) and $W$ is adjacent with $W(t^\#)$. 
Let $\calW_C^*$ be the witness structure obtained from $\calW_C$ by removing every witness set which intersects with $W, W(t^\#)$ and adding $W \cup W(t^\#)$.
Since $W \cup W(t^\#)$ contains $W(t^\#)$, which is adjacent with every other witness set in $\calW_C$ and hence in $\calW^*_C$. This implies that any two witness sets in $\calW^*_C$ are adjacent with each other. Hence, if $G^*$ is obtained from $G$ by deleting $y$ then $G^*/F^*$ is a split graph. Since we are adding two edges of the form $y_ss, y_ss_0$ for every edge of the form $sy$ in $F$, where $s \in S_t \setminus \{s_o\}$ and one edge across $W(t), W(t^\#)$ instead of $ys_o$ we have $|F^*| \leq 2 \cdot |F|$.

\smallskip

\begin{tcolorbox}[colback=red!5!white,colframe=red!75!black]
\centerline {At this stage we rename $G^\star$ to $G$ and $F^\star$ to $F$.}
\end{tcolorbox}

\begin{condition}[ Partition $(0; 1; 1)$] \label{cond:part_011} There exists a witness set, say $W(t)$, which intersects with $X$ and $Y$ but does not intersect with $S$.
\end{condition}
\vspace{0.1cm}

We construct $F^*$ from $F$ in two steps, where the first step deletes the vertices of $Y\setminus Y'$, and the second step deletes the vertices of $X\setminus X'$.
Suppose that there exists a vertex $y$ in $W(t)\cap (Y\setminus Y')$. 
Fix a witness set, say $W(t^\#)$, in $\calW_{C}$ which is different from $W(t)$ and intersects $X'$.
By Remark~\ref{remark-extra-coloring}, $X'$ is of size at least $2k + 2$ and hence such set exists.
Let $\calW_C^*$ be the witness set obtained from $\calW_C$ by removing $W(t), W(t^\#)$ and adding $W(t^\#) \cup (W(t) \setminus \{y\})$.
Note that in $\calW_C$, witness sets $W(t), W(t^\#)$ and $y$ satisfy the premise of Observation~\ref{obs:merging-witness-sets}. This implies that any two witness sets in $\calW^*_C$ are adjacent with each other.
Hence, if $G^*$ is obtained from $G$ by deleting $y$ then $G^*/F^*$ is a split graph. Since we are deleting at least one edge from $F$ and adding only one edge, we have $|F^*| \leq |F|$.
We repeat this process as long as there exists a witness set that satisfy Condition~\ref{cond:part_011} and intersects $Y \setminus Y'$. 

Consider a witness set $W(t)$ which satisfies Condition~\ref{cond:part_011} and does not intersect with $Y \setminus Y'$. This implies $W(t)$ intersects with $X \setminus X'$ and every vertex in $W(t) \cap Y$ is contained in $Y'$.
Consider a vertex $x$ in $W(t) \cap (X \setminus X')$. If $F$ does not contain any edges across $\{x\}$ and $ Y'$ then in $\calW_C$, witness sets $W(t), W(t^\#)$ and $x$ satisfy the premise of Observation~\ref{obs:merging-witness-sets}. 
We obtain the modified graph as mentioned in previous case.
Suppose that there exist edges in $F$ across $\{x\}$ and $ Y'$.
Let $Y_t$ be the set of neighbors of $x$ in $Y'$ via edges in $F$.
For every $y$ in $Y_t$, the set $\{y\}$ is considered by Marking Scheme~\ref{mrk-cmn-nbr-in-C}.
Note that $x$ is contained in the set $M_{CX}(\{y\})$, the set of vertices in $X$ whose neighborhood contains $\{y\}$.
Since $x$ is in $X \setminus X'$, and hence unmarked, there are two vertices in $M_{CX}(\{y\})$, say $x^1_y, x^2_y$, different from $x$ which have been marked for each $y$ in $Y_t$.
Since $X$ is a clique in $G$, at most one of these two vertices are part of witness set contained in $\calW_{I}$. 
Without loss of generality, assume that for every vertex $y$ in $Y_t$, vertex $x^1_y$ is contained in the witness set in $\calW_C$. 
We construct $F^*$ from $F$ by following operation (This step is similar modification as we did in Condition~\ref{cond:part_110}): Arbitrarily fix a vertex $y_0$ in $Y_t$. For every vertex $y$ in $Y_t$, remove the edge $yx$ in $F$ and add two edges $yx^1_y, x^1_yx^1_{y_0}$. 
Let $W = (W(t) \setminus \{x\}) \cup (\bigcup_{y \in Y_t}W(x^1_y))$  where $W(x^1_y)$ is the witness set in $\calW$ containing $x^1_y$.
We obtain the witness structure $\calW^*_C$ from $\calW_C$ by removing every witness set that intersects $W$ and adding $W$.
Since $W$ contains $W(x^1_{y_0})$, it is adjacent with every other witness set in $\calW_C$ and hence in $\calW^*_C$. This implies that any two witness sets in $\calW^*_C$ are adjacent to each other. Hence, if $G^*$ is obtained from $G$ by deleting $x,$ then $G^*/F^*$ is a split graph. Since we are adding at most two edges $yx^1_y,x^1_yx^1_{y_o} $ for one deleted edge of the form $yx$ in $F$, we have $|F^*| \leq 2 \cdot |F|$.

\smallskip

\begin{tcolorbox}[colback=red!5!white,colframe=red!75!black]
\centerline {At this stage we rename $G^\star$ to $G$ and $F^\star$ to $F$.}
\end{tcolorbox}

\begin{condition}[ Partition $(1; 1; 1)$] \label{cond:part_111} There exists a witness set, say $W(t)$, which intersects with $S, X$ and $Y$.
\end{condition}
\vspace{0.1cm}



We further divide this condition based on whether the intersection of $W(t)$ with $(Y\setminus Y')$ is empty (Condition~\ref{cond:part_111}$(a)$) or not (Condition~\ref{cond:part_111}$(b)$) 

\noindent\emph{Condition~\ref{cond:part_111}(a):} Consider that $W(t)\cap (Y\setminus Y')$ is empty. Then there is at least one vertex in $W(t)\cap (X\setminus X')$. In this case, construction of $G^*$ and $F^*$ is similar to that of Condition~\ref{cond:part_110}. Instead of considering a subset of $S$ in Condition~\ref{cond:part_110}, we consider a subset of $S \cup Y'$. 
Let $x$ be a vertex in $W(t)\cap (X\setminus X')$; $Z_t$ be the set of vertices in $W(t) \cap (S \cup Y')$ which are adjacent with $x$ via edges in $F$ and $Q_t$ be the set of vertices in $W(t) \cap X'$ which are adjacent with $x$ via edges of $F$.
Note that every vertex in $Z_t$ is considered by the Marking Scheme~\ref{mrk-cmn-nbr-in-C}. For every vertex $z$ in $Z_t$, vertex $x$ is contained in $M_{CX}(\{z\})$, the set of vertices in $X$ whose neighborhood contains $\{z\}$. Since $x$ is in $X \setminus X'$, and hence unmarked, there are two vertices in $M_{CX}(\{z\})$, say $x^1_z, x^2_z$, different from $x$ which have been marked for each $z$ in $Z_t$. Since $X$ is a clique in $G$, at most one of these two vertices are part of witness set contained in $\calW_{I}$. Without loss of generality, assume that for every vertex $z$ in $Z_t$, vertex $x^1_z$ is contained in witness set of $\calW_C$. We construct $F^*$ from $F$ by following operation: Arbitrarily fix a vertex $z_0$ in $Z_t$. remove the edge $z_0x$ and for every vertex $z$ in $Z_t\setminus \{z_0\}$, remove the edge $zx$ in $F$ and add two edges $zx^1_z, x^1_zx^1_{z_0}$. For every vertex $x'$ in $Q_t$ replace the edge $x'x$ by $x'x^1_{z_0}$.

Let $W$ be the superset of $W(t) \setminus \{x\}$ which contains witness set containing all newly added vertices in $W(t)$. Formally, $W = (W(t) \setminus \{x\}) \cup (\bigcup_{z \in Z_t}W(x^1_z))$ where $W(x^1_z)$ is the witness set in $\calW$ containing $x^1_z$.
Let $\calW^*_C$ be the witness structure obtained from $\calW_C$ by removing every witness set which intersects $W$ and adding $W$.
Since $W$ contains $W(x^1_{z_0})$, it is adjacent with every other witness set in $\calW_C$ and hence in $\calW^*_C$. This implies that any two witness sets in $\calW^*_C$ are adjacent to each other. 
Hence, if $G^*$ is obtained from $G$ by deleting $x$, then $G^*/F^*$ is a split graph. Since we are adding at most two edges only for deleted edges in $F$, we have $|F^*| \leq 2 \cdot |F|$.

\smallskip

\begin{tcolorbox}[colback=red!5!white,colframe=red!75!black]
\centerline {At this stage we rename $G^\star$ to $G$ and $F^\star$ to $F$.}
\end{tcolorbox}

\noindent\emph{Condition~\ref{cond:part_111}$(b)$:} Consider that $W(t)\cap (Y\setminus Y')$ is non-empty. In this case, construction of $G^*$ and $F^*$ is similar to that of Condition~\ref{cond:part_101}.
Let $y$ be a vertex in $W(t)\cap(Y\setminus Y')$ and $S_t$ be the set of vertices in $W(t)\cap S$ which are adjacent with $y$ via edges in $F$ and $Q_t$ be the set of vertices in $W(t)\cap X$ which are adjacent with $y$ via edges in $F$.
Arbitrarily fix a vertex $s_0$ in $S_t$.
We add some edges to ensure connectivity of vertices in $S_t$ and some other to ensure adjacency among witness sets in $\calW_C$.
If $S_t \setminus \{s_o\}$ is an empty set then no edge needs to be included to ensure connectivity. For every vertex $s$ in $S_t \setminus \{s_o\}$, the set $\{s_o, s\}$ is considered by Marking Scheme~\ref{mrk-cmn-nbrs-in-I}. (For any $\alpha > 1$, $d$ is at least two.)
For every $s$ in $S_t \setminus \{s_o\}$, vertex $y$ is contained in $M_{CY}(\{s_o, s\})$, the set of common neighbors of $\{s_o, s\}$ in $Y$.
Since $y$ is in $Y \setminus Y'$, and hence unmarked, there is a vertex in $M_{CY}(\{s_o, s\})$ different from $y$ which has been marked for each set $\{s_o, s\}$.
For every vertex $s$ in $S_t \setminus \{s_o\}$, let $y_s$ be the marked vertex which is adjacent with $s_o$ and $s$. 
We construct $F^*$ from $F$ by following operation: Remove the edge $s_oy$ and for every vertex $s$ in $S_t \setminus \{s_o\}$, remove edge $sy$ in $F$ and add two edges $sy_s, s_oy_s$.
Since $W(t)$ also intersects $X$, there exists another witness set $W(t^\#)$, in $\calW_{C}$ such there is an edge across $W(t^\#)$ and $W(t) \setminus \{y\}$.
Let $uv$ be that edge. Include $uv$ in the solution.

Let $W$ be the superset of $W(t) \setminus \{y\}$ which contains the witness sets of all the newly added vertices in $W(t)$. Formally, $W = (W(t) \setminus \{y\}) \cup (\bigcup_{s \in S_t}W(y_s))$ where $W(y_s)$ is the witness set containing $y_s$. Note that every vertex in $S_t$ is connected with each other (as we have added edges $y_ss, y_ss_o$ for every $s$ in $S_t \setminus \{s_o\}$) and $W$ is adjacent with $W(t^\#)$. Let $\calW^*_C$ be the witness structure obtained from $\calW_C$ by removing every witness set which intersects with $W,W(t^\#)$ and adding $W \cup W(t^\#)$. Since $W \cup W(t^\#)$ contains $ W(t^\#)$, which is adjacent with every other witness set in $\calW_C$ and hence in $\calW'_C$. This implies that any two witness sets in $\calW^*_C$ are adjacent with each other. Hence, if $G^*$ is obtained from $G$ by deleting $y$ then $G^*/F^*$ is a split graph. Since we are adding at most two edges of the form $y_ss, y_ss_o$ for every edge of the form $sy$ in $F$ where $y \in S_t \setminus \{s_o\}$ and at most two edges for edge $ys_o$, we have $|F^*| \leq 2 \cdot |F|$.

\smallskip

We now argue that even after repeating the process, the size of $|F^*|$ is bounded. In every condition, an edge $uv$ is replaced only if one of the endpoints belongs to $(X\setminus X')\cup (Y\setminus Y')$. Every time edges were replaced by those edges that belong to $G'$. So once any of the conditions mentioned above consider an edge in $F$, it will  never be considered by any other condition. This implies that the number of edges in $F^*$ is always upper bounded by $2 \cdot |F^{\circ}|$ where $F^{\circ}$ is a solution for original instance $(G, k)$. If there exists no witness set {\bf that satisfies any condition}, then $G' = G^*$ and solution $F^*$ is the solution with desired properties. This concludes the proof of the lemma.
\end{proof}

Lemma~\ref{lem-numerator-SC}~and~\ref{lem-denominator-SC} are not sufficient to prove a time-efficient $(2 + \epsilon)$-approximate polynomial kernel. This is primarily because of the additive factor in Lemma~\ref{lem-numerator-SC}. We present the following lemma, which describes a solution lifting algorithm whose running time is dependent on $\epsilon$.

\begin{lemma}\label{lemma:solution-lifting-split} For a fixed $\alpha^{\prime}, \alpha$ satisfying $\alpha^{\prime} > \alpha > 1$, there exists a solution lifting algorithm which satisfies following properties.
 \begin{enumerate}
 \item Given a reduced instance $(G', k)$ obtained by applying Reduction Rule~\ref{rr:splitcontraction-reduction} on instance $(G, k)$, partition $(S, X, Y)$ of $V(G)$ and an integer $d = \lceil \frac{1}{\alpha - 1}\rceil$ together with a solution $F'$ for $(G', k)$, it returns a solution $F$ for $(G, k)$ such that $\frac{\SpC(G,k,F)}{\OPT(G,k)} \leq 2 \cdot \alpha^{\prime} \cdot \frac{\SpC(G',k,F')}{\OPT(G',k)}$.
 \item The algorithm runs in time $\mathcal{O}(m^{{\alpha} \cdot c + 2})$ where $m$ is the number of edges in $G$ and $c = \frac{1}{\alpha' - \alpha}$.
 \end{enumerate}
\end{lemma}
\begin{proof} We present a solution lifting algorithm which considers three cases depending on cardinality of $F'$. Before mentioning the algorithm, we recall Lemma~\ref{lem-denominator-SC} which states that if $\OPT(G,k) \leq k$ then $\OPT(G',k) \leq 2 \cdot \OPT(G,k)$. If $\OPT(G, k) = k + 1$ then $\OPT(G', k) \le k + 1 = \OPT(G, k)$. Hence in either case, $\OPT(G', k) \le 2 \cdot \OPT(G, k)$.

 If cardinality of $F'$ is greater than or equal to $k + 1$ then solution lifting algorithm returns a spanning tree $F$ of $G$ (a trivial solution) as solution for $(G, k)$. In this case, $\SpC(G, k, F) = k + 1 = \SpC(G', k, F')$. Since $\OPT(G', k) \le 2 \cdot \OPT(G, k)$, in this case we get  
$\frac{\SpC(G,k,F)}{\OPT(G,k)} \leq 2 \cdot \frac{\SpC(G',k,F')}{\OPT(G',k)}$.
  
If cardinality of $F'$ is at most $k$ but greater than or equal to $c$ then the algorithm uses Lemma~\ref{lem-numerator-SC} to compute a solution $F$ for $(G, k)$ such that cardinality of $F$ is at most $\alpha \cdot |F'| + 1$. Since $\OPT(G', k) \le 2 \cdot \OPT(G, k)$, this implies:
$$\frac{|F|}{\OPT(G,k)} \leq \frac{ 2 \cdot (\alpha \cdot |F'| + 1)}{\OPT(G',k)} \le 2 \cdot \Big(\alpha + \frac{1}{|F'|}\Big) \cdot \frac{ |F'|}{\OPT(G',k)} \le 2 \cdot \alpha^{\prime} \cdot \frac{ |F'|}{\OPT(G',k)}$$
Last inequality follows from the fact that $|F'| \ge c$ and $1/c = \alpha' - \alpha$. Hence in this case, $\frac{\SpC(G,k,F)}{\OPT(G,k)} \leq 2 \cdot \alpha^{\prime} \cdot \frac{\SpC(G',k,F')}{\OPT(G',k)}$.

Consider the remaining case when $|F'| < c$. By definition, $F'$ is a valid solution for $(G', k)$. By Lemma~\ref{lem-numerator-SC}, there exists a solution $F$ for $(G, k)$ such that $|F| \le \alpha |F'|+1 < \alpha \cdot c+1$. Since we are working with minimization problem, this implies $\OPT(G,k) \le \alpha \cdot c$. In this case, solution lifting algorithm computes an optimum solution for $(G, k)$ by brute force, i.e. checking all set of edges in $G$ of size at most $\alpha \cdot c$, and returns it. In this case, we have, $\frac{\SpC(G,k,F)}{\OPT(G,k)} = 1$.

Hence, the solution lifting algorithm returns a solution with the desired property. The running time of the algorithm follows from its description, and the fact that solution satisfying Lemma~\ref{lem-numerator-SC} can be obtained in polynomial time.
\end{proof}

We note that solution lifting algorithm mentioned in Lemma~\ref{lemma:solution-lifting-split} allows us to choose a value of $\alpha$ between $1$ and $\alpha'$. This choice of $\alpha$ is a trade-off between the running time of the algorithm (as $c$ is inversely proportional to $\alpha' - \alpha$) and the size of lossy kernel (as $d$ is inversely proportional to $\alpha - 1$). We now present the main result of this section.

\begin{reptheorem}{thm:split-lossy} For any $\epsilon > 0$, \textsc{Split Contraction} parameterized by the size of solution $k$, admits $(2 + \epsilon)$-approximate polynomial kernelization algorithm which runs in $\mathcal{O}(m^{\alpha \cdot c + 2})$ time and return an instance with $\calO(k^{d^2+d+1})$ vertices. Here, $m$ is number of edges in $G$ and constants $\alpha, c,$ and $d$ depend only on $\epsilon$.
\end{reptheorem}

\begin{proof} Any split graph can have at most one component which contains an edge. Hence, if graph $G$ in given instance $(G, k)$ is not connected, then all but one connected component needs to be contracted to an isolated vertex.
If there are more than one connected components with at least $k + 1$ vertices, then we return a trivial instance. Otherwise, we apply reduction rules to the unique connected component in $G$, which has more than $k + 1$ vertices. Hence, we can assume that the input graph is connected.  
 
 For a given instance $(G,k)$, a kernelization algorithm applies Reduction Rule~\ref{rr:large-splitvd}. If it returns a trivial instance then the statement is vacuously true. If it does not return a trivial instance then the algorithm partitions $V(G)$ in three sets $(S,X,Y)$ such that $|S|\leq 10k$ and $G - S$ is a split graph with $(X, Y)$ as its split partition. For given $\epsilon$, fix $\alpha' = 1 + \epsilon/2$. The algorithm fixes $\alpha$ which is strictly more than $1$ and strictly less than $\alpha'$. The algorithm then applies Reduction Rule \ref{rr:splitcontraction-reduction} on $(G,k)$ with partition $(S, X, Y)$ and $\alpha$. The algorithm returns the reduced instance as an $(2 \cdot \alpha')$-lossy kernel for $(G,k)$.

The correctness of the algorithms follows from Lemma~\ref{lem-denominator-SC}~and~\ref{lemma:solution-lifting-split} combined with the fact that Reduction Rule~\ref{rr:splitcontraction-reduction} is applied at most once. 
By Observation~\ref{obsn:approx-svd-solution}, Reduction Rule~\ref{rr:large-splitvd} can be applied in polynomial time. The size of the instance returned by Reduction Rule~\ref{rr:splitcontraction-reduction} is at most $\calO(k^{d^2 + d + 1})$. Reduction Rule~\ref{rr:splitcontraction-reduction} can be applied in time $n^{\calO(1)}$ if number of the vertices in $G$ is more than $\calO(k^{d^2 + d + 1})$.
\end{proof}
\section{Lower Bound on Inapproximability of \splitcontract}
\label{hardness_splt}
In this section, we show that for any $\delta>0$, \splitcontract parameterized by the solution size does not admit a factor $(5/4-\delta)$-\fpt approximation algorithm, assuming \textsf{Gap-ETH}. Towards this, we give a polynomial time reduction from \colordensecliqueperfect, which is $k^{o(1)}$-factor \fpt inapproximable (Corollary~\ref{lem:Gap-ETH result densest-color}). In the \denseclique\ problem, input is a graph $G$ and an integer $k$; and the goal is to find a subset $S$ of $V(G)$ of size $k$  such that $\textsf{Den}(S)\geq k^{-g(k)}$. Recall that we defined $\textsf{Den}(S)= |E(G[S])|/\binom{|S|}{2}$. \densecliqueperfect is a special case of \denseclique\ problem, in which the input graph $G$ contains a $k$-clique, and it is known to be $k^{o(1)}$-factor \fpt inapproximable~\cite{DBLP:journals/corr/abs-1708-04218}.  
Now, we define a colorful version of \densecliqueperfect. We call a set of edges  \emph{colorful}, if all the edges in the set are colored with pairwise distinct colors. We say that a clique is \emph{colorful-clique}, if the set of edges in this clique is colorful.   
In the \colordensecliqueperfect, an input graph $G$ is given with an edge coloring $\phi : E(G) \rightarrow [\binom{k}{2}]$, and $G$ is promised to have a colorful-$k$-clique. The goal is to find a subset $S$ of $V(G)$ of size $k$ such that the set $E(G[S])$ is colorful and $\textsf{Den}(S)\geq k^{-g(k)}$.

We start with the following known inapproximability result for \densecliqueperfect.

\begin{proposition}[Lemma $5.21$ in \cite{DBLP:journals/corr/abs-1708-04218}] \label{Gap-ETH result densest} Assuming {\rm \textsf{Gap-ETH}}, for every function $g=o(1)$ and every function $f$, there is no $f(k)\cdot n^{\calO(1)}$-time algorithm such that, given an integer $k$ and a graph $G$ with $n$ vertices containing a $k$-clique, always outputs a set $S$ of $k$ vertices such that $\textsf{Den}(S)\geq k^{-g(k)}$.
\end{proposition}

For notational convenience, let $t=\binom{k}{2}$. The size of $(|E(G)|, t)$-perfect hash family is bounded by $e^tt^{\calO(\log t)} \cdot n^{\calO(1)}$ \cite{naor1995splitters}. Using $(|E(G)|, t)$-perfect hash family and Proposition~\ref{Gap-ETH result densest}, we obtain the following result for \colordensecliqueperfect.

\begin{corollary}\label{lem:Gap-ETH result densest-color}
Assuming {\rm \textsf{Gap-ETH}}, for every function $g=o(1)$ and every function $f$, there is no $f(k)\cdot n^{\calO(1)}$-time algorithm such that, given an integer $k$ and an edge colored graph $G$ with $n$ vertices containing a colorful-$k$-clique, always outputs a set $S$ of $k$ vertices such that $E(G[S])$ is colorful and $\textsf{Den}(S)\geq k^{-g(k)}$.
\end{corollary}

For a given subset $S$ of vertices, we say that $S$ \emph{spans} an edge, if  both of its endpoints are in $S$. 
Due to Corollary~\ref{lem:Gap-ETH result densest-color}, we obtain following result.
\begin{corollary}\label{cor:Gap-ETH result densest}
Assuming {\rm \textsf{Gap-ETH}}, for every $0<\epsilon <1 $, and for every function $f$, there is no $f(k)\cdot n^{\calO(1)}$-time algorithm such that, given an integer $k$ and an edge colored graph $G$ with $n$ vertices containing a colorful-$k$-clique, always outputs a set $S$ of $k$ vertices which span at least $\epsilon \binom{k}{2}$ colorful edges.
\end{corollary}

In the following lemma, we strengthen the above result. 

\begin{lemma}\label{lemma:original-problem}
Assuming {\rm \textsf{Gap-ETH}}, for every $0<\epsilon <1 $, $\alpha>1$ and every function $f$, there is no $f(k)\cdot n^{\calO(1)}$-time algorithm such that, given an integer $k$ and an edge colored graph $G$ with $n$ vertices containing a colorful-$k$-clique, always outputs a set $S$ of at most $\alpha k$ vertices that spans at least $\epsilon \binom{k}{2}$ colorful edges.
\end{lemma}
\begin{proof}
Assume that there is a $f(k) \cdot n^{\calO(1)}$-time algorithm, $\calA$, that takes the input $(G, k)$ such that $G$ has a colorful-$k$-clique, and outputs a set $S$ of at most $\alpha k$ vertices, that spans at least $\epsilon \binom{k}{2}$ colorful edges. Partition the set $S$ into $\lceil 2\alpha \rceil$ sets $S_1,S_2,\cdots, S_{\lceil 2\alpha \rceil}$ each of size at most $\nicefrac{k}{2}$. As the set $S$ spans at least $\epsilon \binom{k}{2}$ colorful edges, there exists a pair of sets, say $S_i, S_j$, in this partition, such that $S_i\cup S_j$ spans at least $\epsilon'\binom{k}{2}$ colorful edges, where $\epsilon' = \nicefrac{\epsilon}{\binom{\lceil 2\alpha \rceil}{2}}$. Moreover, the size of $S_i\cup S_j$ is at most $k$. This implies that there exists an algorithm, that outputs a set of vertices of size at most $k$ that spans $\epsilon' \binom{k}{2}$ colorful edges, in time $f(k) \cdot n^{\calO(1)}$.  
 Hence, the existence of such algorithm contradicts Corollary~\ref{cor:Gap-ETH result densest}. 
\end{proof}

 Now, we are ready to give our reduction.
We argue that if \splitcontract parameterized by the solution size admits a factor $(\nicefrac{5}{4}-\delta)$-\fpt approximation algorithm for some $\delta$, then it contradicts Lemma~\ref{lemma:original-problem}, 
and hence \textsf{Gap-ETH}. Towards this, we present a reduction in which given an instance $(G,k)$ of \colordensecliqueperfect, and a constant $\delta$, constructs an instance $(G',k')$ of \splitcontract.

\vspace{0.3cm}
\noindent \textbf{Reduction Algorithm:} Given an instance $(G,k)$ of \colordensecliqueperfect,  
and a constant $\delta > 0$, the algorithm constructs a graph $G'$ as follows. Recall that $t=\binom{k}{2}$.
\begin{itemize}
\setlength{\itemsep}{-2pt}
\item Fix $\rho=\lceil \nicefrac{\delta t}{k}\rceil$, $k' = 2t + \rho k$, and $k^{\circ}=\lceil \nicefrac{5}{2}\cdot k'\rceil+2$. 
\item For every vertex $u$ in $V(G)$, add $\rho$ copies of it to $V(G')$ and convert them into a clique. Formally, define $X_u := \{u_1,\cdots, u_{\rho}\}$ for every vertex $u$ in $V(G)$. Let $Z=\cup_{u \in V(G)} X_u$.  We add $k^\circ +2$ 
 extra vertices in $Z$.  
For every pair of vertices $z_1, z_2 \in Z$, add an edge $z_1z_2$ to $E(G')$. That is, the vertices in $Z'$ form a clique. 
\item For every vertex $z$ in $Z$, we add $k^{\circ}$ pendant vertices $y_1, \cdots, y_{k^{\circ}}$ to $V(G')$. 
We  
denote the set of these vertices as $\mathtt{Guard}_V$. We also add edges $zy_1,\cdots,zy_{k^{\circ}}$ to $E(G')$.  
\item For a given coloring function $\phi : E(G) \rightarrow [t]$, let $E_i$ denote the set of edges, which are assigned color $i$. For every $i \in [t]$, we create a set of vertices $\edgegadget_i$ as follows. We add a vertex $w_e$ corresponding to every edge $e$ in $E_i$. Formally, $\edgegadget_i=\{w_e\mid e\in E_i\}$. For every $i \in [t]$, we add $\edgegadget_i$ to $V(G')$. 
We call these sets  {\em edge selector sets}. Let $e=uv$ be an edge in $G$, we add all the edges between $w_e$ and the vertices in $Z\setminus(X_u\cup X_v)$ to $E(G')$. 
Let $\edgegadget=\cup_{i\in[t]}\mathtt{ES}_i$.
\item 
For every $\edgegadget_i$, we add a vertex $g_i$ to $V(G')$. We call this vertex  \tp, corresponding to $\edgegadget_i$. We also add all the edges between $g_i$ and the vertices in $\edgegadget_i$. We denote the set of cap vertices by $\mathtt{Cap}$. 
\item For every cap vertex $g_i$ in $V(G')$, we add $k^{\circ}$ pendant vertices $y_i^1,\cdots,y_i^{k^{\circ}}$ to $V(G')$. We  
denote the set of these vertices as $\mathtt{Guard}_E$. We also add edges $g_iy_i^1,\cdots,g_iy_i^{k^{\circ}}$ to $E(G')$. 
\item We add a set of $t$ \spec, $\specialvert=\{s_1,\cdots,s_t\}$ to $V(G')$. We add all the edges between $s_i$ and the vertices in $\edgegadget_i$, where $i\in [t]$. We add edges between every pair of vertices in $\specialvert$, that is, for all $s_i,s_j \in \specialvert$, we add an edge $s_is_j$ to $E(G')$. That is, the vertices in $\specialvert$ form a clique. 
We call $s_i$ as  \specs corresponding to $\edgegadget_i$.
\item For every special vertex $x$ in $\specialvert$, we add $k^{\circ}$ pendant vertices $x_1,\cdots,x_{k^{\circ}}$ to $V(G')$. We 
denote the set of these vertices as $\mathtt{Guard}_S$. We also add edges $xx_1,\cdots,xx_{k^{\circ}}$ to $E(G')$. 
\end{itemize}
This completes the construction. Reduction algorithm returns $(G', k')$ as an instance of \splitcontract. See Figure~\ref{fig:split-hardness} for an illustration.

\begin{figure}[t]
 \centering
 \includegraphics[scale=0.75]{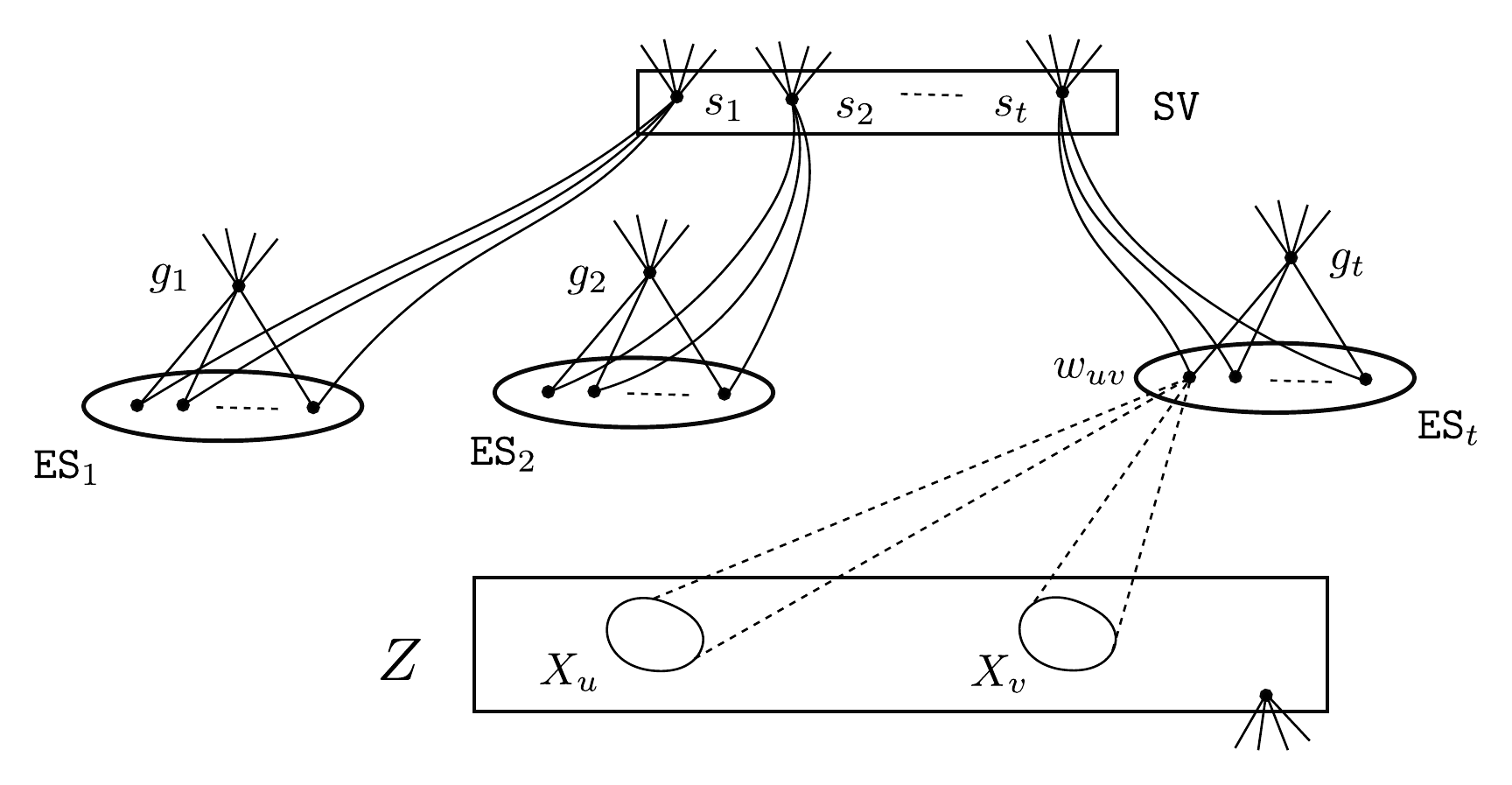}
 \caption{Sets with rectangular boundary represent cliques whereas sets with elliptical boundary represent independent sets. A vertex $w_{uv}$, corresponding to an edge $uv$, is adjacent to all vertices in $Z$ except the ones in $X_u \cup X_v$. Dashed lines shows non-adjacency between the vertex and sets. Please refer to reduction from \colordensecliqueperfect to \splitcontract for details. \label{fig:split-hardness}}
\end{figure}

Note that without loss of generality, we can assume that for a given instance $(G, k)$ of \colordensecliqueperfect, graph $G$ contains at least $k + 1$ vertices, and $k$ is not a constant. 

 \begin{lemma}\label{lem:fwd} 
Let $(G', k')$ be the instance of \splitcontract returned by the reduction algorithm mentioned above when the input is $(G, k)$ and $\delta>0$. Then, there exists a set of edges $F'$ in $G'$ of size most $k'$, such that $G'/F'$ is a split graph.
\end{lemma}

\begin{proof}
 Let $S = \{u^1,\cdots,u^k\}$ be a set of vertices that induces a colorful-$k$-clique in $(G,k)$. 
  Let $E_S = \{e_1,\cdots, e_t\}$ be the set of edges in $G[S]$. Since $E_S$ is a set of colorful edges, let $e_i\in E_i$, where $i$ in $[t]$. 
 Let $\w=\{w_{e_i}\in V(G')\mid e_i\in E_S\}$, we construct a solution $F'$ to $(G',k')$ as follows. For every $w_{e_i}\in \w$, add the edges $w_{e_i}g_i, w_{e_i}s_i \in E(G')$ to $F'$, where $g_i$ is the \tps corresponding to $\edgegadget_i$, and $s_i$ is the \specss corresponding to $\edgegadget_i$. 
 Note that we have added $2t$ edges to $F'$. 
As the number of vertices in $V(G)$ is at least $k + 1$, there exists a vertex, say $u_0$, in $V(G) \setminus S$.
Consider a vertex $z_0$ in $X_{u_0}$ in graph $G'$.  
For every vertex $u$ in $S$, we add the edges $\{u_1z_0,\cdots,u_\rho z_0\}$ to $F'$. 
Thus, for every $u\in S$, we have added $\rho$ edges to $F'$. Hence, $|F'|= 2t+\rho k= k'$.

We now show that $G'/F'$ is a split graph. Let $\mathtt{New}$ be the set of new vertices that are introduced in $G'\slash F'$ by contracting edges in $F'$. Let $C=(Z\setminus (\cup_{u\in S}X_u\cup \{z_0\}))\cup \mathtt{New}$. Let $I = \mathtt{Guard}_V \cup\ \mathtt{Guard}_E \cup\ \mathtt{Guard}_S \cup\ (\edgegadget\setminus \w)$. We claim that $(C,I)$ is a split partition of $G'\slash F'$. By the construction of $G'$, $V(G')=Z\cup\ \edgegadget \cup\ \mathtt{Cap} \cup\ \specialvert \cup\ \mathtt{Guard}_V\cup\ \mathtt{Guard}_E \cup\ \mathtt{Guard}_S$. Since $\{z_0\}\cup (\cup_{u\in S}X_u) \cup W \cup\ \mathtt{Cap} \cup\ \specialvert \subseteq V(F')$, we know that $(C,I)$ is a partition of $G'\slash F'$.
As $I$ is an independent set in $G'$ and no edges incident to $I$ are contracted, this set is also an independent set in $G'/F'$.

We now argue that $C$ is a clique in $G'$.  
Since every pair of vertices in $Z$ is adjacent to each other in $G'/F'$, if the vertices $u, v$ are in $Z\setminus \mathtt{New}$, then they are adjacent to each other in $G'/F'$. 
Consider a vertex $u$ in $\mathtt{New}$; we have the following two cases. 

\noindent {\color{blue}\textsl{Case (A): Vertex $u$ is obtained by contracting $w_{e_i}g_i$ and $w_{e_i}s_i$ for some edge $e_i (= xy)$ in $E_i$}}.
By the construction of $G'$, the vertex $w_{e_i}$ is adjacent to all the vertices in $Z\setminus (X_x\cup X_y)$ of the graph $G'$. Hence, $u$ is adjacent to all the vertices, which are in $Z \setminus (X_x \cup X_y)$, thus, $u$ is adjacent to all the vertices in $Z\setminus (\cup_{u\in S}X_u\cup \{z_0\})$. We now show that $u$ is also adjacent to all the vertices in $\mathtt{New}$. Let $v \in \mathtt{New}$, and suppose that $v$ is a vertex obtained by contracting $w_{e_j}g_j$ and $w_{e_j}s_j$, where $e_j\in E_S$. Since $s_i,s_j$ are in $\specialvert$, $s_is_j\in E(G')$, $u,v$ are adjacent to each other in $G'\slash F'$. Now, suppose that $v$ is obtained by contracting $x_iz_0$, where $x\in S$ and $i\in [\rho]$. Since $u_0 \notin S$, $z_0$ is adjacent to $w_{e_i}$, $uv$ is an edge in $E(G'\slash F')$.

\noindent{\color{blue}\textsl{Case (B): Vertex $u$ is obtained by contracting the edge $x_iz_0$, where $x\in S$}}. Since $z_0\in Z$, $u$ is adjacent to all the vertices in $Z\setminus (\cup_{u\in S}X_u\cup \{z_0\})$. Now, consider another vertex $v\in \mathtt{New}$. Note that as $u,v$ are two different vertices in $V(G'\slash F')$, $v$ can not be obtained by contracting $y_iz_0$ for any $y\in S$. Thus, $v$ is obtained by contracting $w_{e_j}g_j$ and $w_{e_j}s_j$, where $e_j\in E_S$. Then, as argued above $uv$ is an edge in $E(G'\slash F')$ (we only need to interchang $u$ and $v$ in the previous argument). Hence, any two vertices in $C$ are adjacent.

This implies that $(C, I)$ is a split partition of $G'/F'$, and hence $G'/F'$ is a split graph. Since the number of edges in $F'$ is at most $k'$, this proves the lemma.
\end{proof}

In the following lemma, we argue that given an approximate solution for an instance of \splitcontract, one can obtain a set of $\alpha k$ vertices, $\alpha > 1$, that spans at least $\epsilon t$ colorful edges, where $0<\epsilon < 1$. 

\begin{lemma}\label{lem:rev}
Let $(G', k')$ be an instance of \splitcontract returned by the reduction algorithm mentioned above when the input is $(G, k)$ and $\delta >0$. If there exists a set of edges $F'$ in $G'$ such that $G'/F'$ is a split graph and size of $F'$ is at most $(\nicefrac{5}{4} - \delta)k'$, then there exists a set of at most $\nicefrac{1}{\delta}\cdot k$ vertices in $G$ that spans at least $\nicefrac{3\delta}{2} \cdot t$ colorful edges.
\end{lemma} 

We establish some properties of the instance of \splitcontract\ that is returned by the reduction algorithm and its solution before presenting proof of Lemma~\ref{lem:rev}.
In Claims~\ref{claim:clique-side-ver} to \ref{claim:disjoint-from-Z}, we use the following notation:
The reduction algorithm returns $(G', k')$, when input is $(G, k)$, and $\delta$. Let $F'$ be a set of edges in $G'$ such that $G'/F'$ is a split graph and size of $F'$ is at most $(\nicefrac{5}{4} - \delta) \cdot k'$. Let $\calW$ be the $G'/F'$-witness structure of $G'$. 
Let $\psi : V(G') \rightarrow V(G'/F')$ be the onto function corresponding to contracting all the edges in $F'$. 
For a vertex $\tilde{w}$ in $V(G'/F')$, $W(\tilde{w})$ denotes the witness set which is contracted to obtain the vertex $\tilde{w}$.
Hence, for a vertex $u$ in $V(G')$, if $u \in W(\tilde{w})$ then $\psi(u) = \tilde{w}$.
We fix a split partition $(\tilde{C}, \tilde{I})$ of $V(G'/F')$.
Note that $|F'|+2 \leq 2|F'| \le 2 (\nicefrac{5}{4} - \delta)k' \le \nicefrac{5}{2}k' \le k^{\circ}$. Hence, there are at least  
$|F'| + 2$ many pendant vertices adjacent to every cap vertices, special vertices, and every vertex in $Z$ of $G'$. Moreover, since the size of $Z$ is at least $k^{\circ}+2$, there exists at least one vertex in $Z$ which is not in $V(F')$. Let $z^{\star}$ be one such vertex in $V(G')$ which is in $Z \setminus V(F')$. Note that $W(\psi(z^{\star}))= \{z^{\star}\}$.

We refer readers to Section~\ref{sec:intro} for an overview of the proof.
Consider a cap vertex $g_i$ and the witness set $W(\psi(g_i))$.  
A cap vertex $g_i$ is said to be \emph{spoiled}, if $W(\psi(g_i))$ either $(a)$ contains another cap vertex;
$(b)$ intersects with $Z$; or
$(c)$ has more than one vertex from $\edgegadget$. 
We bound the number of cap vertices that can be spoiled because of $(a), (b)$ or $(c)$ in Claims~\ref{claim:multi-cap}, \ref{claim:intersect-with-Z}, and \ref{claim:disjoint-from-Z}, respectively.
We first present few results, which are used in the proof of these claims.

\begin{claim}\label{claim:clique-side-ver} If $u \in V(G')$ is either a cap vertex or a special vertex or in $Z$, then $\psi(u)$ is in $\tilde{C}$. 
\end{claim}
\begin{proof} Any such vertex $u$ is adjacent to $|F'| + 2$ many pendant vertices. This implies that there are at least two pendant vertices, say $u_1, u_2$, which are not in $V(F')$, which in turn  implies that $W(\psi(u_1))$ and $W(\psi(u_2))$ are singleton sets in $\calW$. Since $u_1, u_2$ are not adjacent to each other in $G'$ and $W(\psi(u_1)), W(\psi(u_2))$ are singleton witness sets, $\psi(u_1)$, $\psi(u_2)$ are not adjacent to each other in $G'/F'$. Hence, at most one of them can be in $\tilde{C}$. Without loss of generality, let $\psi(u_1)$ is in $\tilde{I}$. Since $u$ is adjacent to $u_1$ in $G'$, and $u$ is not contained in $\psi(u_1)$, $\psi(u)$ is adjacent to $\psi(u_1)$. This implies $\psi(u)$ is in $\tilde{C}$. 
\end{proof}

\begin{claim}\label{claim:cap-vertex-solution-edge} For every cap vertex $g_i$ there exists a vertex $u_e$ in $\edgegadget_i$ such that $g_iu_e$ is in $F$. 
\end{claim}
\begin{proof} Recall that $\psi(z^{\star})$ is a vertex in $\tilde{C}$ and $W(\psi(z^{\star}))$ is a singleton witness set. 
Assume that for a cap vertex $g_i$, there is no vertex $u_e$ in $\edgegadget_i$ such that edge $g_iu_e$ is in $F$. This implies that $\psi(g_i) \cap \edgegadget_i$ is an empty set. Since neighbors of $g_i$ outside $\edgegadget_i$ (i.e. pendant neighbors of $g_i$) are not adjacent to $z^{\star}$ in $G'$, there is no edge with one endpoint in $W(\psi(g_i))$ and another one in $W(\phi(z^{\star}))$. 
But by Claim~\ref{claim:clique-side-ver}, both $\psi(z^{\star})$ and $\psi(g_i)$ are in $\tilde{C}$. This contradicts the fact that $\tilde{C}$ is a clique. Hence, our assumption is wrong, which concludes the proof of the claim. 
\end{proof}

\begin{claim}\label{claim:special-vertex-solution-edge} For every special vertex $s_i$ there exists a vertex $u_e$ in $\edgegadget$ such that $u_e$ is in $W(\psi(s_i))$. 
\end{claim}
\begin{proof} For the sake of contradiction, assume that there exists a witness set, say $W(\psi(s_i))$ such that $W(\psi(s_i)) \cap \edgegadget =\emptyset$.  
Recall that $\psi(z^{\star})$ is a vertex in $\tilde{C}$ and $W(\psi(z^{\star}))$ is a singleton witness set. Since $W(\psi(s_i))$ does not contain any vertex of $\edgegadget$, $W(\psi(s_i)) \subseteq \mathtt{SV}\cup \mathtt{Guard}_S$. Hence,  
there is no edge with one endpoint in $W(\psi(s_i))$ and another one in $W(\psi(z^{\star}))$. 
By Claim~\ref{claim:clique-side-ver} both $\psi(s_i)$ and $\psi(z^{\star})$ are in $\tilde{C}$. This contradicts the fact that $\tilde{C}$ is a clique. Hence, our assumption is wrong, which concludes the proof of the claim. 
\end{proof}

\begin{claim}\label{claim:special-vertex-t-edges} There are at least $2t$ edges in $F'$, which are incident to either cap vertices, or special vertices. Moreover, every vertex can be assigned to an edge in $F'$,  which is unique to it.
\end{claim}
\begin{proof} By Claim~\ref{claim:cap-vertex-solution-edge}, there are at least $t$ edges incident to cap vertices. Since a cap vertex is neither adjacent to other cap vertex nor to a special vertex, a cap vertex can be uniquely mapped to the edge incident to it.
  
 Consider a special vertex $s_i$. Define $S = W(\psi(s_i)) \cap \specialvert$.
 By Claim~\ref{claim:special-vertex-solution-edge}, $W(\psi(s_i))$ contains a vertex, say $u_e$, which is in $\mathtt{ES}$. Hence, there is an additional edge, which is incident to $S$.
 Without loss of generality, we can assume that $F'$ contains a spanning tree of $G'[S \cup \{u_e\}]$. This implies that there is at least $|S|$ many edges incident to $S$. To assign each vertex to a unique edge, we root this spanning tree at $u_e$. For every vertex in $S$, assign it to the edge connecting that vertex to its parent in this rooted tree. 
 Since witness structure $\calW$, partitions $\specialvert$, there are at least $|\specialvert| = t$ edges in $F'$ which has at least one endpoint in $\specialvert$. Note that none of these edges is incident to cap vertices. This concludes the proof.
\end{proof}

Claim~\ref{claim:special-vertex-t-edges} allows us to define an \emph{one-to-one} function $\lambda$ from the set of cap vertices and special vertices to the edges in $F'$ such that $\lambda(v)$ is an edge incident to $v$. We call such functions as \emph{accounting function}s. 
We use such a function in the following arguments to bound certain kinds of witness sets.
For an accounting function $\lambda$, the set of edges in $F'$ which do not belong to the range of $\lambda$ are called \extra.
We fix a function $\lambda$ and modify it in Claim~\ref{claim:multi-cap}~and~\ref{claim:disjoint-from-Z} at certain vertices of special vertices to obtain another accounting function. We use the number of edges in \extra\ with respect to this function to argue that the number of spoiled cap vertices is not very large. 
Let $\gamma = \nicefrac{3\delta}{4}$. 
Note that, for any accounting function $\lambda$, the number of edges in \extra\ is at most $|F'|- 2t \le (\nicefrac{5}{4} - \delta)k' -2t \le (\nicefrac{5}{4} - \delta) \cdot (2t + \rho k)-2t \le (\nicefrac{1}{2})t - (\nicefrac{3\delta}{4})t = (\nicefrac{1}{2} - \gamma) \cdot t$. If an edge is not incident to cap or special vertex then such an edge is always in \extra\ for any accounting function.  

\begin{claim} \label{claim:multi-cap}
Let $\mathtt{\mathtt{Cap}_1}$ be a subset of $\mathtt{Cap}$ such that every vertex in $\mathtt{Cap}_1$ is in a witness set in $\cal W$ that contains at least two cap vertices. Then, $F'$ contains at least $\nicefrac{1}{2}\cdot |\mathtt{Cap}_1|$ edges in {\rm $\extra$} for an accounting function $\lambda$.
\end{claim}
\begin{proof}
 Let $\calW_1$ be the collection of witness sets in $\calW$ which contains at least two cap vertices. Clearly, $|{\cal W}_1|\leq \nicefrac{1}{2}\cdot |\mathtt{Cap}_1|$. Consider a witness set $W(p)$ in $\calW_1$, we argue that there are at least $|W(p)\cap \mathtt{Cap}_1|-1$ many edges in ${\rm \extra}$ incident to the vertices of $W(p)$ for an accounting function $\lambda$.

 Consider the accounting function $\lambda$ which is constructed/modified before considering $W(p)$. We modify this accounting function   
 for some vertices in $\specialvert$ to obtain another accounting function. Let $T$ be a spanning tree of $G'[W(p)]$ such that $E(T)\subseteq F'$. We arbitrarily fix a cap vertex $g_o$ in $W(p)$ and the root of $T$ at $g_o$.
 Consider a cap vertex $g_i$ and let $g_j$ be the first cap vertex on the unique path from $g_i$ to $g_o$ in $T$(Note that $g_j$ can be equal to $g_o$).
Let $P_{g_ig_j}$ be the unique path between $g_i$ and $g_j$ in the tree $T$.
By the construction of $G'$, 
there exists a vertex $u^i_e$ in $\edgegadget_i$ such that $g_iu^i_e$ is an edge in $E(P_{g_ig_j})$. We modify $\lambda$ in a way that there exists an edge incident to $u^i_e$ which is in \extra. Since $u^i_e$ is not in the domain of accounting function, this edge in \extra\ is unique to $g_i$.  
Hence,  for every cap vertex in $W(p) \cap \mathtt{Cap}_1$ except for $g_o$, there is an edge in \extra\ with respect to an accounting function $\lambda$.

Consider the path $P_{g_ig_j}$ between $g_i$ and $g_j$ in the spanning tree $T$. Clearly, there exists a vertex $u_e^i \in \edgegadget_i$ and a vertex $w\in Z \cup \{s_i\}$ such that $g_iu_e^i, u_e^iw \in E(P_{g_ig_j})$, otherwise there cannot be a path between $g_i$ and $g_j$.  
If $w\in Z$, then $u_e^iw$ is in $\extra$ for $\lambda$. Suppose that $w= s_i$. Then, in the path $P_{g_ig_j}$, either $s_i$ is adjacent to $u_{e'}^i \in \edgegadget_i$ or a special vertex $s_q \in \specialvert$, then $s_iu_{e'}^i \in E(P_{g_ig_j})$. Note that either $s_iu_e^i$ or $s_iu_{e'}^i $ is in \extra\ for $\lambda$. 
Now, suppose that $s_is_q \in E(P_{g_ig_j})$.
If $\lambda(s_i)=s_is_q$, then $s_iu_e^i$ is in $\extra$ for $\lambda$. Suppose that $\lambda(s_i)=s_iu_e^i$. Now, we modify $\lambda$ to obtain $\lambda(s_i)=s_is_q$.
We denote the successor and predecessor of a vertex $s$ in the path $P_{g_ig_j}$ by $\mathtt{succ}[s]$ and $\mathtt{pred}[s]$, respectively. Let $s_\ell$ be the first special vertex in the path $P_{g_ig_j}$ such that either $\lambda(s_\ell)=s_\ell \mathtt{succ}[{s_\ell}]$ or $\mathtt{succ}[{s_\ell}] \in \edgegadget_\ell$.  
Let $\mathtt{\widetilde{SV}}$ be the set of the vertices in the subpath of $P_{g_ig_j}$ from $s_i$ to $s_\ell$.  For every special vertex $s$ in $\mathtt{\widetilde{SV}}$, we set $\lambda(s)=s\mathtt{succ}[s]$.
 Note that $\lambda$ is still a \emph{one-to-one} function as earlier either $s_\ell \mathtt{pred}[s_\ell]$ (when $\lambda(s_\ell)=s_\ell \mathtt{succ}[{s_\ell}]$) or $s_\ell\mathtt{succ}[{s_\ell}]$ (when $\lambda(s_\ell)=s_\ell \mathtt{pred}[{s_\ell}]$, and $\mathtt{succ}[{s_\ell}] \in \edgegadget_\ell$) was not in the range of $\lambda$.
 Now, since $\lambda(s_i)=s_is_q$, $s_iu_e^i$ is in $\extra$ for the modified $\lambda$. 
 This implies that for any cap vertex $g_i\neq g_0$ there is an edge incident to $u_e^i$ which is in $\extra$ for an accounting function $\lambda$.
 Hence, $F'$ has at least $|W(p)\cap \mathtt{Cap}_1|-1$ edges in $\extra$.

By the above discussion, for every witness set $W(p) \in {\cal W}_1$, $F'$ has at least $|W(p)\cap \mathtt{Cap}_1|-1$ edges in $\extra$ for the function $\lambda$. Since $\calW$ partitions vertices in $\mathtt{Cap}$, we can infer that $F'$ has at least $|\mathtt{Cap}_1|-|{\cal W}_1|$ edges in $\extra$ for the function $\lambda$. 
\end{proof}

\begin{claim}\label{claim:intersect-with-Z}
Let $\mathtt{Cap}_2$ be a subset of $\mathtt{Cap}$ such that every vertex in $\mathtt{Cap}_2$ is in a witness set in $\cal W$ that intersects with $Z$ and contains exactly one cap vertex. Then, $F'$ contains at least $|\mathtt{Cap}_2|$ edges in {\rm $\extra$} for any function $\lambda$.
\end{claim}
\begin{proof} 
Let $\calW_2$ be the collection of witness sets in $\calW$ that intersects with $Z$ and contains exactly one cap vertex. Consider a witness set $W(p)$ in $\calW_2$ and let $z$ be a vertex in $W(p) \cap Z$. 
 Hence, $F'$ contains at least one edge incident to $z$. For any function $\lambda$, an edge incident to a vertex of $Z$ is in $\extra$ as the vertices of $Z$ are neither adjacent to cap vertices nor special vertices. This implies that $F'$ contains an edge in $\extra$ that is incident to a vertex of $W(p)$.
 Since $\calW$ partitions vertices in $\mathtt{Cap}_2$, and every witness set in $\calW_2$ contains exactly one cap vertex, $F'$ has at least $|\mathtt{Cap}_2|$ edges in $\extra$ for the function $\lambda$.
\end{proof}

\begin{claim}\label{claim:disjoint-from-Z}
Let $\mathtt{Cap}_3$ be a subset of $\mathtt{Cap}$ such that every vertex in $\mathtt{Cap}_3$ is in a witness set in $\cal W$ that is disjoint from $Z$; contains exactly one cap vertex; and 
 at least two vertices of $\mathtt{ES}$. Then, $F'$ contains at least $|\mathtt{Cap}_3|$ edges in {\rm $\extra$} for a function $\lambda$.
\end{claim}
\begin{proof}
 Let $\calW_3$ be a collection of witness sets in $\calW$ that are disjoint from $Z$; contains exactly one cap vertex; and at least two vertices of $\mathtt{ES}$.
 Consider a witness set $W(p)$ in $\calW_3$. We argue that $F'$ has at least one edge in {\rm $\extra$} that is incident to a vertex of $W(p)$ for a function $\lambda$.

 Let $T$ be a spanning tree of $G'[W(p)]$ such that $E(T)\subseteq F'$.
Consider the accounting function $\lambda$ which is constructed/modified before considering $W(p)$. Let $g_i$ be the cap vertex contained in $W(p)$.
By Claim~\ref{claim:cap-vertex-solution-edge}, there exists $u_e^i \in \edgegadget_i$ such that $g_iu_e^i \in F'$.
We consider two cases depending on whether another vertex in $\edgegadget \cap W(p)$ is in $\edgegadget_i$ or not.
Suppose that there exists a vertex $u_{e'}^i \in \edgegadget_i \cap W(p)$ such that $u_{e'}^i\neq u_{e}^i$ and
if $g_iu_{e'}^i \in E(T)$, then either $g_iu_{e}^i$ or $g_iu_{e'}^i$ is in $\extra$ for $\lambda$. Suppose that $g_iu_{e'}^i \notin E(T)$ and consider a path from $g_i$ to $u_{e'}^i$, say $P_{g_iu_{e'}^i}$, in the spanning tree $T$. Since $W(p)$ does not intersect with $Z$, $E(P_{g_iu_{e'}^i}) = \{g_iu_e^i, u_e^is_i, s_iu_{e'}^i\}$, either $s_iu_e^i$ or $s_iu_{e'}^i$ is in $\extra$ for $\lambda$. Now, suppose that $u_e^i$ is the only vertex in $ \edgegadget_i \cap W(p)$. Since $W(p)$ contains at least two vertices of $\mathtt{ES}$, there exists a vertex $u_{e'}^j \in \edgegadget_j \cap W(p)$, where $j\in [t], j\neq i$. Consider a path $P_{g_iu_{e'}^j}$ from $g_i$ to $u_{e'}^j$, since $W(p)$ is disjoint from $Z$, $u_{e}^is_i,u_{e'}^js_j \in E(P_{g_iu_{e'}^j})$ (otherwise there can not be a path from $g_i$ to $u_{e'}^j$), and the path from $s_i$ to $s_j$ in $T$ contains only special vertices. Let $S$ be the set of vertices in the path from $s_i$ to $s_j$ in $T$. We modify $\lambda$ to define $\lambda(s)=s\mathtt{succ}[s]$, for all $s \in S$. Note that $\lambda$ is still a {\em one-to-one} function, and $s_iu_{e}^i$ is in $\extra$. Hence, for every witness set in $\mathtt{Cap}_3$, we have an edge in $\extra$ for some function $\lambda$. 
\end{proof}

\begin{proof}(of Lemma~\ref{lem:rev})
 Let $\mathtt{Cap}_1, \mathtt{Cap}_2,$ and $\mathtt{Cap}_3$ be the subset of $\mathtt{Cap}$ as defined in Claim~\ref{claim:multi-cap}, \ref{claim:intersect-with-Z} and \ref{claim:disjoint-from-Z}. Note that sets $\mathtt{Cap}_1, \mathtt{Cap}_2, \mathtt{Cap}_3$ are pairwise disjoint. Let $\mathtt{Cap}_4$ be the collection of cap vertices in $\mathtt{Cap} \setminus (\mathtt{Cap}_1 \cup \mathtt{Cap}_2 \cup \mathtt{Cap}_3)$.

 We first argue that $|\mathtt{Cap}_4| \geq 2\gamma t$. Since there are $t$ many cap vertices, $|\mathtt{Cap}_1| + |\mathtt{Cap}_2| + |\mathtt{Cap}_3| + |\mathtt{Cap}_4| = t$.
 Recall the accounting function $\lambda$, which was fixed before Claim~\ref{claim:multi-cap}. We modify this function in Claim~\ref{claim:multi-cap} and \ref{claim:disjoint-from-Z} at some special vertices to obtain another accounting function. Note that the modifications to $\lambda$ in Claim~\ref{claim:multi-cap} are at special vertices, which are contained in witness sets that contains at least two cap vertices. The modifications to $\lambda$ in Claim \ref{claim:disjoint-from-Z} are at special vertices, which are contained in witness sets that contains exactly one cap vertex. Since witness sets partition special vertices, one modification does not affect another. Since Claim~\ref{claim:intersect-with-Z} holds true for any accounting function, we know that for the function $\lambda$, there are at least $ \nicefrac{1}{2} \cdot |\mathtt{Cap}_1| + |\mathtt{Cap}_2| + |\mathtt{Cap}_3|$ many edges in \extra. Since there are at most $(\nicefrac{1}{2} - \gamma)t$ many edges in \extra\ for any accounting function, $|\mathtt{Cap}| - |\mathtt{Cap}_4| = |\mathtt{Cap}_1| + |\mathtt{Cap}_2| + |\mathtt{Cap}_3| \le (1 - 2\gamma)t$. Hence, $|\mathtt{Cap}_4| \ge 2\gamma t$. This implies that there are at least $2\gamma t$ many cap vertices, which are contained in a witness set, that does not contain any other cap vertex, no vertex of $Z$, and precisely one vertex of $\mathtt{ES}$.

 Let $\calW^\star$ be the subset of witness sets that contain at least one cap vertex. Let $\calW_1, \calW_2,$ and $\calW_3$ be the subset of $\calW^\star$ as defined in the proofs of Claim~\ref{claim:multi-cap}, \ref{claim:intersect-with-Z} and \ref{claim:disjoint-from-Z}.
Let $\calW_4$ be the collection of remaining witness sets in $\calW^\star \setminus (\calW_1 \cup \calW_2 \cup \calW_3)$.  
Note that any witness set in $\calW_4$ contains exactly one cap vertex, no vertex of $Z$, and exactly one vertex of $\mathtt{ES}$.
Hence, $|\calW_4| = |\mathtt{Cap}_4| \geq 2\gamma t$. 

 Let $W(p)$ be a witness set in $\calW_4$ and $u_e$ be a vertex in $W(p) \cap \edgegadget$. We argue that all the non-neighbors of $u_e$ in $Z$ are in $V(F')$. For the contradiction, suppose that there exists a non-neighbor of $u_e$, say $w$, in $Z$ that does not belong to $V(F')$. Note that $W(p)$ contains a cap vertex. By Claim~\ref{claim:clique-side-ver}, $\psi(w)$ and $p$ are in $\tilde{C}$.
 Since $W(p)$ neither contains a vertex of $Z$ nor any other vertex of $\mathtt{ES}$, $p$ and $\psi(w)$ are not adjacent to each other, a contradiction to the fact that $\tilde{C}$ is a clique. Hence, our assumption is wrong and all the non-neighbors of $u_e$ in $Z$ are in $V(F')$. Recall that for every vertex $v$ in $G$, the reduction algorithm has added $\rho$ many copies of $v$ in $G'$, and the set of these vertices is   
 denoted by $X_{v}$. If the vertex $u_e$ in graph $G'$ corresponds to the edge $e = v_1v_2$ then all vertices in $X_{v_1} \cup X_{v_2}$ are contained in $V(F')$.  
  
 Let $Y$ be the subset of $\mathtt{ES}$ such that every vertex in $Y$ is in a witness set in $\calW_4$. Since each witness set in $\calW_4$ contains exactly one vertex of $Y$, $|Y| \geq 2\gamma t$.
 Let $Y_{E}$ be the set of edges in $G$, which corresponds to vertices in $Y$.
 Let $\bar{N}_{Z}(Y)$ be the set of non-neighbors of $Y$ in $Z$. We know that $\bar{N}_Z(Y)\subseteq V(F')$.  
 Since size of $F'$ is at most $|F'|\leq (\nicefrac{5}{2})t-\gamma t$ and $2t$ edges are incident to $\mathtt{Cap} \cup \specialvert$ (Claim~\ref{claim:special-vertex-t-edges}) there are at most $\nicefrac{t}{2}-\gamma t$ edges incident to the vertices of $Z$. Since every edge can be incident to at most two vertices of $\bar{N}_Z(Y)$, $|\bar{N}_Z(Y)|\leq t-2\gamma t$.
Let $S$ be the set of vertices, which are endpoint of some edge in $Y_E$ in $G$.
Since we have added $\rho$ copies for each vertex in $S$, we have $\rho |S| = |\bar{N}_Z(Y)| \le t-2\gamma t$.
This implies there are at most $(1 - 2\gamma)\nicefrac{t}{\rho} \le \nicefrac{k}{\delta}$ (as $\rho = \lceil \nicefrac{\delta t}{k} \rceil$) vertices which span at least $2\gamma t$ many edges in $G$.
Since  a witness set in $\calW_4$ contains exactly one cap vertex and exactly one vertex of $\mathtt{ES}$, due to Claim~\ref{claim:cap-vertex-solution-edge}, any two edges in $Y_E$ corresponds to two vertices in $G'$, which are in different edge selector sets, that corresponds to different color class in edge coloring $\phi$ in $G$. Hence, the set of edges $Y_E$ is colorful. This completes the proof of the lemma.  
\end{proof}

We are now in a position to prove the main theorem of this section.

\begin{reptheorem}{thm:split-no-lossy} Assuming {\sf Gap-ETH}, no \FPT\ time algorithm can approximate {\sc Split Contraction} within a 
 factor of $\left(\nicefrac{5}{4}-\delta \right)$, for any fixed constant $\delta>0$. 
\end{reptheorem}
\begin{proof}
 For the sake of contradiction, assume that for a given fixed $\delta > 0$, there exists an \FPT\ time algorithm, say $\calA_{\delta}$, which can approximate {\sc Split Contraction} within a factor of $(\nicefrac{5}{4} - \delta)$.

 Consider an instance $(G, k)$ of \colordensecliqueperfect. We run the reduction algorithm mentioned in this section to obtain an instance $(G', k')$ of \textsc{Split Contraction}, where $k' = 2t + \lceil \nicefrac{\delta t}{k}\rceil, k' \in \calO(k^2)$. Let $\texttt{opt}(G')$ be the number of minimum edges that needs to be contracted in $G'$ to convert it into a split graph.

 By Lemma~\ref{lem:fwd}, there exists a set of edges $F'$ in $G'$ such that $G'/F'$ is a split graph and the size of $F'$ is at most $k'$. This implies that $\texttt{opt}(G') \le k'$. Let $\tilde{F}$ be the set of edges returned by algorithm $\calA_{\delta}$ when the input is $(G', k')$. Since $\calA_{\delta}$ returns an approximate solution within factor $(\nicefrac{5}{4} - \delta)$, we know that $|\tilde{F}| \le (\nicefrac{5}{4} - \delta) \cdot \texttt{opt}(G') \le (\nicefrac{5}{4} - \delta)\cdot k'$. By Lemma~\ref{lem:rev}, there exists a set of at most $\nicefrac{k}{\delta}$ vertices in $G$ that spans at least $(\nicefrac{3\delta}{2})t$ edges in $G$. Note that the proof of Lemma~\ref{lem:rev} can easily be converted into a polynomial time algorithm to obtain these set of vertices and edges in $G$ given $F'$. Since $\calA_{\delta}$ is an \FPT\ time approximation algorithm, it runs in time $f(k')\cdot |V(G')|^{\calO(1)} = f(k) \cdot |V(G)|^{\calO(1)}$ time.

 We can conclude that there is a $f(k)\cdot |V(G)|^{\calO(1)}$-time algorithm such that, given an integer $k$ and an edge colored graph $G$ containing a colorful-$k$-clique, always outputs a vertex set of size at most $\nicefrac{k}{\delta}$ vertices that span at least $2\gamma t$ colorful edges. Fix positive constants $\epsilon, \alpha$ such that $0 < \epsilon < 1$; $1 <\alpha$ and $3 \le 4 \alpha \epsilon$. For $\delta = \nicefrac{1}{\alpha}$, the above conclusion contradicts Lemma~\ref{lemma:original-problem}. Hence, our assumption is wrong, which implies the correctness of the theorem. 
\end{proof}
\section{(No) \fpt Approximation Algorithm for Chordal Contraction}
\label{hardness_chordal}
In this section, we show that unlike \textsc{Clique Contraction} or \textsc{Split Contraction}, one can not obtain a lossy kernel of any size for \textsc{Chordal Contraction}.
It is known that for every $\alpha \ge 1$ and parameterized optimization problem $\Pi$, $\Pi$ admits a fixed parameter tractable $\alpha$-approximation algorithm if and only if $\Pi$ has an $\alpha$-approximate kernel \cite[Proposition~$2.11$]{lossy}.
We prove that \chc\ parameterized by the solution size $k$, cannot be approximated within a factor of $F(k)$ in \FPT\ time.
Towards this we give a reduction from {\sc Set Cover} and use a known result that no \FPT\ time algorithm can approximate {\sc Set Cover} within a factor of $F(k)$, where $F(k)$ is the function of $k$ alone~\cite{DBLP:conf/stoc/SLM18}.

The parameterized optimization version of \chc\ is defined as follows.

$$\textsc{ChC}(G,k,F)=\begin{cases}
\min\left\{|F|,k+1\right\}\;\; &\text{if}\; G/F\; \text{is a chordal graph}\\
\infty \; \; &\text{otherwise}.
\end{cases}$$




Lokshtanov et al.~\cite{elimNew} proved that \textsc{Chordal Contraction}, parameterized by solution size $k$, does not admit a (classical) kernel of any size under widely believed assumption. They proved this by presenting a \emph{parameter preserving} reduction from \textsc{Hitting Set} problem to \textsc{Chordal Contraction}. In the following lemma, we argue that such \emph{parameter preserving} reduction can also be obtained in case of the optimization version of problems. Our reduction is the same as the reduction given by Lokshtanov et al.~\cite{elimNew}. We need some more arguments to show that there exists a \texttt{$1$-appt} from \textsc{Set Cover$/k$} to \textsc{Chordal Contraction} when parameterized by solution size. For the sake of completeness, we give the full reduction.
%

\begin{lemma}\label{lemma:1-appt-chordal} There exists an $1$-approximate polynomial parameter transformation (\texttt{$1$-appt}) from \textsc{Set Cover$/k$} to \textsc{Chordal Contraction} parameterized by solution size. 
\end{lemma}
\begin{proof} To prove this lemma, we present a reduction algorithm, $R_{\calA}$, which given an instance $((U, \calS), k)$ of \textsc{Set Cover$/k$} outputs an instance $(G, k')$ of \textsc{Chordal Contraction}. We also present a solution lifting algorithm that takes as input an instance $((U, \calS), k)$ of \textsc{Set Cover$/k$}; an output instance $(G, k') = R_{\calA}((U, \calS), k)$ of \textsc{Chordal Contraction}; and a solution $F$ to the instance $(G, k')$; and outputs a solution $\calF$ to $((U, \calS), k)$ such that \textsc{SC$/k((U, \calS), k, \calF)$} = \textsc{ChC}$(G, k, F)$.

 Without loss of generality, we assume that every element of the universe $U$ is contained in some set $S_i$ in $\calS$ as otherwise $((U, \calS), k)$ is a trivial instance. We first present a reduction algorithm.

\smallskip
\noindent \textbf{Reduction Algorithm :}
Given an instance $((U,\calS), k)$ of the \textsc{Set Cover} problem with $U = \{u_1, \cdots , u_n \}$ and $\calS = \{S_1, \cdots, S_m\}$, the algorithm constructs graph $G$ as follows: It creates a vertex $s_j$ for each set $S_j$ in $\calS$ and three vertices $a_{i}, b_{i}$ and $c_{i}$ for each element $u_i$ in the universe $U$. It also adds a special vertex $g$ to $G$. It adds following edges in $G$.  
 \begin{itemize}
 \item an edge between any two different vertices corresponding to sets; (In other words, the algorithm converts set $\{s_{1}, s_2, \cdots, s_m\}$ into a clique by adding all edges $s_js_{j^{\prime}}$ for $1 \le j, j^{\prime} \le m$ and $j \neq j^{\prime}$.)
 \item edge $gs_j$ for every $j$ in $\{1, 2, \cdots, m\}$ and edges $ga_i, gb_i$ for every $i$ in $\{1, 2, \cdots, n\}$; 
 \item edges $a_ic_i$ and $b_ic_i$ for every $i$ in $\{1, 2, \cdots, n\}$;
 \item for an element $x_{i}$ and a set $S_j$, if $x_i$ is in $S_j$ then it adds edge $c_is_{j}$;
 \end{itemize}
The algorithm returns $(G, k)$ as an instance of \textsc{Chordal Contraction}.
See Figure \ref{fig:redc4} for an illustration.

It is easy to verify that graph $G$ does not contain any induced cycle of length five or more. 
We have created cycles of length four for each element in the universe which intersects with each other only in $g$. 
Informally speaking, all these cycles can be killed by introducing the edge $gc_j$ for every cycle.
To introduce all these chords with at most $k$ contractions, we need to carefully select at most $k$ sets (and contracts edges of the form $gs_i$) which \emph{covers} all the elements.
We argue that introducing chords of the form $gs_j$ also kills other cycles of length four in the graph.
 
 \begin{figure}[t]
  \centering
  \includegraphics[scale=0.75]{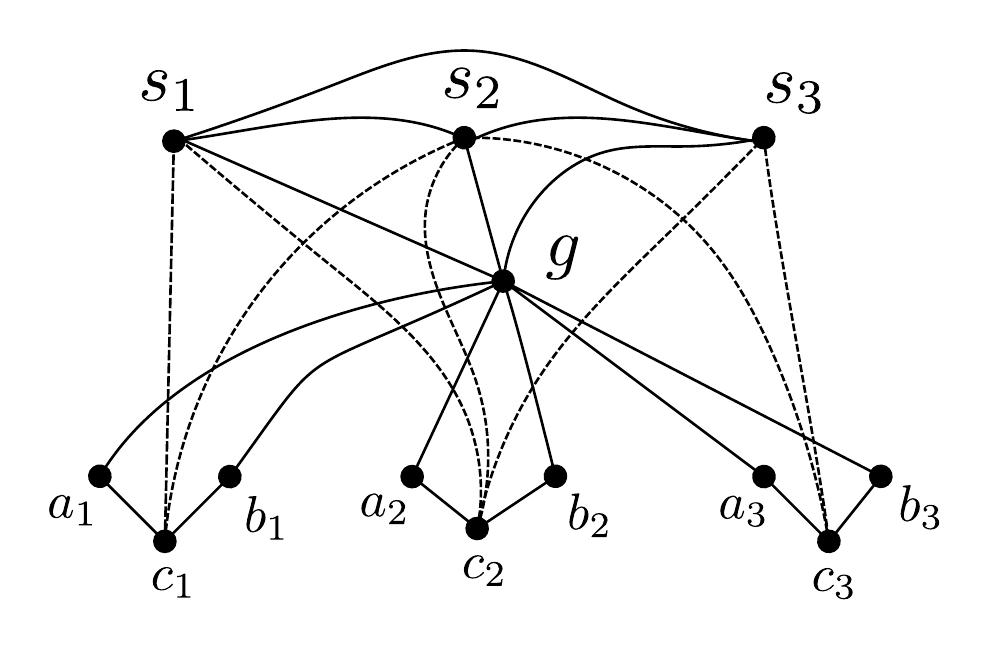}
  \caption{Reduction from an instance $((U, \calS), k)$ of \textsc{Set Cover} to an instance $(G, k)$ of \textsc{Chordal Contraction}. Here $U = \{u_1, u_2, u_3\}; \calS = \{S_1, S_2, S_3\}$ where $S_1 = \{u_1, u_2\}, S_2 = \{u_1, u_2, u_3\}$, and $S_3 = \{u_2, u_3\}$. Straight lines (eg. $a_1g$) are part of gadget construction whereas dashed lines (eg. $c_1s_1$) are added because of containment relationship between elements and sets. Please refer to the Reduction Algorithm in the proof of  Lemma~\ref{lemma:1-appt-chordal}. \label{fig:redc4}}
 \end{figure}

\smallskip
\noindent \textbf{Solution Lifting Algorithm:} Let $F^{\circ}$ be the given solution for $(G, k)$. If $|F^{\circ}| \ge k + 1$ then the algorithm returns $\calF = \calS$ as a solution. Otherwise, the algorithm first constructs another solution for $(G, k)$ with the following two operations. $(a)$ If $F^{\circ}$ contains an edge of the form $ga_i, gb_i, a_ic_i$ or $b_ic_i$, then replace it by $gs_j$, where $S_j$ is any set containing $u_i$. $(b)$ If $F^{\circ}$ contains an edge of the form $s_jc_i$, then replace it by $gs_j$. Let $F$ be the solution obtained from $F^{\circ}$ by exhaustively applying these two operations. Consider a $G/F$-witness structure of $G$ and let $W(g)$ be the witness set containing $g$. Define $\calF$ as a collection of sets in $\calS$ whose corresponding vertices are contained in $W(g)$. The algorithm returns $\calF$ as a solution for $((U, \calS), k)$.

We justify the two modification operations defined in solution lifting algorithm.
It is easy to see that $|F| \le |F^{\circ}|$.
We slightly abuse notations and rename new vertex added while contracting edge $gs_j$ as $g$.
The only cycle affected by contracting edges of the form $ga_i, gb_i, a_ic_i$ or $b_ic_i$ is $\{g, a_i, c_i, b_i\}$.
By contracting an edge of the form $gs_j$ where $S_j$ is any set containing $u_i$, we introduce another chord $gc_i$ which \emph{destroys} the cycle.
Similarly, if $F^{\circ}$ contains an edge of the form $s_jc_i$, then the only four-cycles of $G$ that gets a chord in $G/\{s_jc_i\}$ are: $\{g, a_i, c_i, b_i\}$, $\{s_j, g, a_i, c_i\}$, and $\{s_j, g, b_i, c_i\}$. 
All of these cycles get a chord when the edge $gs_j$ is contracted instead.
This implies that if $G/F^{\circ}$ is a chordal graph then so is $G/F$.

We now prove that the cycles present in the graph $G$ are of a very specific type. This claim is similar to Proposition~$1$ in \cite{elimNew}. Note that because of Claim~$1$, to convert $G$ into a chordal graph it is sufficient to introduce chords $gc_i$ for every $i$ in $\{1, 2, \cdots, n\}$.
  
\vspace{0.2cm}
\noindent \emph{Claim 1.} Graph $G$ does not contain any induced cycle of length five or more. The only induced cycles of length four in the graph $G$ are of one of the three forms: $(i)$ $\{g, a_{i}, c_{i}, s_{j}\}$ for some element $u_i$ and set $S_j$ containing it; $(ii)$ $\{g, b_{i},c_{i},s_{j}\}$ for some element $u_i$ and set $S_j$ containing it; $(iii)$ $\{g, a_{i}, c_{i}, b_{i}\}$ for some element $u_i$.\\
\noindent \emph{Proof:} We define subsets $T, A, B, C$ of $V(G)$ as collections of vertices $s_j$'s, $a_i$'s, $b_i$'s and $c_i$'s respectively. Formally, $T = \{s_1, s_2, \cdots , s_m\}$; $A = \{a_1, a_2, \cdots, a_n\}$; $B = \{b_1, b_2, \cdots, b_n\}$ and $C = \{c_1, c_2, \cdots, c_n\}$. Note that $G[T]$ is a clique whereas $A, B, C$ are independent sets in $G$.

Since $G[T \cup \{g\}]$ is a clique, any induced cycle of length at least four contains at most two vertices of $T \cup \{g\}$.
As $G \setminus (T \cup \{g\})$ is a collection of induced paths on three vertices, and hence acyclic, the largest induced cycle possible in $G$ is of length five. We note that every induced paths is of the form $\{a_i, c_i, b_i\}$ and only other vertex adjacent to $a_i, b_i$ is $g$. Hence, such path can not be part of an induced $C_5$. 
 This implies that $G$ does not contain an induced cycle of length five or more.  


Assume that there exists an induced $C_4$ with two vertices, say $s_j, s_{j^{\prime}}$, in $T$. By construction, the only vertices which are adjacent with $s_j, s_{j^{\prime}}$ are in the set $C \cup \{g\}$. Since $g$ is adjacent with both these vertices, it can not be part of induced $C_4$ that contains $s_j, s_{j^{\prime}}$. This implies that the remaining two vertices in $C_4$ are from $C$. As $s_j, s_{j^{\prime}}$ are adjacent to each other, the remaining two vertices in $C_4$ must be adjacent to each other. This contradicts the fact that $C$ is an independent set in $G$. Hence our assumption is wrong and no such induced $C_4$ exists.

Consider an induced $C_4$ which contains exactly one vertex, say $s_j$, in $T$. Assume that this induced $C_4$ does not contain $g$. By construction, only neighbors of $s_j$ outside $T$ are in $C$. Let $c_i, c_j$ are two vertices contained in this induced $C_4$. Since the only common neighbor of $c_i, c_j$ outside $T$ is $g$, our assumption that this $C_4$ does not contain $g$ is wrong. This implies every such induced $C_4$ contains $g$. Since the only other neighbor of $s_j$ is in $C$, one of the remaining vertex in this cycle is from set $C$. 
Let that vertex be $c_i$.
Since $a_i$ or $b_i$ are the only common neighbors of $g$ and $c_i$, the only possible cycles are of the form $(i)$ or $(ii)$ mentioned in the claim. 
As there is no edge between following pair of vertices- $(g, c_i)$, $(s_j, a_i)$ and $(s_j, b_i)$, this cycle is indeed an induced $C_4$.

Consider an induced $C_4$ which does not intersect with $T$. Since $G \setminus (T \cup \{g\})$ is acyclic, every cycle of this type contains $g$. In this cycle, neighbors of $g$ are from sets $A$ and $B$. By construction, vertices $a_{i}$ and $b_{i^{\prime}}$ have a common neighbor only if $i = i^{\prime}$. This implies the only possible induced $C_4$ which does not intersect with $T$ is of the form $(iii)$.

As we have considered all cases exhaustively, this proves the claim. \hfill $\diamond$

\vspace{0.2cm}

We now prove that any solution of \textsc{Set Cover/$k$} naturally leads to a solution for \textsc{Chordal Contraction}. For any subset $\calF$ of $\calS$, let $F_{\calF}$ be the set of edges in $G$ which are incident on $g$ and $s_i$ for some $s_i$ in $\calF$. 

\vspace{0.2cm}
\noindent\emph{Claim 2:} If $\calF$ is a set cover of instance the $((U,\calS), k)$, then $G/F_{\calF}$ is a chordal graph.\\
\noindent\emph{Proof :} Let $H$ be the graph obtained from $G$ by contracting all edges in $F_{\calF}$. Since $\calF$ covers all the elements of $U$, contracting all the edges in $F_{\calF}$ introduces edge $gc_i$ for every $i$ in $\{1, 2, \cdots,n\}$ in graph $H$. By Claim~$1$, all the induced cycles in $G$ are of the form $\{g, a_{i}, c_{i}, s_{j}\}$ or $\{g, b_{i},c_{i},s_{j}\}$ or $\{g, a_{i}, c_{i}, b_{i}\}$ for some element $u_i$. Hence, there is no induced cycle of length four or more in $H$; thus it is a chordal graph. \hfill $\diamond$
\vspace{0.2cm}

Let $F$ be the given solution for $(G, k)$ such that $G/F$ is a chordal graph and $|F|$ is at most $k$.  
Let $\calF$ be the solution for \textsc{Set Cover/$k$} instance returned by the solution lifting algorithm. Recall that the modification operation $(a)$ or $(b)$ mentioned in the solution lifting algorithm does not change the size of $F$. 

\vspace{0.2cm}
\noindent\emph{Claim 3:} $\calF$ is a set cover of size at most $|F|$ for $(U, \calS)$.\\
\noindent \emph{Proof:} Consider the four-cycle given by $\{g, a_i, c_i, b_i\}$ in graph $G$. Since $G/F$ is a chordal graph, there exists an edge $gc_i$ in $E(G/F)$. This implies there is an edge $gs_j$ in $F$ for some set $S_j$ which contains $u_i$. Hence there is a set $S_j$ in $\calF$ which contains element $u_i$. Since this is true for any element in $U$, $\calF$ is a set cover for $(U, \calS)$. As $|W(g)|$ is at most $|F| + 1$ and it contains vertex $g$, upper bound on $\calF$ follows. \hfill $\diamond$
\vspace{0.2cm}

We are now in the position to conclude the proof. Claim~$2$ implies that $\OPT_{\textsc{ChC}(G, k)} \le \OPT_{\textsc{SC$/k$}}((U, \calS), k)$. Moreover, if $|F| \ge k + 1$, then solution lifting algorithm returns $\calF = \calS$. In this case, \textsc{ChC}$(G, k, F) = k + 1 = $ \textsc{SC$/k((U, \calS), k, \calS)$}. If $|F| \le k$ and $G/F$ is a chordal graph then by Claim~$3$, \textsc{SC$/k((U, \calS), k, \calF)$} $\le$ \textsc{ChC}$(G, k, F)$. This implies that there exists a \texttt{$1$-appt} from \textsc{Set Cover$/k$} to \textsc{Chordal Contraction} when parameterized by solution size.
\end{proof}

Karthik et al.~\cite[see conclusion]{DBLP:conf/stoc/SLM18} showed that assuming \FPT $\neq$ {\sf W[1]}, no \FPT\ time algorithm can approximate  {\sc Set Cover} within a factor of $F(k)$. Pipelining this result with 
Lemma~\ref{lemma:1-appt-chordal}, we get the following result.

\begin{reptheorem}{thm:chordal-no-lossy} Assuming \FPT $\neq$ {\sf W[1]}, no \FPT\ time algorithm can approximate  {\sc Chordal Contraction} within a 
 factor of $F(k)$. Here, $F(k)$ is a function depending on $k$ alone.
\end{reptheorem}
\section{Conclusion}
\label{section_conclusion}
In this paper, we studied the {\sc $\cal F$-Contraction} problem, where $\cal F$ is a subfamily of chordal graphs, in the realm of parameterized approximation.  We showed that \textsc{Clique Contraction} admits 
a {\sf PSAKS}. On the other hand, \textsc{Split Contraction} admits a factor $(2+\epsilon)$-\FPT-approximation algorithm, for any $\epsilon > 0$. In fact, we showed that for any $\epsilon > 0$, \textsc{Split Contraction} admits an $(2 + \epsilon)$-approximate kernel with $\calO(k^{f(\epsilon)})$ vertices.  We complemented this result by showing that, assuming {\sf Gap-ETH}, no \FPT\ time algorithm can approximate {\sc Split Contraction} within a 
factor of $\left(\frac{5}{4}-\delta \right)$, for any fixed constant $\delta>0$.  Finally, we showed that, assuming \FPT $\neq$ {\sf W[1]}, {\sc Chordal Contraction} does not admit any $F(k)$-\FPT-approximation algorithm.   
 Our results add to this growing list of collection of \FPT-approximable and \FPT-in-approximable problems.    
 
 We find it extremely interesting that  three closely related  problems have different behavior with respect to \FPT-approximation.  Our paper also shows that further classification of problems using lossy kernels are of interest and could shed new light on even well-studied problems.  
The paper naturally leads to the following question: can the gap between upper and lower bounds for \textsc{Split Contraction} brought closer? We conjecture that \textsc{Split Contraction} does admit $\frac{5}{4}$-\FPT-approximation algorithm.

\bibliographystyle{alpha}
\bibliography{references,references-saket}
\end{document}